\documentclass[12pt]{article}

\newif\ifaer
\aerfalse

\ifaer
\usepackage[letterpaper,hmargin=1.5in,vmargin=1.5in]{geometry}
\fi

\usepackage{graphicx}
\usepackage[colorinlistoftodos,prependcaption,textsize=tiny]{todonotes}

\usepackage{float,subcaption}
\floatstyle{boxed}
\restylefloat{figure}
\usepackage{comment,url,hyperref}
\usepackage{xfrac}
\usepackage{xcolor}

\usepackage{natbib}

\usepackage{tikz}
\usepackage{pgfplots}
\usetikzlibrary{patterns}
\usepgfplotslibrary{fillbetween}
\usetikzlibrary{intersections}

\usetikzlibrary{arrows, decorations.markings}

\tikzstyle{vecArrow} = [thick, decoration={markings,mark=at position
	   1 with {\arrow[semithick]{open triangle 60}}},
	      double distance=1.4pt, shorten >= 5.5pt,
		     preaction = {decorate},
			    postaction = {draw,line width=1.4pt, white,shorten >= 4.5pt}]
				\tikzstyle{innerWhite} = [semithick, white,line width=1.4pt, shorten >= 4.5pt]

\usepackage{amsmath,amsthm,amssymb}
\usepackage{amsmath,amssymb}
\usepackage{bbm}
\usepackage[toc,page]{appendix}

\usepackage{cleveref}

\theoremstyle{plain}
\newtheorem{theorem}{Theorem}[section]
\newtheorem{lemma}[theorem]{Lemma}
\newtheorem{claim}[theorem]{Claim}

\newtheorem{proposition}[theorem]{Proposition}
\newtheorem{corollary}[theorem]{Corollary}

	\newenvironment{numberedtheorem}[1]{%
		\renewcommand{\thetheorem}{#1}%
\begin{theorem}}{\end{theorem}\addtocounter{theorem}{-1}}

\theoremstyle{plain}
\newtheorem{definition}{Definition}[section] 
\newtheorem{example}[definition]{Example}

\newcommand{\R}{{\mathbb R}}
\newcommand{\X}{{\mathcal X}}

\newcommand{\E}{\mathop{\mathrm {E}}\displaylimits}

\newcommand{\abs}[1]{\left|#1\right|}
\newcommand{\eps}{\varepsilon}

\newcommand{\ncdf}[2][]{\Phi\ifthenelse{\not\equal{}{#1}}{_{#1}}{}\!\left({\def\givenn{\middle|}#2}\right)}
\newcommand{\npdf}[2][]{\phi\ifthenelse{\not\equal{}{#1}}{_{#1}}{}\!\left({\def\givenn{\middle|}#2}\right)}

\newcommand{\ridge}{\lambda}

\newcommand{\ridgebias}[1][\ridge]{\mu_{#1}}
\newcommand{\ridgevar}[1][\ridge]{\sigma^2_{#1}}

\newcommand{\N}{\mathcal{N}}

\DeclareMathOperator{\argmin}{argmin}

\newcommand{\given}{\,\mid\,}

\newcommand{\prob}[2][]{\text{\bf Pr}\ifthenelse{\not\equal{}{#1}}{_{#1}}{}\!\left[{\def\givenn{\middle|}#2}\right]}
\newcommand{\expect}[2][]{\text{\bf E}\ifthenelse{\not\equal{}{#1}}{_{#1}}{}\!\left[{\def\givenn{\middle|}#2}\right]}
\newcommand{\var}[2][]{\text{\bf Var}\ifthenelse{\not\equal{}{#1}}{_{#1}}{}\!\left[{\def\givenn{\middle|}#2}\right]}

\newcommand{\tparen}{\big}
\newcommand{\tprob}[2][]{\text{\bf Pr}\ifthenelse{\not\equal{}{#1}}{_{#1}}{}\tparen[{\def\given{\tparen|}#2}\tparen]}
\newcommand{\texpect}[2][]{\text{\bf E}\ifthenelse{\not\equal{}{#1}}{_{#1}}{}\tparen[{\def\given{\tparen|}#2}\tparen]}

\newcommand{\sprob}[2][]{\text{\bf Pr}\ifthenelse{\not\equal{}{#1}}{_{#1}}{}[#2]}
\newcommand{\sexpect}[2][]{\text{\bf E}\ifthenelse{\not\equal{}{#1}}{_{#1}}{}[#2]}

\newif\ifMS
\MSfalse
\newif\ifitcs
\itcsfalse
\newif\ifEC
\ECtrue
\newif\ifitcssubm
\itcssubmfalse
\begin{document}

\ifitcssubm
\title{Bias-Variance Games}
\else
\title{Bias-Variance Games\thanks{This work was initiated during the Special Quarter on Data Science and Online Markets held in the spring of 2018 at Northwestern University when the fifth author was supported as a McCormick Advisory Council Visiting Associate Professor.  The first, second, and fourth authors gratefully acknowledge the support of National Science Foundation award number 1718670.  The first, third, and fourth authors gratefully acknowledge the support of National Science Foundation award number 1618502.}
}
\fi

\ifaer
\begin{titlepage}
\fi
\author{Yiding Feng\thanks{Microsoft Research. Email: \texttt{yidingfeng@microsoft.com}.} 
\and Ronen Gradwohl\thanks{Department of Economics and Business Administration, Ariel University. Email: \texttt{roneng@ariel.ac.il}.} 
\and Jason Hartline\thanks{Department of Computer Science, Northwestern University. Email: \texttt{hartline@northwestern.edu}}  
\and Aleck Johnsen\thanks{Department of Computer Science, Northwestern University. Email: \texttt{aleckjohnsen@u.northwestern.edu}} 
\and Denis Nekipelov\thanks{Departments
 of Economics and Computer Science,
 University of Virginia. Email:
 \texttt{denis@virginia.edu}
 }}

\date{}							

\maketitle

\begin{abstract}
\ifaer
Firms increasingly rely on predictive analytics via machine learning algorithms to drive a wide array of managerial decisions. In this paper, we study the effect of competition on the choice of such algorithms, focusing on the tradeoffs between bias and variance in the algorithms' predictions. Absent competition, firms care only about the magnitude of predictive error and not its source. With competition, however, firms prefer to incur error caused by variance over error caused by bias, even at the cost of higher total error.

\else
Firms engaged in electronic commerce increasingly rely on predictive
analytics via machine-learning algorithms to drive a wide array of
managerial decisions. The tuning of many standard machine learning
algorithms can be understood as trading off bias (i.e.,
accuracy) with variance (i.e., precision) in the algorithm's
predictions.  The goal of this paper is to understand how competition
between firms affects their strategic choice of such algorithms. To
this end, we model the interaction of two firms choosing learning
algorithms as a game and analyze its equilibria.  Absent competition,
players care only about the magnitude of predictive error and not its
source. In contrast, our main result is that with competition, players
prefer to incur error due to variance rather than due to bias, even at the
cost of higher total error.  In addition, we show that competition can
have counterintuitive implications---for example, reducing the error
incurred by a firm's algorithm can be harmful to that firm---but
we provide conditions under which such phenomena do not occur. In addition to our theoretical analysis, we also
validate our insights by applying our metrics to several publicly available datasets.
\fi
\end{abstract}

\ifaer
    \thispagestyle{empty}

\end{titlepage}

    \renewcommand{\thefootnote}{\arabic{footnote}}
    \setcounter{footnote}{0}
\fi
\section{Introduction}

Firms that engage in electronic commerce increasingly rely on predictive analytics to drive a wide array of managerial decisions, ranging from product recommendations to customer targeting and pricing.
Given
data, perhaps from past consumer behavior, a firm facing a potential
customer will use predictive models or learning algorithms (henceforth
algorithms) to anticipate the customer's future behavior and
preferences, allowing the firm to better tailor its recommendation,
targeting, and pricing decisions. In general, the success of such
predictive analyses depends on the effectiveness of the algorithms
used. Research on such algorithms has proliferated, 
and their capabilities have advanced incredibly over the past
couple of decades.


The point of departure for this paper is the observation that, in many applications,
firms that utilize predictive analytics do so in a
competitive environment, and so the efficacy of a firm's analytics
depends not only on its own expertise and technology but also on that
of its competitors. In this paper we address the
question of how the competitive nature of the interaction affects a
firm's choice of algorithms. For example, while a particular algorithm
may be best for a monopolistic firm targeting a customer, it may be
suboptimal in a competitive environment. Furthermore, the optimal
choice of algorithm in the competitive environment may depend on the
competitors' choices of algorithms.

For a concrete example, consider the increasingly popular box
subscription companies that mail personalized monthly boxes of
fashion, food, or other products to subscribers. The appeal of these
companies lies in their high level of personalization, often achieved by
machine learning algorithms \citep{sinha2016data}. Stitch Fix, for
example, uses such algorithms to predict each subscriber's fashion
taste and sends a box matching this taste
\citep{gaudin2016stitch}. The better the fit, the more satisfied the
subscriber, leading to greater customer acquisition and retention. Of
course, Stitch Fix competes with other companies, such as Trunk Club,
that also personalize their boxes using learning
algorithms. Ultimately, the profitability of such a company will
depend not only on how well it manages to predict a customer's fashion
taste, but also on the predictions of its competitors.

For a more general example, consider the so-called ``Long Tail''
marketplaces, which are characterized by huge numbers of goods that
individually have low demand but that collectively make up substantial
market share. One of the key drivers of Long Tail markets is the
ability of firms to connect supply and demand, typically through machine
learning algorithms that predict consumers' tastes and match them to
products \citep{anderson2006long}. Often, many firms compete in the
same Long Tail market, and in this case their success depends not only
on their own ability to match goods to consumers, but also on the
predictive ability of their competitors.

One useful way of analyzing algorithms' predictive ability is by
examining the different sources of error they incur. There are two
general types of errors: a lack of accuracy---called bias---in which
the predictions are not, on average, equal to the true value; and a lack
of precision---called variance---in which the predictions are not
clustered tightly around their average. The total error of an
algorithm can be decomposed into these two kinds of error.

In
practice, there are various ways to control the bias and variance of
an algorithm. For example, one could allow the algorithm to consider
more complex functions to map data onto predictions, such as deeper
decision trees or regressions with higher degree functions, which
result in lower bias but higher variance. Alternatively, a technique
called regularization---intuitively, penalizing predictions that are
less smooth---is often used to decrease variance at the expense of
higher bias. Finally, increasing the amount of training data decreases
variance.  Algorithms that predict well are ones that control the tradeoff between
bias and variance so as to minimize the
total error, regardless of its source.

In this paper we aim to understand how competition affects the
optimal way to trade off bias and variance.  We show
that, holding total error fixed, absent competition there is no
preference for variance versus bias. In contrast, in competitive environments
it is better to reduce bias at the expense of variance, even when this leads to higher total error.  
This result
holds up under several natural theoretical models of predictive error
and in an empirical study.  Consequently, training an algorithm in
isolation to minimize error does not lead to optimal parameter
settings for algorithms in competitive environments.  An implication
of these results is that, in competitive environments, there is an
added benefit to algorithms that consider more complex functions and
an added cost to regularization.

\paragraph{Overview of model and results.}

In this paper we model the interaction of two firms as a game, and
analyze the game's equilibria.
Players' actions are learning algorithms, functions that map {\em feature vectors} to predicted {\em labels}. Players' payoffs depend both on the error of their chosen algorithm's prediction---specifically, the squared distance between the true label and the algorithm's predicted label---and on whether or not their algorithm's
prediction is better than that of their opponent. We view the game as consisting of three stages: First, in the {\em ex ante} stage, each player chooses an algorithm. Second, in the {\em interim} stage, a feature vector is realized, and each player's chosen algorithm yields a distribution over predictions for this particular feature vector. Finally, in the {\em ex post} stage, predictions and payoffs are realized. Although players act only in the ex ante stage, most of our analysis focuses on understanding players' preferences in the interim stage.

For the analysis of the interim stage, we suppose that the feature vector is fixed. We abstract away from the details of specific
algorithms, and instead model an algorithm as a probability
distribution over prediction errors. Thus, players' actions correspond
 to probability distributions, and players' action  spaces---the sets of
possible actions they can choose---correspond to families of
probability distributions that range over biases and variances.  A
canonical example is the set of normal distributions with different
means and standard deviations.  

Our main theoretical result considers a two-player game where each
player's action space consists of normal distributions with the same
total mean squared error---namely, they have the same total error but different biases and variances. 
Absent competition, a player would be indifferent
among all these distributions. In contrast, we prove that, in the competitive scenario, each
player would prefer error distributions with lower bias (and therefore higher
variance), and that this holds regardless of the actual prediction made by the
opponent.  In
game theoretic terms, minimal-bias is an {\em ex post dominant strategy}.
This strong result persists in games with more than two players.
It also implies that the unique  {\em Nash equilibrium} of the game is the one in which each player chooses the distribution with minimal bias. 

We then extend our analysis, and show that the preference for lower bias persists when the total error is not fixed. In particular, we consider a case where, absent competition, players are not indifferent between the various available combinations of bias and variance, but rather where there is some most-preferred distribution with nonzero bias that has minimal total error. Here we show that, under competition, players strictly prefer a distribution with lower bias than this most-preferred distribution, even though it has higher total error.

We supplement these theoretical results with numerical analyses that
demonstrate the robustness of the theoretical findings. First, we numerically test
the robustness of our insight on the strategic preference for reduced bias with non-normal families of distributions, such
as Laplace, logistic, and uniform.  Our insight persists for many of
the variations, but, notably, it fails for uniform distributions.
Second, we
investigate the dependence of the results on the form of players'
payoff functions.  In particular, we study variations in the benefit from winning
relative to the cost of prediction errors.  We find by numerical calculation
that the {\em ex post} preference for lower bias fails to extend. However,
we also find that minimal-bias remains a dominant strategy---that is,
players prefer lower bias (and higher variance) for every choice of
probability distribution by the opponent (although not for every
realization of this distribution). 

After establishing players' preferences for lower bias in the interim stage, we next turn to the analysis of the ex ante stage. In this stage each player chooses an algorithm that, for each possible feature vector, will yield a distribution over total error with some bias and some variance. This general setting is substantially more complicated, as each potential algorithm may imply a different total error, bias, and variance for each feature vector. Nonetheless, we show that, under some assumptions, our insight on the preference for lower bias persists. We also theoretically, numerically, and empirically verify the assumptions necessary for this result, focusing on the particular learning algorithm of ridge regression---a variant of linear regression that allows for flexibility to control the bias-variance tradeoff using a
{\em regularization parameter}.

Finally, we conduct an empirical study of our bias-variance game for a
family of learning algorithms on benchmark datasets.  In this study,
players utilize a particular learning algorithm to make predictions, given a particular
dataset.\footnote{We use the California housing prices
  data from the 1990 Census, a data set first utilized by
  \citet{pace1997sparse} and included in the Python Scikit-learn
  library. We also use data on wine quality,
designed and utilized by \citet{cortez2009modeling}.} Specifically, the players use a ridge regression algorithm.  When there is only one player we show that the optimal choice of
regularization parameter---the parameter that controls the bias-variance tradeoff---is large, but when there are two players, 
payoffs increase as the parameter is
lowered. In other words, in the latter scenario there is a preference for lower bias and higher
variance.  Thus, the algorithmic optimizations of the non-competitive and
competitive settings are qualitatively distinct and result in quite
different preferences with respect to the tradeoff between bias and variance.

To give some context for the theoretical results, we provide a few
additional observations about algorithms in competitive situations.
Counterintuitively, we show that there are families of distributions and opponent
choices for which a player prefers a distribution with higher bias
(respectively, higher variance) even while holding variance (respectively,
bias) fixed.  Nonetheless, we also show that higher bias (with variance
fixed) is not beneficial for natural families of
distributions, such as normal and Laplace.  Moreover, for normal distributions, our main
theoretical result (described above) strengthens this conclusion on the harmfulness of higher
bias by
showing that decreasing bias is beneficial even at the expense of
increased variance (holding the total error fixed).  The above
counterintuitive observation---the possibility that increasing bias
can be beneficial---highlights the obstacles that our main theoretical
analysis must overcome.

\paragraph{Related literature.}
The analysis of strategic interactions that involve machine learning algorithms is a newly burgeoning area of 
study in both economics and computer science. 
For example, \citet{eliaz2019model} study the interaction of a rational agent and a learning algorithm, and consider
the question of whether the agent has an incentive to truthfully report her information to the algorithm.
\citet{liang2019games} and \citet{olea2019competing} study scenarios in which there are multiple algorithms that compete with one another.
\citet{liang2019games} considers games of incomplete information in which the players have data and use algorithms
to derive their beliefs. \citet{olea2019competing} study a game between agents competing to predict a common variable, and
where agents obtain the same data but differ in the algorithms they utilize for prediction.
In all these papers, the algorithms under consideration are fixed exogenously. Our paper, in contrast,
focuses on the strategic choice of algorithms in competitive environments.

On the computer science side,
our study 
is
related to the ``dueling algorithms'' framework of
\citet{immorlica2011dueling}.  Within this framework, \citet{BT-19},
building on \citet{ben2017best}, study the problem of 
multiple learners selecting
\ifaer algorithms to make predictions within a particular dataset.
\else
a hypothesis, i.e., a function mapping the features to a prediction,
from the same hypothesis class on the same data set.  
\fi
They work within the PAC-learning framework of \citet{valiant1984theory}, and consider
equilibria in the game where, for each point in the dataset, a
payoff of one is split evenly between all players whose predictions
are within a given error tolerance. 
 A key point of difference between
this setup and ours is that they consider competition between
specific algorithms, such as linear regressors, and study
the questions of whether equilibria exist and can be learned.
On the other hand, we study the general tradeoff between bias 
and variance in the equilibrium choice of algorithms.


\paragraph{Organization.}

\Cref{sec:model-and-preliminaries} introduces the general framework and describes the ex ante and interim stages, the one- and two-player games
that we study, and some basic properties.
\Cref{sec:no-tradeoff} considers the one-player intuition that
reducing bias or reducing variance is always beneficial, when everything
else is held fixed, in the two-player game.  \ifitcssubm(In this ITCS submission, \Cref{sec:no-tradeoff} is deferred to \Cref{appx:no-tradeoff}.) \fi  Counterintuitively, there
are two-player scenarios where a player would want to increase bias or
variance.  On the other hand, for normally-distributed errors,
reducing bias to zero is beneficial.  \Cref{sec:tradeoff} contains our
main analysis of the two-player interim game.  For normally-distributed error,
we show that there is an {\em ex post} preference for lowering bias at the
expense of variance, we argue that the natural ridge regression
algorithm indeed has normally-distributed error, and we present
simulation results for other distributions of error and a variety of utility functions.
In \Cref{sec:non-constant-tradeoff-general-framework} we then turn to the ex ante game, where we show that the insight on the preference for variance over bias persists.
In \Cref{sec:empirical} we consider an empirical version of the game played on a standard benchmark data
sets with ridge regression, and show that qualitative conclusions
of our theoretical analysis continue to hold. Finally, \Cref{sec:conclusions} concludes with some
discussion. 



\section{Model and Preliminaries}\label{sec:model-and-preliminaries}
We begin in Section~\ref{sec:general-framework} with a brief overview of the general framework for statistical decision theory, including a formalization of the ex ante and interim stages. This is followed by a description of the model for the interim stage game in Section~\ref{sec:prelim}. Since most of our analysis is about this stage, we defer the formalization of the ex ante stage game to Section~\ref{sec:non-constant-tradeoff-general-framework}.

\subsection{The General Framework}
\label{sec:general-framework}
We begin by sketching a general framework of statistical decision theory, which is our point of departure.
For more details and references see \citet{friedman2001elements}.

There is a prior distribution $\pi$ over pairs $(x,f(x))$, where each $x\in\R^p$ is a feature vector and each $f(x)\in\R$ is a label. 
A typical example is that $f$ is linear plus unbiased noise, namely
$$
f(x^i)=w^Tx^i+e_i,
$$
where the $e_i$ are independent
random variables with $E[e_i]=0$ and $E[e_i^2]=\sigma^2_e$, and where $w\in\R^p$ is a $p$-dimensional vectors of fixed but unknown parameters.

A learning algorithm takes training data $D=\{(x^1,y^1),\ldots,(x^m,y^m)\}$ as input, and produces an estimator $\hat f$ as output. Given a new point $x$, the estimator $\hat f$ predicts that the corresponding value of $f(x)$ is $\hat f(x;D)$.

Define the {\em loss} of a learning algorithm at feature vector $x$ as
$$L(\hat f, x) = \left(\hat f(x;D) - f(x)\right)^2.$$
Note that the loss is random---it depends on both the randomness in $D$ and the randomness inherent in $f$.

The {\em risk} function of a learning algorithm at $x$ is 
$$R(\hat f, x) = \E[L(\hat f, x)],$$
where the expectation is taken over both
the randomness in $D$ that produced $\hat f$ and the randomness in $f$ (but note that $x$ is fixed).

A well-known result is that when the loss is defined as above (squared-loss), then the risk can be decomposed into bias and variance:
$$R(\hat f, x) = \left(\E_D[\hat f(x;D)]-\E[f(x)]\right)^2 + \mathrm{Var}[\hat f(x;D)] + \mathrm{Var}[f(x)],$$
where $\E[\hat f(x;D)]-\E[f(x)]=\mathrm{Bias}_D[\hat f(x;D)]$ 
is the {\em bias} of the estimator at $x$, the second term is the {\em variance} of the estimator (computed with respect to the randomness of $D$) at $x$, and the third term is the {\em irreducible error}. 

Ideally, we would like our learning algorithm to have minimal risk for all $x$'s. However, this is generally not possible. Instead, one way of identifying a good algorithm is to find one that minimizes  the {\em Bayes risk} $$R_B(\hat f) = E_{x}[R(\hat f, x)],$$
which is the expected value of the risk with respect to the distribution over feature vectors $x$.

Now, in practice we may not know the distribution over $x$'s, and in this case a common approach is to consider the {\em empirical risk}
$$R_{\mathrm{emp}}(\hat f) = \frac{1}{2m}\sum_{i=1}^m  \left(\hat f(x^i;D) - y^i\right)^2$$
of an estimator. Minimizing the empirical risk is not helpful, since one could always choose $\hat f$ for which $\hat f(x;D) = y$ when $(x,y)\in D$, and $\hat f(x;D) = 0$ otherwise. Common ways to get around this problem are (i) to limit the class of functions from which $\hat f$ is chosen, and (ii) to find an estimator that minimizes the {\em regularized} empirical risk
$$R_{\mathrm{emp}}(\hat f) + \lambda J(\hat f).$$
Here, $\lambda\in \R_+$ is the {\em regularization parameter} and $J$ is some penalty function. For example, in ridge regression, which we will use as a running example throughout the paper, the goal is to find a linear estimator $\hat f(x;D) = \hat w^T x$ for which $R_{\mathrm{emp}}(\hat f) + \ridge\|\hat w\|^2_2$ is minimal.

In regularized empirical risk minimization, there are various ways to choose the regularization parameter $\lambda$. The goal is to choose $\lambda$ so that the Bayes risk of the resulting estimator is minimal. Typically, modifying $\lambda$ affects the bias and variance of an estimator. For example, in ridge regression, increasing $\lambda$ decreases the  variance and increases the bias of the chosen estimator. At the optimal $\lambda$, the tradeoff between the bias and  variance of the chosen estimator is such that its Bayes risk is minimal.\footnote{In practice, finding such a $\lambda$ is often done using cross-validation \citep[see, e.g.,][]{friedman2001elements}.}


The standard framework can be viewed as consisting of three stages: (i) in the ex ante stage, the estimator $\hat f$ is chosen; (ii) in the interim stage, the feature vector $x$ and its label $f(x)$ are realized (although the latter remains unknown); and (iii) in the ex post stage, the estimate $\hat f(x;D)$ is produced.

\subsection{The Interim Game}
\label{sec:prelim}
In most of our analysis we will examine the interim stage, and so we will assume that $x$ and $f(x)$ are fixed (and the latter unknown). 
Any choice of algorithm $\hat f$ implies a distribution $\hat f(x;D)$ over predictions about the value of $f(x)$, and so a distribution over the error in prediction, $|f(x) - \hat f(x;D)|$. In our analysis here we will abstract
away from the particular algorithm, and will only be concerned with the
algorithm's error in prediction. 

We will focus on a standard measure of error, namely, 
mean squared error. Thus, if the predictive error of an algorithm is a number $|a|$, then the squared
error is $a^2$. The best algorithm is one with predictive error of~0. 


This paper is concerned with players' choices of algorithms. In our abstraction for the interim game we will thus let each player choose a distribution over errors, each representing the distribution over the error in predicting $f(x)$ by some algorithm. (Looking ahead, when we describe the ex ante game in \Cref{sec:non-constant-tradeoff-general-framework}, players will actually be choosing regularization parameters of certain algorithms, which indirectly imply various distributions over error.) We will suppose that players choose from a class of distributions, as follows.  Fix a random variable $Z$ with mean 0 and standard deviation 1. A common choice for $Z$ will be normal, a choice we motivate in \Cref{sec:normal-motivation}, but we can
also consider uniform, triangle, Laplace, and other distributions. Players will then choose a distribution from a class
made up of shifts and spreads of $Z$. A typical example will be the class 
$$\X_Z=\left\{\sigma Z+\mu:\mu^2+\sigma^2\geq 1\right\},$$
the set of distributions whose squared-bias plus variance is at least 1 (which means their total squared error
is at least 1). Each element of this class yields a different distribution over predictive error, and
represents a different algorithm whose error has the corresponding distribution.

Some discussion of this modeling assumption is warranted. The class of distributions $\X_Z$ represents the error distributions
of {\em all} algorithms available to a player. In particular, this precludes the possibility that a player chooses an algorithm, 
observes the realized prediction, and then uses some alteration of that prediction. For a concrete example of what this implies, consider the distribution $X$ with bias $\mu=1$ and variance $\sigma^2=0$, and observe that $X\in\X_Z$ above. Then our assumption precludes the possibility that the player takes the realization of $X$ and subtracts 1 from it, yielding a perfect estimator.

To understand why this assumption is reasonable in our context, recall that the choice of algorithms is actually made in the ex ante stage, before a specific $x$ is realized. In the context of regularized empirical risk minimization, the optimal estimator $\hat f$ minimizes the Bayes risk, but this does not imply that it minimizes the risk for every realized $x$. Making modifications to $\hat f$, such as subtracting 1 from it, may lower the risk for certain $x$'s, but in expectation will be harmful since $\hat f$ is already optimal.

In our model, the algorithm chosen in the ex ante stage will yield different error distributions for different $x$'s. And while for some $x$'s it may be beneficial to take the realization of the algorithm and subtract 1 from it, for other $x$'s this same modification will be harmful.
Furthermore, as will be formalized in Section~\ref{sec:non-constant-tradeoff-general-framework}, eventually the algorithm will be chosen so as to maximize utility, in expectation over all $x$'s. And if subtracting 1 from a particular algorithm's realization is beneficial, then this algorithm does not maximize utility.

\subsubsection{One player}\label{sec:one-player} As a benchmark, consider a setting in which there is only one player, and suppose that she chooses an error distribution $X$.
On realization $a$ (that is, a prediction with predictive error $a$), let the player's utility be $u(a) = 1-a^2$: a benefit of 1 minus her squared error.
This utility function captures the idea that the player obtains positive utility from making a perfect prediction (the benefit of 1),
but that this utility decreases with the squared error (the loss of $a^2$). In the box subscription company example from the
introduction, the benefit represents future profits from a particular customer, whereas the loss accounts for the
lack of customer retention in case of inaccurate taste predictions. Note that the total utility $1-a^2$ could be negative;
in the example this would represent a firm's failure to recoup the costs of an initial loss leader or promotion.
We focus on this utility function for our theoretical analysis, but in \Cref{sec:numerical} 
we argue that our results are robust to other specifications of the utility
function, and in particular to ones where the chance of a negative payoff is negligible.

Given the utility function $u(a) = 1-a^2$, a player's expected utility from $X$ is 
$u(X)=\expect{u(X)} = 1 - \expect{X^2}$. The following is a simple observation:

\begin{claim}\label{claim:util}
If $X=\sigma Z+\mu$, where $Z$ is a random variable with mean 0 and variance 1, then $u(X) = 1-\mu^2-\sigma^2$.
\end{claim}

\ifMS
\proof{Proof.}
\else\begin{proof}
\fi
$X$ is a random variable with variance $\sigma^2$ and expected value $\mu$. Since $\sigma^2=\var{X} = \expect{X^2}-\expect{X}^2 =\expect{X^2}-\mu^2$,
it follows that $\expect{X^2} = \sigma^2+\mu^2$. Thus, $u(X) = 1-\expect{X^2} = 1-\mu^2-\sigma^2$.
\ifMS
\Halmos
\endproof
\else \end{proof}
\fi


This observation implies three corollaries, formalized below. First, for fixed bias, the
player prefers minimal variance. Similarly, for fixed variance, the
player prefers minimal bias.  Finally, the player is indifferent
between distributions that have the same bias squared plus variance. Although these corollaries are 
straightforward, in \ifitcssubm\Cref{appx:no-tradeoff} \else Section~\ref{sec:all-else-fixed} \fi we will show that, without further assumptions, none of them 
hold in the competitive setting.

\begin{corollary}
For any $\mu$, if $\sigma<\sigma'$, then the player prefers $X=(\sigma Z+\mu)$ to $X'=(\sigma' Z+\mu)$.
\end{corollary}

\begin{corollary}
For any $\sigma$, if $\mu^2<\nu^2$, then the player prefers $X=(\sigma Z+\mu)$ to $X'=(\sigma Z+\nu)$.
\end{corollary}

\begin{corollary}
If $X=\sigma Z +\mu$, $X'=\tau Z + \nu$, and 
$\mu^2+\sigma^2 = \nu^2+\tau^2$, 
then the player is indifferent between $X$ and $X'$.
\end{corollary}

\paragraph{Individual rationality.}
An additional definition that will be useful is that of individual
rationality. Intuitively, a distribution is individually
rational if a player derives non-negative utility from
choosing it. Formally:
\begin{definition}
A distribution $X$ satisfies {\em individual rationality (IR)} if
$u(X)\geq 0$.
\end{definition}
A simple observation that follows from \Cref{claim:util}
is that $X=\sigma Z+\mu$ is IR if and only if $\mu^2+\sigma^2\leq 1$.

\subsubsection{Two players}
Recall that our goal is to analyze the effect of competition on firms' choices of learning algorithms.
As described in the introduction, we will model this as a game, and we call this game the {\em Bias-Variance Game}.
In this game, each of two players simultaneously chooses an error distribution. Prediction errors are realized, and the
player with lower prediction error, say $a_i$, obtains utility $1-a_i^2$, just like the one-player case. The player
with higher prediction error obtains utility 0.

This specification of utilities, as well as variants that we discuss in \Cref{sec:numerical}, capture the
main competitive force we wish to analyze: the desire of a player both to minimize error (this is the $-a_i^2$ term), and to obtain lower
error than her competitor (this is the benefit of 1 from winning).

In our notation, when discussing player $i\in\{1,2\}$, we will denote by $j=3-i$ the identity of the other player. 
A formal description of the game follows:
\begin{definition}[The Bias-Variance Game]\label{BVG} Given two classes of distributions, $\X_1$ and $\X_2$, 
the two player bias-variance game proceeds as follows:
\begin{enumerate}
\item Each player $k \in \{1,2\}$ simultaneously chooses a distribution $X_k$ from $\X_k$.
\item Each $X_k$ is realized as some $a_k$.
\item Each player $i$ obtains utility $u_i(a_i, a_j) = (1 - a_i^2)\cdot {\bf 1}\{a_i < a_j\}$. That is,
the player $i \in \argmin_k a_k$ wins and obtains utility $u_{i}(a_{i},a_{j})=1-a_{i}^2$; the other player, $j$, loses and obtains utility  $u_{j}(a_{j},a_{i})=0$. 
\end{enumerate}
\end{definition}

In our theoretical analysis we will primarily consider bias-variance games where $\X_1 = \X_2 =
\X$ is a family of distributions in which the error of every $X \in \X$ is normalized
to $\mu^2 + \sigma^2 = 1$. In our numerical analysis in \Cref{sec:numerical} we relax this restriction.

\paragraph{Individual rationality.}
As in the one-player benchmark, we will be interested in player choices that are individually rational. A straightforward result is that
individual rationality of the one-player setting implies individual
rationality of the two-player setting, as the following proposition demonstrates:

\begin{proposition}
Bias-variance games with distributions $X \in \X$ satisfying $\mu^2 +
\sigma^2 \leq 1$ are individually rational: namely, expected payoffs in
the game are non-negative.
\end{proposition}

\ifMS \proof{Proof.}
\else\begin{proof}\fi
In the one-player setting player $i$'s expected payoff with realization $a_i
\sim X_i$ is $\expect{1-a_i^2} \geq 0$ which satisfies individual rationality.  Consider the two-player
game where the other player $j = 3-i$ has realization
$a_j$.  When the other player's realization is $|a_{j}| \leq 1$, the payoff
of player $i$ is non-negative for all $a_i$: for $|a_i| \leq |a_{j}|
\leq 1$ player $i$ wins and has non-negative payoff $1-a_i^2$ and for
$|a_i| > |a_{j}|$ then player $i$ loses and has payoff 0.  When the
other player's realization is $|a_{j}| > 1$, then player $i$'s
distribution of payoffs in the two player game dominates his
distribution of payoffs in the one-player setting (when player $i$ loses, instead of a negative payoff her payoff is zero).  As the latter
setting had non-negative expectation, so does the former.
\ifMS\Halmos\endproof
\else \end{proof}\fi

\paragraph{Solution concepts.} \ifaer We will consider various solution concepts.
\else
When there is a single player, the concept of optimality is
straightforward. When there are more players, however, different actions may be
better depending on the actions of other players.   ``Solutions'' to
games thus consist of finding strategies of players that satisfy
different notions of equilibrium. \fi For our main theoretical result on
the preference of variance over bias we will utilize a very strong
notion, namely that of ex post dominant strategies.  A strategy is ex
post dominant for player $i$ if it yields that player the highest
utility regardless of the {\em realized} prediction of the opponent.
\begin{definition} A strategy $X_i\in\X_i$ is {\em ex post dominant} for player $i$ if for all $X_i'\in\X_i$ and all realizations
$a_j$ of player $j$ it holds that $u_i(X_i,a_j)\geq u_i(X_i',a_j)$.
\end{definition}

For our numerical and empirical results we will consider two weaker
notions. The first, dominant strategies, requires that a strategy be
optimal against any strategy of the opponent, but not necessarily
against any realization of that strategy:
\begin{definition}
  A strategy $X_i\in\X_i$ is {\em dominant} for player $i$ if for all
  $X_i'\in\X_i$ and all $X_j\in\X_j$ it holds that $u_i(X_i,X_j)\geq
  u_i(X_i',X_j)$. If the inequality is strict for all $X_i'\neq X_i$, then $X_i$ is {\em strictly dominant}.
\end{definition}

The second, pure Nash equilibrium, does not require that a strategy be
optimal against any strategy of the opponent, but only against that
player's own Nash equilibrium strategy:
\begin{definition} A strategy profile $(X_1,X_2)$ is a {\em pure Nash equilibrium}
if for each player $i$ and strategy $X_i'\in\X_i$ it holds that $u_i(X_i,X_j)\geq u_i(X_i',X_j)$.
\end{definition}

Observe that if $X_i$ is ex post dominant, then it is also dominant. Furthermore, if $X_1$ and $X_2$ are dominant for players $1$ and $2$, respectively, then $(X_1,X_2)$ is a pure Nash equilibrium. Finally, if $X_1$ and $X_2$ are strictly dominant then $(X_1,X_2)$ is the {\em unique} Nash equilibrium.


\ifitcssubm
\section{Reducing Bias or Variance, All Else Fixed}
\label{sec:no-tradeoff}
For this version of the paper submitted to ITCS for review, this section is deferred to \Cref{appx:no-tradeoff} in order to bring forward the deeper results of independent \Cref{sec:tradeoff}.
\else
\section{Reducing Bias or Variance, All Else Fixed}
\label{sec:no-tradeoff}
\label{sec:all-else-fixed}
We begin our analysis by describing some counterintuitive implications of competition. In particular, we show that the simple corollaries from
the one-player benchmark in Section~\ref{sec:one-player} no longer hold, and that reducing bias (resp., variance)
holding variance (resp., bias) fixed can be harmful.
\begin{example}[Reducing variance can be harmful; see \Cref{fig:reduce-variance}]\label{ex:reduce-variance}
Suppose player 2 plays the distribution $\N(0,\eps)$, where $\eps$ is
some small number, and player 1 plays the distribution $\N(1/2,
1/2)$. Player 1's strategy is monotone and satisfies IR, and she
obtains positive expected utility: Given that she wins, she is likely
within $\eps$ of 0, and she wins with positive probability. However, if
player 1 decreases her variance to 0, she will obtain utility
close to 0, since she will hardly ever win (for small enough $\eps$).
\end{example}

\begin{example}[Reducing bias can be harmful; see \Cref{fig:reduce-bias}]\label{ex:reduce-bias}
Player 2 plays the uniform distribution on the interval $[-1-\eps, 1+\eps]$. 
For small enough $\eps>0$ this satisfies IR.
Player 1 plays the uniform distribution on the interval $[-1,1+2\eps]$.
Again, for small enough $\eps>0$ this satisfies IR.
Now consider a deviation by Player 1 to the interval $[-1-\eps, 1+\eps]$, a deviation that reduces bias.
This is harmful: Before the deviation, she never won when her realization
was in $(1+\eps,1+2\eps]$. After the deviation, however, the only
difference is the additional 
possibility of winning when her realization is in $[-1-\eps,-1)$.
But such victories are harmful, as they consist only
of negative utilities.
\end{example}

\begin{figure}
\begin{subfigure}{.5\textwidth}
\begin{center}
\input{fig-example-reduce-variance}
  \caption{\Cref{ex:reduce-variance}}
\label{fig:reduce-variance}
\end{center}
\end{subfigure}
\begin{subfigure}{.5\textwidth}
\begin{center}
\input{fig-example-reduce-bias}
  \caption{\Cref{ex:reduce-bias}}
\label{fig:reduce-bias}
\end{center}
\end{subfigure}
\label{fig:reduce variance or bias}
\caption{Illustrations for \Cref{ex:reduce-variance} 
and \Cref{ex:reduce-bias}.
Thin black curves are probability density functions of Player 2's distributions.
Thick blue (dashed) curves are probability density functions 
of Player 1's distributions (after reducing variance or bias).
}
\end{figure}

Unlike \Cref{ex:reduce-variance}, \Cref{ex:reduce-bias} is somewhat
unnatural for our application to machine learning algorithms, in that
the error distributions used are uniform. In
\Cref{sec:normal-motivation} we argue that we should expect the
distribution of the predictions of learning algorithms to be
normal. Is there an example in which reducing bias is harmful, but
where the class of distributions is more natural? The following two
theorems state that there is not.  The first considers a class of
distributions that are single-peaked, monotone, and with convex
tails,\footnote{A distribution is defined to have convex tails if its
  probability density function is convex and decreasing away from the
  mean.} and the second considers normal distributions.

\begin{theorem}\label{thm:lower-bias}
Let $Z$ be monotonically increasing, convex on $[-\infty,0]$, and symmetric
around 0. Let $X_i=\sigma Z+\mu$ be IR (so as to satisfy
$\mu^2+\sigma^2\leq 1$) and $X_i'=\sigma Z$. Then $u_i(X'_i,c)\geq
u_i(X_i,c)$ for any realization $c$ of player $j$.
\end{theorem}
The proof of this theorem and of \Cref{thm:lower-bias normal}, below, are
given in 
\ifaer
Appendix~\ref{app:no-tradeoff}.
\else
\Cref{app:no-tradeoff}.
\fi

Note that the assumption that $X_i$ is IR is necessary. To see this, consider an $X_i$ that is not IR, for example one in which $\mu=1000$ and $\sigma=10$. Suppose also that the opponent's realization is $c=500$. Observe that, in this case, $u_i(X_i,500)$ is close to 0, since the probability that $i$ wins is small. However, decreasing $\mu$ to 0 leads to 
an $X_i'$ for which $u_i(X_i',500)<0$: player $i$ nearly always wins, in which case her utility will be close to $1-\mu^2-\sigma^2=-99$.

The assumption that $Z$ is 
single-peaked (which follows from monotonicity
and symmetry) is also necessary. To see this, consider a $Z$
with two peaks, for example some infinitesimal perturbation of a Bernoulli
distribution. The peaks are at $1$ and $-1$ so that $\mu=0$.
Now, consider the case in which $X_i=Z+1$, and suppose $c=1$.
Then player $i$ wins whenever the realization is
close to the lower peak, which happens with probability $1/2$, and obtains utility
close to 1 conditional on winning. Hence, $u_i(X_i,c)=1/2$. However,
under $X_i'=Z$ the utility conditional on winning is close to 0,
since the realized values will be near the peaks $1$
and $-1$.

One drawback of Theorem~\ref{thm:lower-bias} is that it does not apply to normal distributions, since such distributions are not convex on $[-\infty,0]$.
However, as we are particularly interested in normal distributions, we have a version of the theorem that applies specifically to them:

\begin{theorem}\label{thm:lower-bias normal}
Let $Z$ be normal; let $X_i=\sigma Z+\mu$ and satisfy $\mu+\sigma\leq 1$; and let $X_i'=\sigma Z$. Then $u_i(X'_i,c)\geq u_i(X_i,c)$ for any realization $c$ of player $j$.
\end{theorem}

\fi
\section{Tradeoff Between Bias
  and Variance}
\label{sec:tradeoff}

Thus far, we have argued that simple observations that hold in the
one-player case fail to extend to the competitive setting, and that
reducing bias or variance may actually be harmful (holding all else
fixed). We now turn to our main analysis, which considers the {\em
  tradeoff} between bias and variance: if there is only one player, she
is indifferent between the two sources of error as long as their sum
(or rather, the sum of bias squared and variance) is fixed. In this
section, we show that in a competitive setting players are no longer
indifferent, and, furthermore, that there is a clear preference for
error incurred by variance over error incurred by bias.

We proceed as follows. In \Cref{sec:main-analysis} we state our main
result about the tradeoff between bias and variance in a
two-player bias-variance game.  We assume that $\mu^2+\sigma^2 = 1$,
which ensures the individual rationality constraint is satisfied.\footnote{In fact, this assumption implies that
the IR constraint binds in
the one-player setting---that is, $u(X) = 1 - \mu^2 -
\sigma^2 = 0$ for all $X$. This assumption is necessary for our analysis, but in \Cref{sec:numerical} we show numerically
that it is not necessary for our main result to hold.} We fix the random variable $Z$ to be
normal, and show that for arbitrary realizations $a_j$ of the opponent
player $j$, reducing bias $\mu_i$ while increasing variance to
$\sigma_i^2 = 1 - \mu_i^2$ is always strictly beneficial for player $i$. That
is, the minimal-bias strategy is ex post dominant. This result implies that the profile in which both players choose their minimal-bias strategy is the unique Nash equilibrium of the game. We also argue that our result extends to more than two players and to asymmetric strategy classes. 

Next, in \Cref{sec:non-constant-tradeoff}, we show that our insight on the preference for lower bias persists also when the total error is not constant across available
error distributions. In particular, we suppose that $\mu^2+\sigma^2=g(\mu^2)$ for some convex function $g$, and so there is some optimal distribution with minimal error. We show that,
in the competitive setting, on the margin it is beneficial for players to choose a distribution with lower bias than this optimal distribution.

We then turn to the various assumptions in the analysis---namely,
the normality assumption and the form of the utility functions. In \Cref{sec:normal-motivation}, we motivate
the focus on normal distributions by arguing formally that they are a
natural distribution for error in machine learning contexts. Then, in
\Cref{sec:numerical}, we consider distributions other than normal, as
well as variations on the players' utility functions. We perform
numerical analyses on these variations, and demonstrate that our
insight on the preference for lower bias is robust.

\subsection{Preference for Lower 
Bias in Normal Distribution}\label{sec:main-analysis}

In \Cref{thm:lower-bias normal} we showed that reducing bias in normal
distributions is beneficial if the variance is fixed.  Here we give a
``stronger'' result: that as long as the {\em total} error is fixed,
reducing bias (which increases variance) is beneficial.  This is
the main result of the paper.

\begin{theorem}
\label{thm:lower-bias normal under tradeoff}
Let $Z$ be normal with mean 0 and variance 1, and let $X_i=\sigma Z+\mu$ 
and $X_i'=\tau Z + \nu$, where $\mu^2+\sigma^2=\nu^2+\tau^2=1$.
If $\nu^2<\mu^2$ then $u_i(X'_i,a)> u_i(X_i,a)$ for any realization $a>0$ of player $j$.
\end{theorem}

The proof of Theorem~\ref{thm:lower-bias normal under tradeoff}, given
in 
Appendix~\ref{apx:critical inequality},
is essentially a straightforward
(albeit long and somewhat involved) calculation. The main idea is to
consider the expected utility of a player as a function of her
chosen bias and some realization of the opponent's strategy. We
calculate the derivative of this expected utility with respect to the
bias, and show that it is negative. Thus, increasing bias (and so
decreasing variance) leads to lower expected utility.

An immediate corollary of Theorem~\ref{thm:lower-bias normal under tradeoff} is that the minimal-bias
strategy is ex post strictly dominant.\footnote{For strictness we ignore the 0-measure event $a=0$, for which all strategies are payoff equivalent.}

\begin{corollary}
  \label{cor:lower-bias normal under tradeoff}
Let $\bar{\mu}\in[0,1)$ be an exogenous lower bound on the choice of $|\mu|$.  Let $Z$ be normal with mean 0 and variance 1.
The minimal-bias strategy $Z^*$ with mean $\bar{\mu}$ and variance $\sqrt{1-\bar{\mu}^2}$ 
is ex post strictly dominant within the strategy class 
$\X = \{\sigma Z + \mu~|~ 
\mu^2 + \sigma^2 = 1
~\cap~ |\mu| \geq \bar{\mu}\}$.
\end{corollary}

\paragraph{Many players and asymmetric strategies} Because no-bias is 
ex post dominant%
---and so player $i$ prefers minimal-bias for any
realization of the opponent's error---the result of \Cref{cor:lower-bias  normal under tradeoff} 
immediately extends to more than two players.  Consider the bias-variance game with any number of players,
and in which a player's payoff is $1-a^2$ if her realization $a$ is lower than all others' realizations, and
0 otherwise. Then, within the same strategy class as in \Cref{cor:lower-bias  normal under tradeoff}, 
minimal-bias is an ex post dominant strategy.
To see this, observe that player $i$'s utility can
be maximized with the opponents' errors summarized by ex post error $c
= \min_{j\neq i} a_j$.  If agent $i$ has lower error then $i$ wins,
otherwise $i$ loses. Because \Cref{cor:lower-bias  normal under tradeoff} implies that no-bias is dominant regardless
of the other's realization, it is, in particular, optimal
given realization $c$.

Another immediate implication of the result that minimal bias is optimal for player $i$ regardless of the {\em realization} of player $j$'s strategy is that $j$'s strategy class could be different from $i$'s---for example, it could include normal distributions with a different bias-variance tradeoff, or even distributions other than normal.


\subsection{Non-Constant Bias-Variance Tradeoff}\label{sec:non-constant-tradeoff}
Our main result, Theorem~\ref{thm:lower-bias normal under tradeoff}, states that, when the tradeoff between bias and variance is fixed (specifically, when $\mu^2+\sigma^2=1$) players strictly prefer a distribution with lower bias. In this section we show that the insight on the preference for lower bias persists also when the tradeoff is not fixed, and that in this case players prefer to lower bias on the margin.

To this end, consider the class of normal distributions with $\mu^2 + \sigma^2 = g(\mu^2)$, where $g:\R_+\mapsto \R_+$ is strictly convex and differentiable. Also, suppose that $\min_\mu g(\mu^2) = 1$, let $\mu_*^2 = \arg\min_{\mu^2} g(\mu^2)$, and suppose $\mu_*>0$.
Strict convexity of $g$ implies that, in the single-player game, the unique optimal choice of distribution from this class is the one in which the bias squared is $\mu_*^2$.
As with the case of constant bias-variance tradeoff, however, this is no longer the optimal choice under competition:

\begin{theorem}\label{thm:non-constant-tradeoff}
Let $Z$ be normal with mean 0 and variance 1, and for every $\mu\geq 0$ let $X_i(\mu^2)=\sigma Z + \mu$ with $\sigma^2 = g(\mu^2)-\mu^2$. Then
$\frac{du_i\left(X_i(\mu_*^2) ,a\right)}{d\mu^2}<0$ for every realization $a>0$ of player $j$.
\end{theorem}

\begin{proof}
Let $X_i(\mu^2,T) =\sigma Z + \mu$, where $\sigma^2 = T-\mu^2$, and define $v(\mu) = \begin{bmatrix}1\\ \frac{dg(\mu^2)}{d\mu^2} \end{bmatrix}$. Then for any bias $\mu$ and tradeoff $T$,
\begin{align}
    \frac{du_i\left(X_i(\mu^2) ,a\right)}{d\mu^2} &= \frac{du_i\left(X_i(\mu^2,g(\mu^2)) ,a\right)}{d\mu^2} \nonumber\\ 
    &= \nabla_{v(\mu)} u_i\left(X_i(\mu^2,T) ,a\right)\nonumber\\
    &= \frac{\partial u_i\left(X_i(\mu^2,T) ,a\right)}{\partial\mu^2} + \frac{dg(\mu^2)}{d\mu^2} \cdot \frac{\partial u_i\left(X_i(\mu^2,T) ,a\right)}{\partial T}\label{eqn:derivative}.
\end{align}
Evaluated at the point $\mu=\mu_*$, we thus have
\begin{align*}
    \frac{du_i\left(X_i(\mu_*^2) ,a\right)}{d\mu^2}
    &= \frac{\partial u_i\left(X_i(\mu_*^2,1) ,a\right)}{\partial\mu^2} + \frac{dg(\mu_*^2)}{d\mu^2} \cdot \frac{\partial u_i\left(X_i(\mu_*^2,T) ,a\right)}{\partial T}\\
    &=\frac{\partial u_i\left(X_i(\mu_*^2,1) ,a\right)}{\partial\mu^2}\\
    &<0.
\end{align*}
The first equality follows since $g(\mu_*^2)=T=1$. The second equality holds since $g$ has minimum at $\mu_*^2$ and so $\frac{dg(\mu_*^2)}{d\mu^2}=0$. Finally, the inequality follows from Theorem~\ref{thm:lower-bias normal under tradeoff}.
\end{proof}

\subsection{Motivation for Normal Distribution}
\label{sec:normal-motivation}

One of the main assumptions we make in our analysis above is that the
error of players' algorithms is normally distributed.  This assumption
is natural in the context of machine learning algorithms, and many
commonly-used econometric and machine learning procedures with tuning
parameters determining the bias-variance tradeoff have been
demonstrated to produce predictions with
asymptotically normal error distributions; some examples can be found in
\citet{chen:99}, \citet{ormoneit2002kernel}, \citet{hable2012asymptotic}, and
\citet{wager2018estimation}.
To formally describe these results we need a definition:

\begin{definition}
  \label{def:asymptotic-normal}
  Let $\{\X_1,\X_2,\ldots\}$ be an infinite sequence of
  random variables.  The sequence is {\em asymptotically normal with asymptotic bias
    $\mu$ and asymptotic variance $\sigma^2$} if $\sqrt{m} \X_m$ converges with $m$ in distribution to $\N(\mu,\sigma^2)$, the normal distribution
  with bias $\mu$ and variance $\sigma^2$.\footnote{Rearranging, we see that as the
    sequence converges so $\X_m$ resembles $\N(\mu/\sqrt{m},\sigma^2/m)$.
    Asymptotic normality is a consequence of the Central Limit Theorem
    where the variance converges at a rate of $1/m$ (the standard
    deviation converges at a rate of $1/\sqrt{m}$); thus, $\sigma^2$
    is the variance normalized by $m$ which converges to a constant.}
\end{definition}

Parameter estimates obtained from the ridge regression---which we employ for our empirical results in
\Cref{sec:empirical}---are known to be
asymptotically normal, with $m$ corresponding to the size of the 
training sample \citep[e.g.,][]{hoerl:70,brown:80}. In this section we provide some formal
definitions and statements to describe these results.

Ridge regression is linear
regression with a quadratic regularizer.  Features $x$ lie in a $p$
dimensional space $\R^p$ and hypotheses are given by weights $w \in
\R^p$ as the linear weighted sums of the features, i.e., $w^Tx$.  The
hypothesis selected by the ridge regression on the set of $m$ training
points $\{x_i,y_i\}^m_{i=1}$ is $\hat f(x) = \widehat{w}^Tx$, where the
weights $\widehat{w}$  minimize
the following quantity, the regularized empirical risk
function with the quadratic regularizer and regularization parameter
$\ridge$:
\begin{align}
\label{eq:ridge}
\widehat{R}(w) &= \frac{1}{2m}\sum\nolimits^m_{i=1}(y_i-w^Tx_i)^2+\ridge\|w\|^2_2.
\end{align}

The regularization parameter $\ridge$ is responsible for the
bias-variance tradeoff in the prediction.  The presence of the
regularizer reduces the dependence of the vector of predicted weights
$\widehat{w}$ on individual observations, making the corresponding
predictions more ``stable'' as the regularization parameter $\ridge$
increases. At the same time, the increase in $\ridge$ makes the
prediction less data-dependent and, therefore, more biased.  In other
words, while increasing $\ridge$ increases the bias of prediction, it
decreases the variance. This last statement is formalized in Proposition~\ref{prop:multidimensional-ridge}
in Appendix~\ref{apx:multidimensional-ridge}.

In the ``classic" ridge regression model of \citet{hoerl:70} for a fixed
regularization parameter $\lambda,$ asymptotic bias of the estimator remains finite while its 
variance decreases at rate $1/m$ as the sample size $m \rightarrow \infty.$ A possible approach 
to maintain a meaningful bias-variance tradeoff is to re-scale the regularization parameter to 
set it $\lambda=c/\sqrt{m}$ in which case the ``new" regularization parameter is $c.$ This is the approach 
that was pursued in \citet{brown:80} and the follow-up literature. 
In the recent literature on ridge regression with high-dimensional regressors, including \citet{karoui:13}, \citet{dobriban:18}, \citet{hastie:19}
\citet{tsigler:20}, and \citet{wu2020optimal}, nonzero asymptotic bias and variance arise in the regime with the dimension
$p$ of the feature vector growing with the size of the sample $m.$



\subsection{Numerical 
Results for Other Distributions
and Payoffs}\label{sec:numerical}

\begin{figure*}
\begin{subfigure}{.5\textwidth}
\begin{center}
\includegraphics[height=6cm]{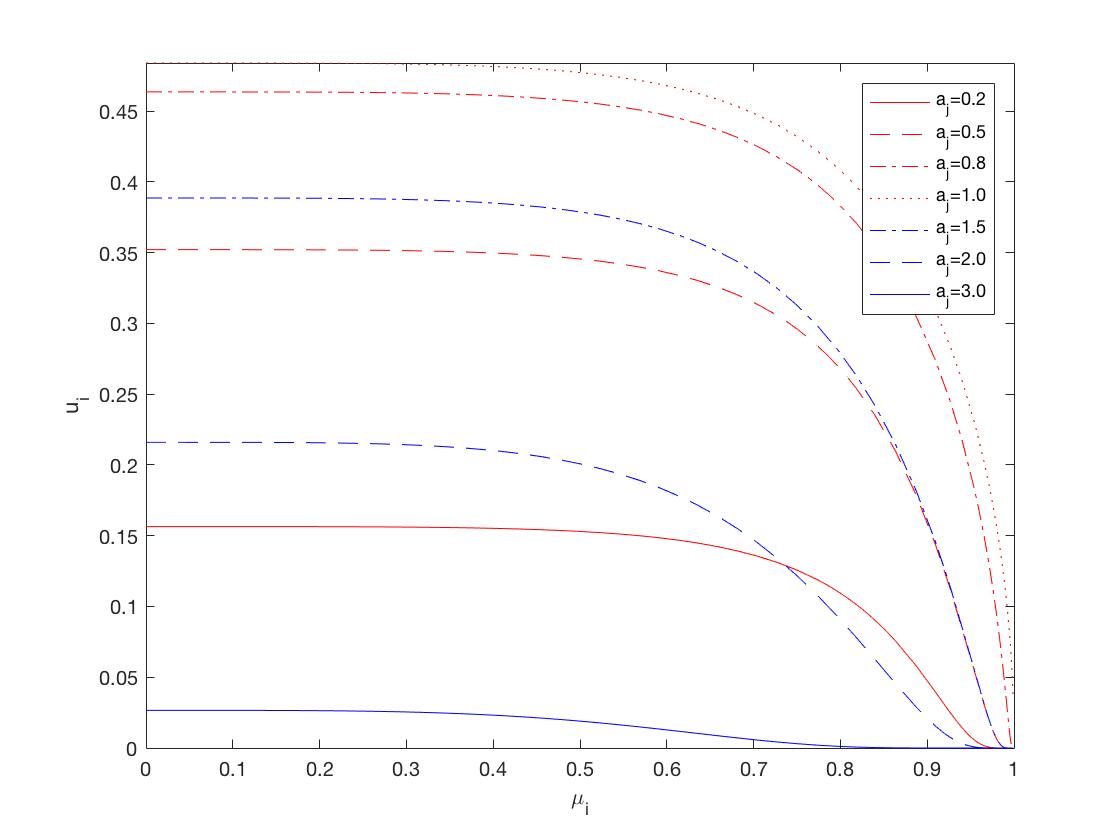}
  \caption{Normal Distribution}
\label{fig:util normal}
\end{center}
\end{subfigure}
\begin{subfigure}{.5\textwidth}
\begin{center}
\includegraphics[height=6cm]{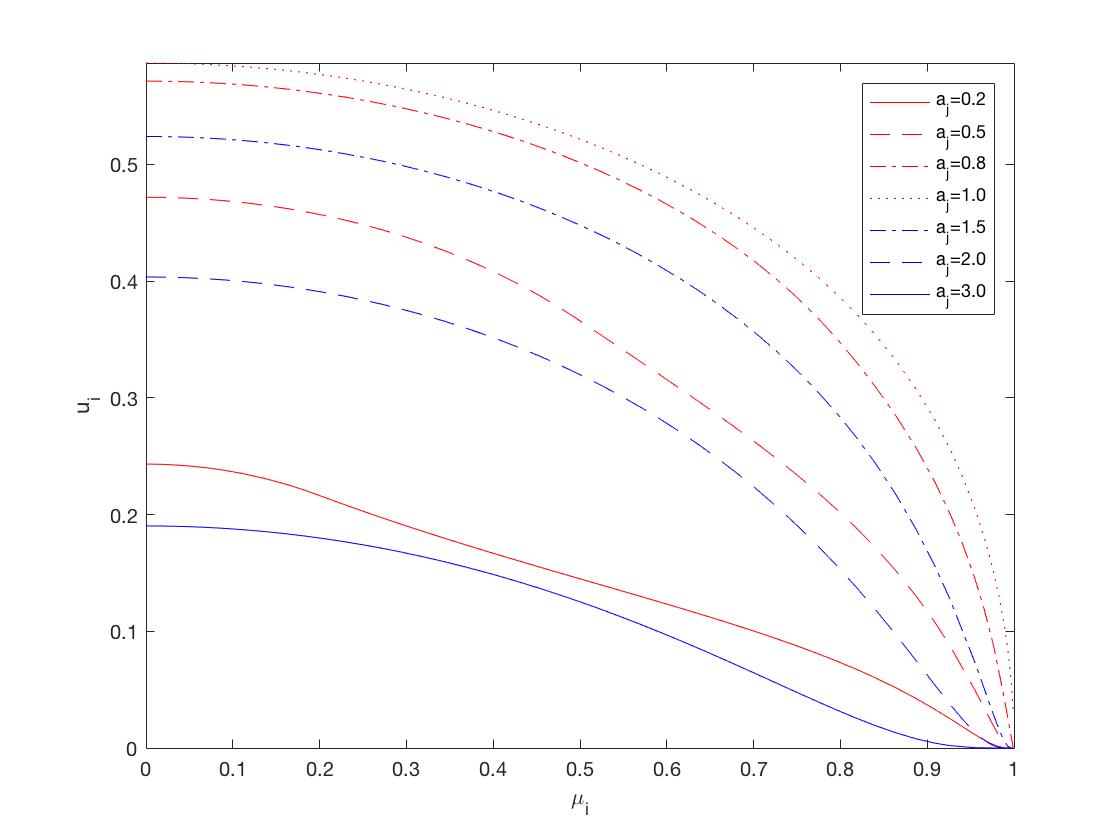}
  \caption{Laplace Distribution}
\label{fig:util laplace}
\end{center}
\end{subfigure}
\caption{
Ex post utility curves of player 
$i$ against
different realization $a_j$
from opponent player $j$.
}
\label{fig:ex post utility}
\end{figure*}

In this section we illustrate the robustness of
\Cref{thm:lower-bias normal under tradeoff}.  Specifically, in the
proof of \Cref{thm:lower-bias normal under tradeoff}, there are three
assumptions that enable a clean closed form for ex post utility, and
thus simplify our argument: (a) the
the error distributions are normal; (b) the utility function is $u_i(a_i,
a_j) = (1 - a_i^2)\cdot {\bf 1}\{a_i < a_j\}$; and (c)
 $\mu^2 + \sigma^2 = {\bf 1}$. The fact that both the benefit from winning and the total error are 1 simplifies the proof of \Cref{thm:lower-bias normal under tradeoff}.

In this section
we vary the 
distributions, utility function,
and tradeoff, and numerically evaluate 
the players' utilities.
The subsequent calculations  
suggest that our insight on the preference of variance over bias holds generally in many 
cases
beyond our theoretical 
assumptions (a), (b) and (c).

\paragraph{Other distributions.}
We first numerically calculate the ex post utility curves against arbitrary
realizations from the opponent player, as well as the expected utility curves
against arbitrary strategies of the opponent, on various
distributions.  As shown in \Cref{thm:lower-bias normal under
  tradeoff}, with the normal distribution, the ex post utility curve is
always decreasing, for all opponents' realizations.  Numerical evaluations with the \emph{Laplace} distribution indicate that its ex post
utility curve is also decreasing for all opponents' realizations;
see \Cref{fig:ex post utility}, where we plot the ex post utility
curves $u_i(\mu_i, a_j)$ holding realization $a_j$ fixed for both normal
and Laplace distributions.  The $x$-axis of each figure
is player $i$'s bias $\mu_i$, and the $y$-axis is the ex post utility
$u_i(\sqrt{1-\mu_i^2}Z+\mu_i, a_j)$.  Another interesting observation
from this numerical result is that with normal distributions, the ex
post utility is almost flat for $\mu_i \leq 0.5$---i.e., while
$\mu_i=0$ is a dominant strategy, picking any $\mu_i \in [0,0.5]$ is a
``pretty good'' strategy.  

For the \emph{logistic} distribution, the ex
post utility curve is no longer decreasing; however, the expected
utility curve is decreasing against arbitrary strategies of the
opponent. Thus, the no-bias distribution
$X_i = Z$ is a dominant (although not ex post dominant) strategy. 

Finally, for \emph{uniform}
distributions, monotonicity does not hold on either the ex post utility
curve or the expected utility curve. In fact, there exists a pure Nash equilibrium with non-zero bias. See \Cref{fig:expected utility}
for utility plots under logistic and uniform distributions.

These plots and calculations indicate that the preference for variance over bias is robust---and holds for Laplace and logistic distributions---but not universal---as
demonstrated by the uniform distribution

\begin{figure*}
\begin{subfigure}{.5\textwidth}
\begin{center}
\includegraphics[height=6cm]{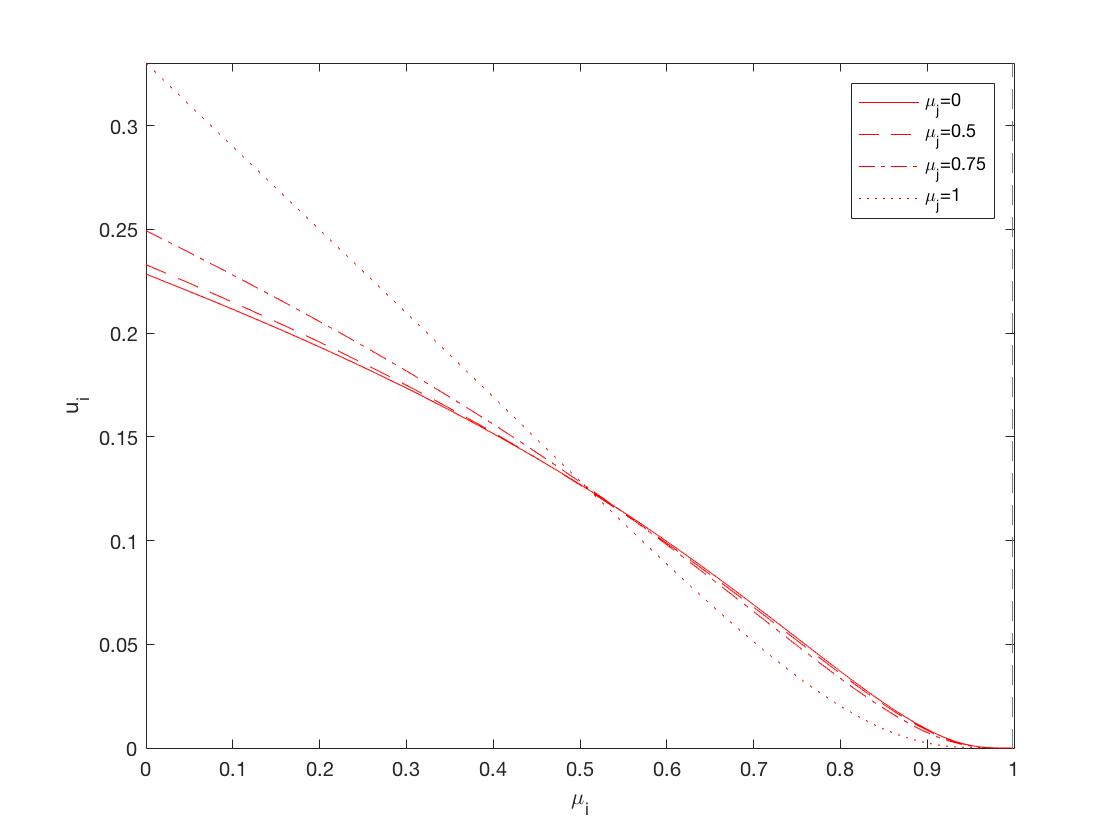}
  \caption{Logistic Distribution}
\label{fig:util logistic}
\end{center}
\end{subfigure}
\begin{subfigure}{.5\textwidth}
\begin{center}
\includegraphics[height=6cm]{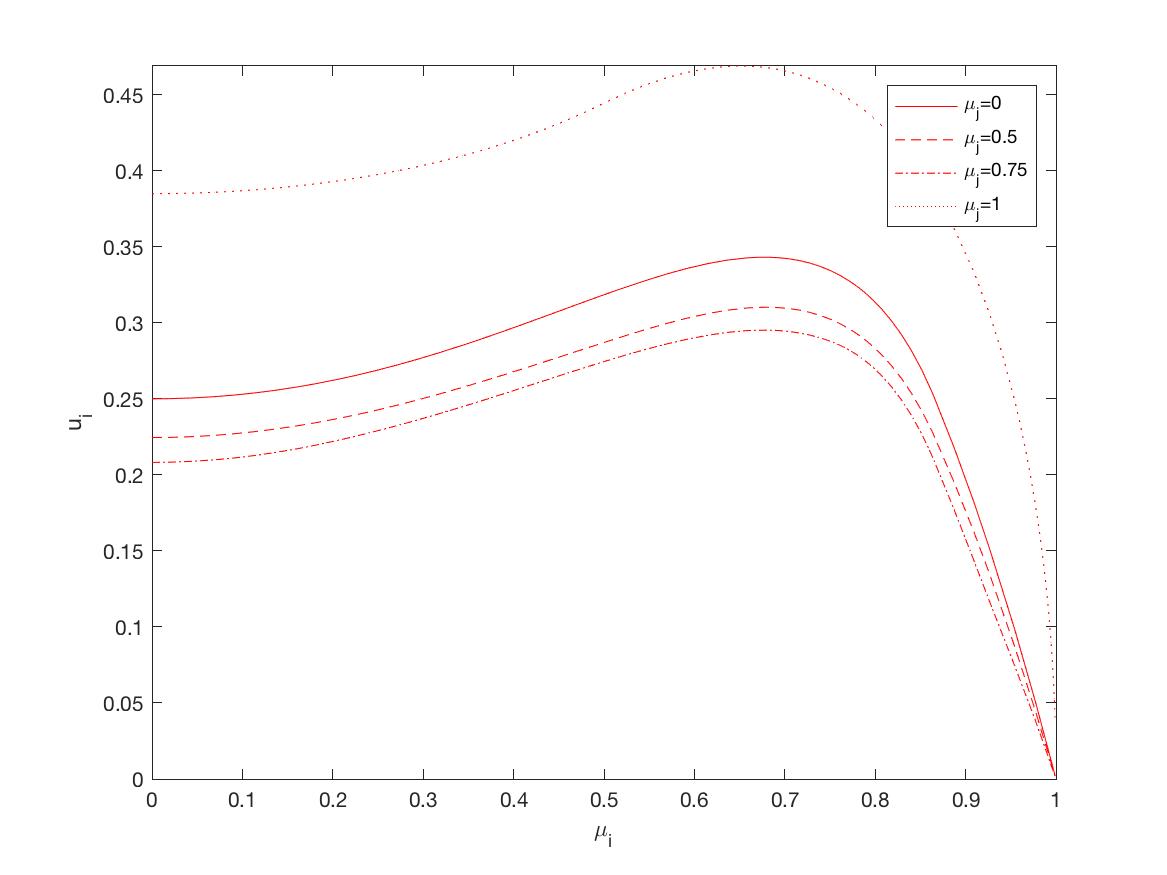}
  \caption{Uniform Distribution}
\label{fig:util uniform}
\end{center}
\end{subfigure}
\caption{
Expected utility curves
of player $i$ against
different strategy $\mu_j$
from opponent player 
$j$.
}
\label{fig:expected utility}
\end{figure*}

\paragraph{Other utility functions.}
Another assumption that is important for our theoretical analysis is
the form of the utility function. More generally, suppose players'
utility functions $u_i(a_i, a_j) = (R - a_i^2)\cdot {\bf 1}\{ a_i <
a_j\}$ for arbitrary reward $R > 0$, with normal distribution $Z$. 
Observe that for very large $R$, the resulting game is close to a
constant-sum game, as the error terms are nearly irrelevant, and
that the probability of negative payoffs is negligible.

Our
main theoretical result considered the case $R=1$, but for general
reward $R \not= 1$, in which the individual rationality constraint
either does not bind or is violated, the monotonicity of ex post
utility may not hold.  However, numerical calculations demonstrate
that the expected utility curve is still decreasing against arbitrary
strategies of the opponent, and so no-bias remains a dominant
strategy---see \Cref{fig:expected utility different payoff}.

\begin{figure*}[t]
\begin{subfigure}{.5\linewidth}
\begin{center}
\includegraphics[height=6cm]{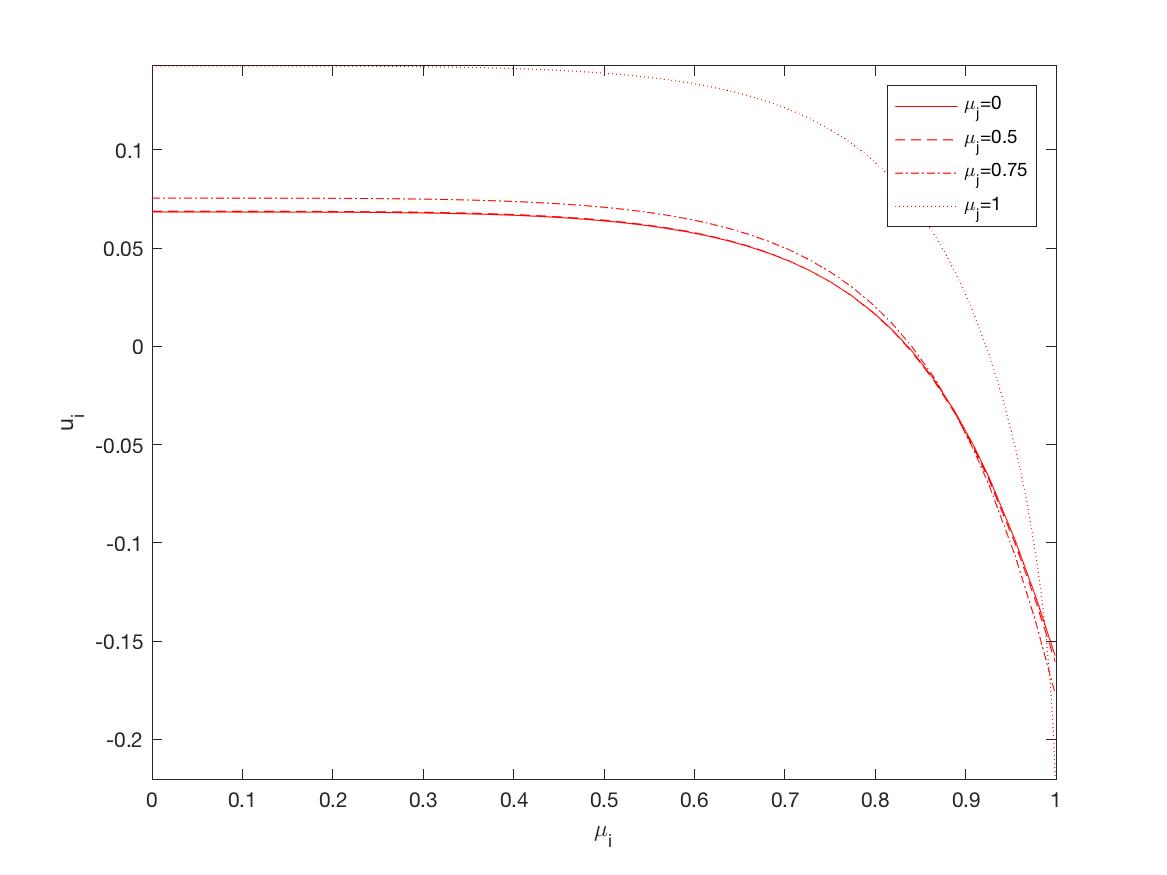}
  \caption{$u_i(a_i, a_j) = (0.5 - 
  a_i^2)
  \cdot {\bf 1}\{a_i < a_j\}$}
\label{fig:util normal half}
\end{center}
\end{subfigure}
\begin{subfigure}{.5\linewidth}
\begin{center}
\includegraphics[height=6cm]{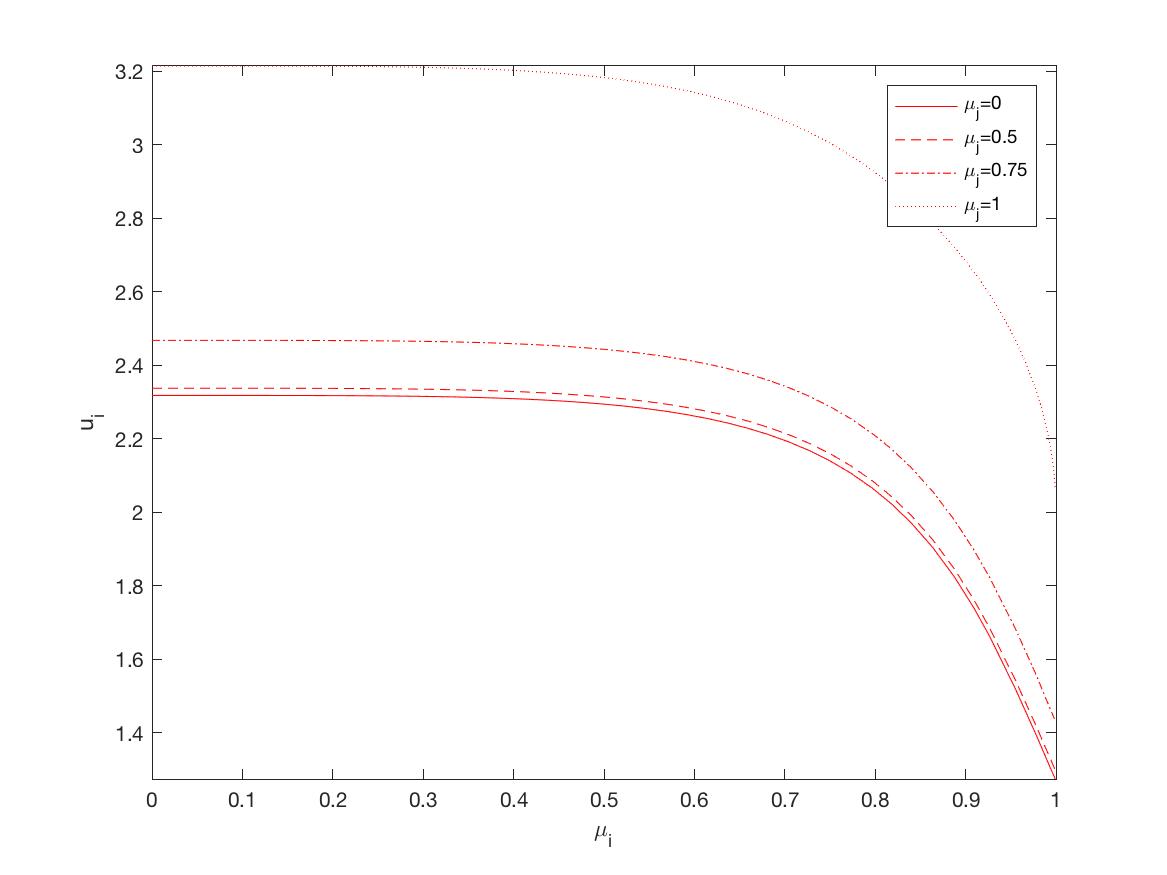}
  \caption{$u_i(a_i, a_j) = (5 - 
  a_i^2)
  \cdot {\bf 1}\{a_i < a_j\}$}
\label{fig:util normal five}
\end{center}
\end{subfigure}
\caption{
Expected utility curves
of player $i$ against
different strategy $\mu_j$
from opponent player 
$j$ with normal distribution. Rewards are $R=0.5$ and $R=5$, respectively.
}
\label{fig:expected utility different payoff}
\end{figure*}

\paragraph{Other tradeoffs.}
Finally, our theoretical analysis focused on
the assumption that  $\mu^2 + \sigma^2=1$, but more generally one might suppose that the tradeoff (for both players) is different. 
Similarly to what we observe 
when we vary the reward $R$ in the utility function, 
the ex post monotonicity of
utility may not hold
when $\mu^2 +\sigma^2 \not=1$.  
However, numerical calculations demonstrate
that the expected utility curve is still decreasing against arbitrary
strategies of the opponent, and so no-bias remains a dominant
strategy---see \Cref{fig:expected utility different tradeoff}.

\begin{figure*}[t]
\begin{subfigure}{.5\linewidth}
\begin{center}
\includegraphics[height=6cm]{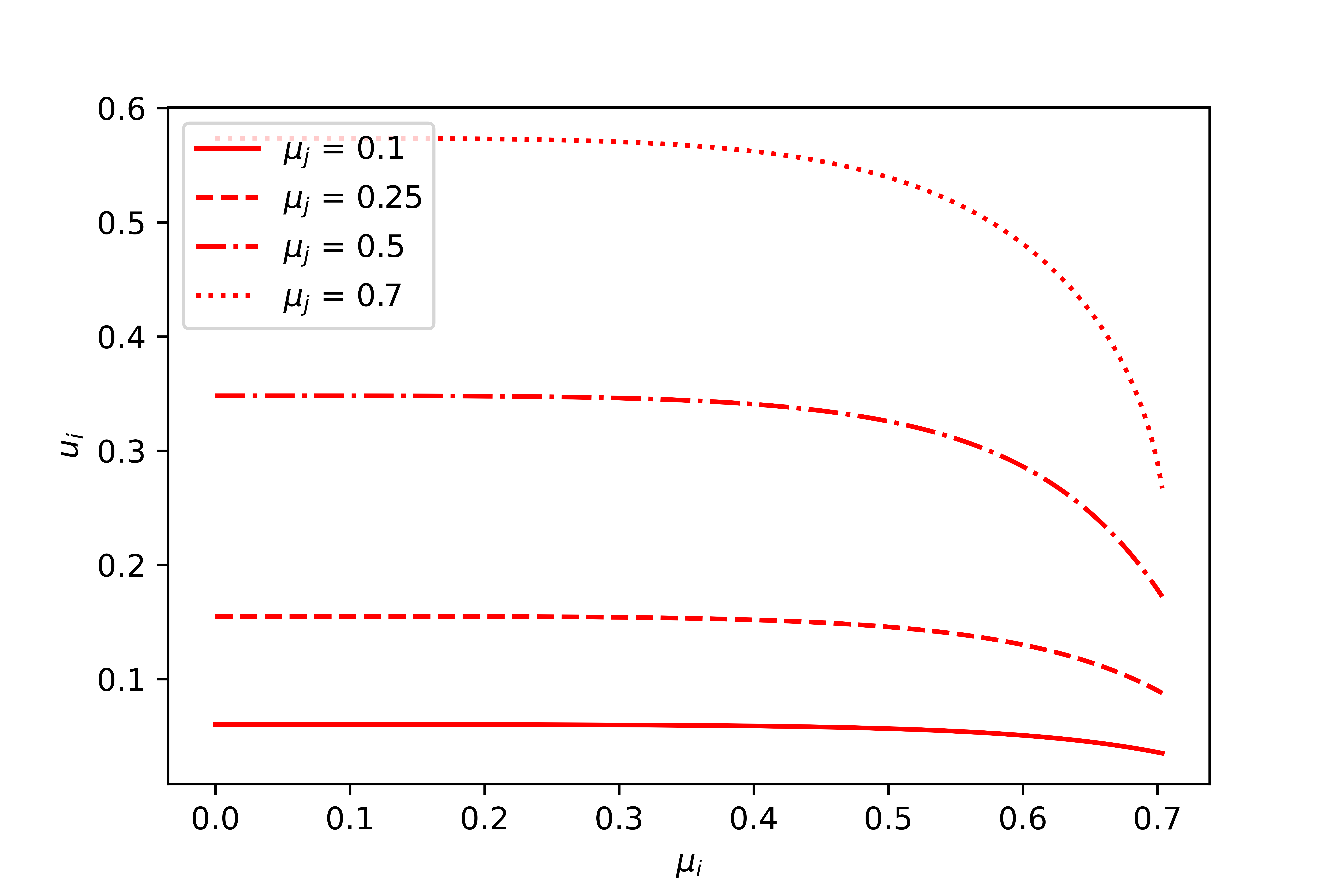}
  \caption{$\mu^2+\sigma^2=0.5$}
\label{fig:util normal tradeoff half}
\end{center}
\end{subfigure}
\begin{subfigure}{.5\linewidth}
\begin{center}
\includegraphics[height=6cm]{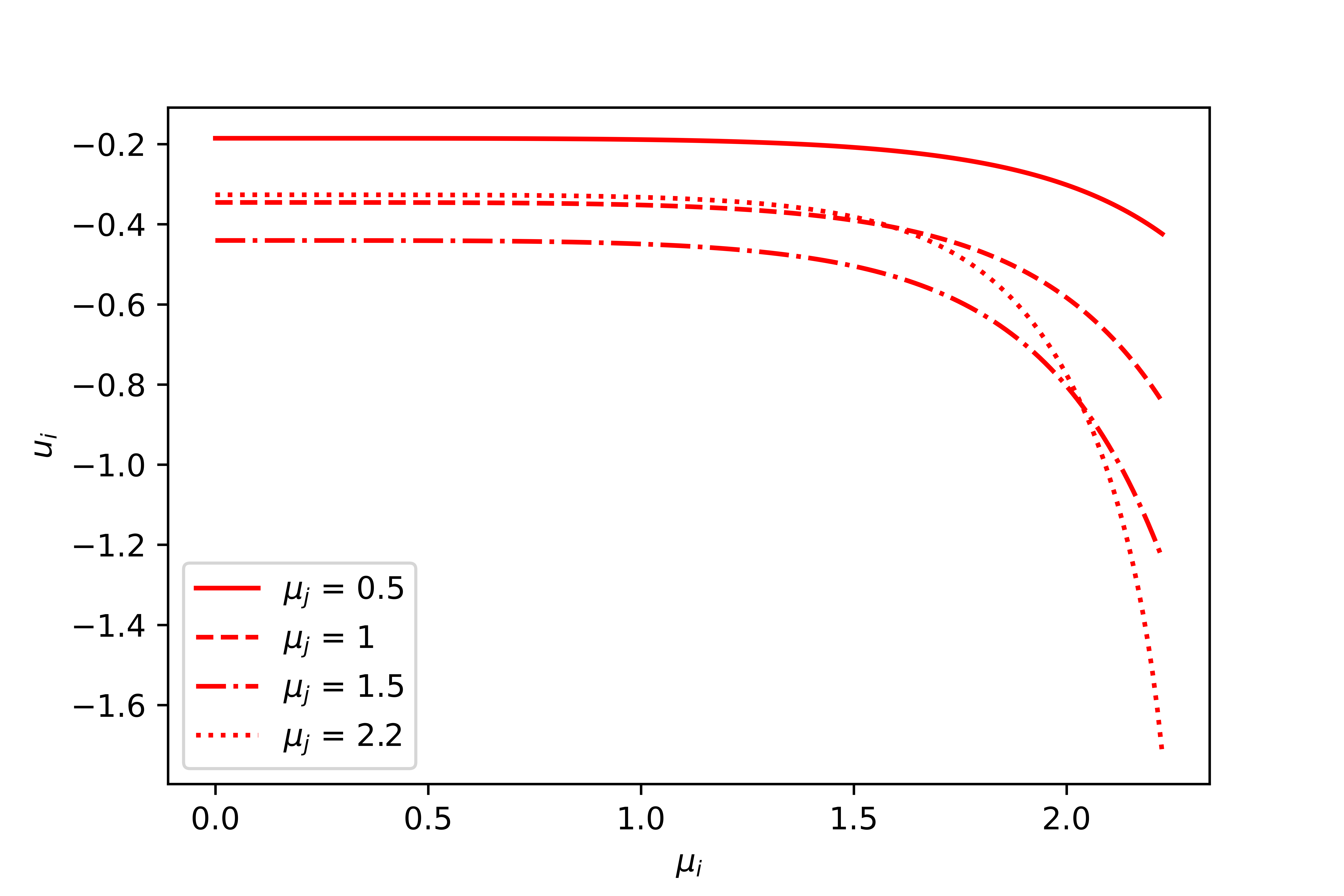}
  \caption{$\mu^2+\sigma^2=5$}
\label{fig:util normal tradeoff five}
\end{center}
\end{subfigure}
\caption{
Expected utility curves
of player $i$ against
different strategy $\mu_j$
from opponent player 
$j$ with normal distribution. Tradeoffs 
are $\mu^2+\sigma^2=0.5$ and 
$\mu^2+\sigma^2=5$, respectively.
}
\label{fig:expected utility different tradeoff}
\end{figure*}

\section{The Ex Ante Game} \label{sec:non-constant-tradeoff-general-framework}
Having established players' preferences for lower bias in the interim stage, we now turn to analyze the ex ante stage. In our analysis thus far we assumed that $x$ is fixed, and were concerned with the effects of different algorithms' estimates of $f(x)$. We now take a step back to the ex ante stage of the general framework, where the choice of algorithm is made prior to observing the realized feature vector $x$.

Fix a learning problem as described in Section~\ref{sec:general-framework}: a function $f$ and a distribution $\pi$ over $(x, f(x))$ pairs. For simplicity, we will suppose the there is a fixed class of available estimators indexed by a regularization parameter, namely $\{\hat f_\lambda\}_{\lambda\in\R_+}$. That is, when players choose algorithms, their decision consists of choosing a parameter $\lambda$, thus yielding the estimator $\hat f_\lambda$.
The ex ante game is then the following:
\begin{enumerate}
    \item Each player $i$ obtains an independent dataset $D_i$ and chooses a regularization parameter $\lambda_i$.
    \item A pair $(x,f(x))$ is chosen according to $\pi$.
    \item Utilities are realized:
$$u_i\left(x;\hat f_{\lambda_i} (x; D_i),\hat f_{\lambda_j} (x;D_j)\right) =  1- \left(\hat f_{\lambda_i}(x;D_i) - y\right)^2$$
if $\abs{\hat f_{\lambda_i}(x;D_i) - f(x)}<\abs{\hat f_{\lambda_j}(x;D_j) - f(x)}$, and 
$$u_i\left(x;\hat f_{\lambda_i} (x; D_i),\hat f_{\lambda_j} (x;D_j)\right) =0$$
otherwise.\footnote{In the 0-measure event that predictions are identical, suppose ties are broken uniformly at random.}
\end{enumerate}

Let $\lambda^*$ be the regularization parameter for which the Bayes risk  is minimal, namely $\lambda^* = \arg\min_\lambda R_B(\hat f_\lambda)$. In particular, for this choice of $\lambda$, the sum of squared-bias and variance under $\hat f_\lambda$, in expectation over $x$, is minimal.

Denote $\hat f^* = \hat f_{\lambda^*}$. For each $x$, this optimal estimator $\hat f^*$ yields 
a particular (and possibly different) bias $\mu_x=\mathrm{Bias}_D[\hat f^*(x;D)]$ and variance $\sigma_x^2=\mathrm{Var}[\hat f^*(x; D)]$. Let $T_x = \mu_x^2+\sigma_x^2$ be the bias-variance tradeoff at $x$ yielded by $\hat f^*$. It is straightforward to see that, in the one-player game where the utility is one minus the squared-error, this choice of $\lambda^*$, and so of $\hat f^*$, maximizes expected utility. Once again, we will show that under competition players can profitably deviate by lowering $\lambda$ and choosing an estimator with lower bias (and higher total error).

Our result, Theorem~\ref{thm:general-framework-non-constant-tradeoff} below, relies on five assumptions that are sufficient to guarantee the preference for lower bias. Before stating the theorem we list the assumptions and their theoretical, numerical, and empirical justifications.

\begin{itemize}
    \item[A1] For each $x$ and every $\lambda$, the distribution of $\hat f_\lambda(x;D)$ is normal. Note that this is a generalization of the normality assumption we made in analysis of the interim stage, with formal justification in Section~\ref{sec:normal-motivation} and empirical demonstration in Appendix~\ref{apx:empirical-validation}.
    \item[A2] The result of Theorem~\ref{thm:lower-bias normal under tradeoff} holds also when the tradeoff between squared-bias and variance is not 1: For any $T$, if both players choose a distribution from the class $\X=\{\sigma Z + \mu: \mu^2+\sigma^2=T\}$ then lower bias is a dominant strategy. Note that our simulations on different tradeoffs in Section~\ref{sec:numerical} suggest that this is true.
    \item[A3] Lowering $\lambda$ simultaneously decreases $\mathrm{Bias}_D[\hat f_\lambda(x;D)]$ for every $x$. Formally,
    $$\frac{d}{d\lambda}\mathrm{Bias}_D[\hat f_\lambda(x;D)]>0$$
    for every $x$.
    We note that this assumption is provably true in the specific case where each $\hat f_\lambda$ is a ridge regression with regularization parameter $\lambda$---see Proposition~\ref{prop:multidimensional-ridge}.
       \item[A4] For each $x$, let $g_x(\lambda)$ be the sum of squared-bias and variance at $x$ and $\lambda$:
    $$g_x(\lambda) = \mathrm{Bias}_D[\hat f_\lambda(x;D)]^2 + \mathrm{Var}[\hat f_\lambda(x;D)].$$
    Then $\E_x[g_x(\lambda)]$ is strictly convex and differentiable.
    \item[A5]  Let $X(\mu^2, T)$ denote an independent normal random variable with mean $\mu\geq 0$ and variance $\sigma^2$ such that $\mu^2+\sigma^2=T$. Then at the optimal $\lambda^*$, the random variables
     $$\frac{dg_x(\lambda^*)}{d\lambda}~~\mbox{and}~~\frac{\partial u_i\left(x;X(\mu_x^2,T) ,X(\mu_x^2,T_x)\right)}{\partial T},$$ where the randomness is over $x$ and the $X(\cdot,\ \cdot)$ normal distributions, are (weakly) negatively correlated.
    We note that a sufficient condition for this assumption to hold is that for every $x$, the minimal risk is achieved at $\mu_x$. In particular, we prove that this is true for ridge regression with one-dimensional feature vectors---see  Proposition~\ref{prop:one-dimensional-ridge-A5}. In addition, in Appendix~\ref{apx:empirical-validation} we validate assumption A5 empirically in the same datasets and games from Section~\ref{sec:empirical}.
\end{itemize}

We now state and prove our main theorem for this section.
Let $\lambda^*$ be the optimal regularization parameter and $\hat f^*$ the resulting estimator. Then when player $j$ chooses $\hat f^*$, player $i$ can profitably deviate by choosing $\hat f_\lambda$ with $\lambda <\lambda^*$. Formally:
\begin{theorem}\label{thm:general-framework-non-constant-tradeoff}
Suppose assumptions A1-A5 hold.
Then 
$$\frac{d}{d\lambda} \E\left[u_i\left(x; \hat f_{\lambda^*}(x; D_i),\hat f^*(x;D_j)\right)\right]<0.$$
\end{theorem}
\ifEC
\noindent The proof follows similar lines to that of Theorem~\ref{thm:non-constant-tradeoff}, and appears in Appendix~\ref{apx:general-framework-non-constant-tradeoff}.
\else
\begin{proof}
By assumption A1, $$\hat f^*(x;D_j)=X(\mu_x^2,T_x),$$
a normal distribution with squared-bias $\mu_x^2$ and total error $T_x=g_x(\lambda^*)$.
Denote by $g(x;\mu^2)$ the total error of $\hat f^*$ when the feature vector is $x$ and the squared-bias is $\mu^2$.
Then, as in Equation (\ref{eqn:derivative}), for each fixed $x$ we can write
\begin{align*}&\frac{du_i\left(x;X(\mu^2,g(x;\mu^2)) ,X(\mu_x^2,T_x)\right)}{d\mu^2}\\
&~~~=\frac{\partial u_i\left(x;X(\mu^2,T) ,X(\mu_x^2,T_x)\right)}{\partial\mu^2} + \frac{dg(x;\mu^2)}{d\mu^2} \cdot \frac{\partial u_i\left(X(\mu^2,T) ,X(\mu_x^2,T_x)\right)}{\partial T},
\end{align*}
where $T=g(x;\mu^2)$.
Thus,
\begin{align*}&\frac{du_i\left(x; \hat f_{\lambda^*}(x; D_i),\hat f^*(x;D_j)\right)}{d\lambda} \\
&~~~=\frac{d\mu^2}{d\lambda}\cdot\frac{\partial u_i\left(x;X(\mu^2,T_x) ,X(\mu_x^2,T_x)\right)}{\partial\mu^2} + \frac{d\mu^2}{d\lambda}\cdot\frac{dg(x;\mu_x^2)}{d\mu^2} \cdot \frac{\partial u_i\left(X(\mu_x^2,T) ,X(\mu_x^2,T_x)\right)}{\partial T}\\
&~~~=\frac{d\mu^2}{d\lambda}\cdot\frac{\partial u_i\left(x;X(\mu^2,T_x) ,X(\mu_x^2,T_x)\right)}{\partial\mu^2} + \frac{dg_x(\lambda^*)}{d\lambda} \cdot \frac{\partial u_i\left(X(\mu_x^2,T) ,X(\mu_x^2,T_x)\right)}{\partial T}.
\end{align*}

By assumption A3,
$$\frac{d\mu^2}{d\lambda}>0,$$
and by assumption A2, 
$$\frac{\partial u_i\left(x;X(\mu^2,T_x) ,X(\mu_x^2,T_x)\right)}{\partial\mu^2} < 0.$$

By assumption A5, 
\begin{align*}
&\E\left[\frac{dg_x(\lambda^*)}{d\lambda} \cdot \frac{\partial u_i\left(X(\mu^2,T) ,X(\mu_x^2,T_x)\right)}{\partial T}\right]\\
&~~~\leq \E\left[\frac{dg_x(\lambda^*)}{d\lambda}\right] \cdot \E\left[\frac{\partial u_i\left(X(\mu^2,T) ,X(\mu_x^2,T_x)\right)}{\partial T}\right].
\end{align*}

Finally, since $\lambda^*$ minimizes the Bayes risk, assumption A4 implies that
$$\E\left[\frac{dg_x(\lambda^*)}{d\lambda}\right] = \frac{d}{d\lambda}\E[g_x(\lambda^*)]=0.$$

Putting these together implies the claimed inequality.
\end{proof}
\fi

Theorem~\ref{thm:general-framework-non-constant-tradeoff} states that, when both players choose  $\lambda^*$ to minimize the Bayes risk, then each player can profitably deviate by choosing a lower $\lambda$. Thus, the action profile in which both players choose $\lambda^*$ is not a Nash equilibrium. Under a stronger version of assumption A2, we can show that lowering $\lambda$ below $\lambda^*$ is strictly beneficial for player $i$ for {\em every} choice $\lambda_j$ of player $j$. In particular, this holds under assumption A2':
\begin{itemize}
\item[A2'] The result of Theorem~\ref{thm:lower-bias normal under tradeoff} holds also when the tradeoff between squared-bias and variance is different for the two players: For any $T_1$ and $T_2$, if each player $i$ chooses a distribution from the class $\X_i=\{\sigma Z + \mu: \mu^2+\sigma^2=T_i\}$ then lower bias is a dominant strategy.
\end{itemize}

\section{Empirical Results}
\label{sec:empirical}

In this
section we test our insight on the preference of variance
over bias in real datasets.

\paragraph{Data description}
We consider two
widely-used benchmark datasets, chosen mostly for convenience and
because they are 
standard datasets often used to test new ideas in
machine learning.  
We implement a similar regression problem 
(discussed below) 
separately for the two datasets.

The first dataset is
the \emph{California housing
prices data}
from the 1990 Census, a dataset first utilized by
\citet{pace1997sparse} and included in the Python Scikit-learn
library. 
It includes 20640 observations on 9
features, such as number of rooms, median income, etc.  
From
this data we set up a regression problem of predicting the order of
magnitude of the median house prices (i.e., its logarithm) from the
other features.  For features where orders of magnitude may be more
relevant than absolute magnitude, we include both the feature and its
logarithm.

The second dataset is the 
\emph{wine quality data}
designed by \citet{cortez2009modeling}.
It records 12 features 
for 6497 observations,
where 11 features are physicochemical
properties and the last feature is 
the quality of wine.
The dataset is partitioned into
two sub-datasets -- 
one for red wine 
with 1599 observations 
and one for white wine
with 4898 observations.
We set up the regression problems
of predicting the quality of wine 
from the other physicochemical properties 
for the two sub-datasets separately.

\paragraph{Experiment setup}
We analyze a game between two players, 
each of whom gets roughly half
the data and employs a ridge regression 
in order to predict 
median housing prices 
in
the California housing prices data
(resp.\ 
qualities of wine in 
the wine quality data).
In both datasets,
various features 
tend to
be highly correlated, 
which makes the design matrix
(covariance matrix of the feature vector) nearly 
non-invertible. 
This, in turn, makes the estimated
feature weights in the linear regression unstable,
largely varying from sample to sample. 
A common tool used to stabilize the estimated feature weights is linear regression
with regularization, namely, ridge regression.
Running a ridge regression involves setting a regularization parameter
$\ridge$, where $\ridge=0$ is the standard ordinary least squares
(OLS) regression, and increasing $\ridge$ leads to a model with
greater bias but lower variance of the resulting predictions. Thus, we will let the players play
the game for various values of $\ridge$.  Further discussion of ridge
regression is given in \Cref{sec:normal-motivation}.

The design of the game is as follows:
\begin{enumerate}
\item The dataset is uniformly partitioned with 
10\% as the test set
and 90\% as the training set.
Each player 
uniformly draws 50\% 
of the training set 
as her own training set.
\item Players choose regularization parameters $\ridge_1$ and
  $\ridge_2$, respectively, and perform ridge regression on their own
  training sets.
\item The expected payoffs of players are given by the bias-variance
  game (\Cref{BVG} and the description in \Cref{sec:non-constant-tradeoff-general-framework}) evaluated on the test set.  On each point in the
  test set, the player whose prediction is closest to the true label
  wins and obtains payoff of one minus the squared error of the
  prediction; the other player's payoff is zero.  Each player's
  payoff in the game is the average payoff over the points in the test
  set.
\end{enumerate}

\paragraph{Experiment results}
We used Monte Carlo simulations to approximate expected utilities of
the players. We repeat steps 1-3 above by independently drawing
training and validation samples 100 times for the California housing prices data
(resp.\ 10000 times for the wine quality data)
and computing payoffs that
result from a given pair of choices of regularization parameters
$(\ridge_1,\ridge_2)$.
We contrast the optimal choice of regularization parameter in the
single-player setting with the two-player game.

%

%

For the California housing prices data,
Player~1's utility in the
single-player setting is depicted 
in \Cref{fig:empirical_single}
as a function of the player's regularization
parameter, ranging from $\ridge_1=0$ to $\ridge_1=1000$.
The optimal
parameter choice is $\ridge_1\approx 100$.  In \Cref{fig:empirical_two},
Player~1's utility in the two-player game is plotted as a function of
$\ridge_1$ for fixed values of Player~2's regularization parameter
$\ridge_2$.  
For all choices of $\ridge_2$, Player~1's utility is
optimized by selecting $\ridge_1 = 0$.  

\begin{figure*}
\begin{subfigure}{.5\linewidth}
\begin{center}
\begin{tikzpicture}[scale=0.45]\begin{axis}[height=11cm,width=18cm,ymin=.6,ymax=.7,xmin=0,xmax=1000,tick pos=left]
\addplot [draw = blue, line width=0.3mm] coordinates {
(0, 0.6339)
(10, 0.6561)
(20, 0.6643)
(30, 0.6688)
(40, 0.6715)
(50, 0.6732)
(60, 0.6743)
(70, 0.6751)
(80, 0.6755)
(90, 0.6758)
(100, 0.6758)
(110, 0.6758)
(120, 0.6757)
(130, 0.6755)
(140, 0.6753)
(150, 0.6750)
(160, 0.6746)
(170, 0.6743)
(180, 0.6739)
(190, 0.6735)
(200, 0.6730)
(210, 0.6726)
(220, 0.6721)
(230, 0.6717)
(240, 0.6712)
(250, 0.6707)
(260, 0.6702)
(270, 0.6697)
(280, 0.6692)
(290, 0.6687)
(300, 0.6682)
(310, 0.6677)
(320, 0.6672)
(330, 0.6667)
(340, 0.6661)
(350, 0.6656)
(360, 0.6651)
(370, 0.6646)
(380, 0.6641)
(390, 0.6636)
(400, 0.6631)
(410, 0.6626)
(420, 0.6621)
(430, 0.6616)
(440, 0.6611)
(450, 0.6606)
(460, 0.6601)
(470, 0.6596)
(480, 0.6591)
(490, 0.6586)
(500, 0.6581)
(510, 0.6577)
(520, 0.6572)
(530, 0.6567)
(540, 0.6562)
(550, 0.6558)
(560, 0.6553)
(570, 0.6548)
(580, 0.6544)
(590, 0.6539)
(600, 0.6535)
(610, 0.6530)
(620, 0.6526)
(630, 0.6521)
(640, 0.6517)
(650, 0.6512)
(660, 0.6508)
(670, 0.6503)
(680, 0.6499)
(690, 0.6495)
(700, 0.6490)
(710, 0.6486)
(720, 0.6482)
(730, 0.6478)
(740, 0.6474)
(750, 0.6469)
(760, 0.6465)
(770, 0.6461)
(780, 0.6457)
(790, 0.6453)
(800, 0.6449)
(810, 0.6445)
(820, 0.6441)
(830, 0.6437)
(840, 0.6433)
(850, 0.6429)
(860, 0.6425)
(870, 0.6421)
(880, 0.6417)
(890, 0.6413)
(900, 0.6409)
(910, 0.6406)
(920, 0.6402)
(930, 0.6398)
(940, 0.6394)
(950, 0.6391)
(960, 0.6387)
(970, 0.6383)
(980, 0.6379)
(990, 0.6376)
};
\end{axis}\end{tikzpicture}
  \caption{Single Player Utility}
\label{fig:empirical_single}
\end{center}
\end{subfigure}
\begin{subfigure}{.5\linewidth}
\begin{center}
\begin{tikzpicture}[scale=0.45]\begin{axis}[height=11cm,width=18cm,ymin=.25,ymax=.5,xmin=0,xmax=1000]
\addplot [draw=blue, dotted, line width=0.3mm,tick pos=left] coordinates {
(0, 0.3461)
(10, 0.3365)
(20, 0.3282)
(30, 0.3232)
(40, 0.3197)
(50, 0.3173)
(60, 0.3157)
(70, 0.3149)
(80, 0.3146)
(90, 0.3141)
(100, 0.3134)
(110, 0.3129)
(120, 0.3126)
(130, 0.3122)
(140, 0.3120)
(150, 0.3120)
(160, 0.3116)
(170, 0.3113)
(180, 0.3112)
(190, 0.3109)
(200, 0.3110)
(210, 0.3109)
(220, 0.3108)
(230, 0.3109)
(240, 0.3109)
(250, 0.3106)
(260, 0.3106)
(270, 0.3105)
(280, 0.3105)
(290, 0.3106)
(300, 0.3105)
(310, 0.3105)
(320, 0.3105)
(330, 0.3104)
(340, 0.3105)
(350, 0.3104)
(360, 0.3103)
(370, 0.3095)
(380, 0.3094)
(390, 0.3087)
(400, 0.3087)
(410, 0.3082)
(420, 0.3081)
(430, 0.3080)
(440, 0.3079)
(450, 0.3079)
(460, 0.3078)
(470, 0.3078)
(480, 0.3076)
(490, 0.3076)
(500, 0.3076)
(510, 0.3077)
(520, 0.3076)
(530, 0.3075)
(540, 0.3074)
(550, 0.3073)
(560, 0.3073)
(570, 0.3070)
(580, 0.3070)
(590, 0.3069)
(600, 0.3068)
(610, 0.3068)
(620, 0.3068)
(630, 0.3067)
(640, 0.3065)
(650, 0.3065)
(660, 0.3065)
(670, 0.3065)
(680, 0.3064)
(690, 0.3063)
(700, 0.3062)
(710, 0.3061)
(720, 0.3059)
(730, 0.3058)
(740, 0.3058)
(750, 0.3056)
(760, 0.3056)
(770, 0.3056)
(780, 0.3055)
(790, 0.3054)
(800, 0.3053)
(810, 0.3053)
(820, 0.3051)
(830, 0.3051)
(840, 0.3049)
(850, 0.3047)
(860, 0.3046)
(870, 0.3046)
(880, 0.3045)
(890, 0.3044)
(900, 0.3043)
(910, 0.3043)
(920, 0.3042)
(930, 0.3041)
(940, 0.3040)
(950, 0.3038)
(960, 0.3038)
(970, 0.3037)
(980, 0.3035)
(990, 0.3034)
};
\addlegendentry{$\lambda_2 = 0$}
\addplot [draw=blue, densely dashdotted, line width=0.3mm] coordinates {
(0, 0.3924)
(10, 0.3904)
(20, 0.3888)
(30, 0.3858)
(40, 0.3823)
(50, 0.3766)
(60, 0.3697)
(70, 0.3620)
(80, 0.3531)
(90, 0.3444)
(100, 0.3361)
(110, 0.3282)
(120, 0.3221)
(130, 0.3176)
(140, 0.3138)
(150, 0.3110)
(160, 0.3091)
(170, 0.3077)
(180, 0.3065)
(190, 0.3061)
(200, 0.3055)
(210, 0.3049)
(220, 0.3043)
(230, 0.3038)
(240, 0.3037)
(250, 0.3030)
(260, 0.3026)
(270, 0.3023)
(280, 0.3020)
(290, 0.3015)
(300, 0.3014)
(310, 0.3010)
(320, 0.3010)
(330, 0.3010)
(340, 0.3009)
(350, 0.3008)
(360, 0.3004)
(370, 0.2999)
(380, 0.2993)
(390, 0.2993)
(400, 0.2992)
(410, 0.2991)
(420, 0.2987)
(430, 0.2987)
(440, 0.2986)
(450, 0.2987)
(460, 0.2987)
(470, 0.2985)
(480, 0.2983)
(490, 0.2983)
(500, 0.2982)
(510, 0.2983)
(520, 0.2979)
(530, 0.2980)
(540, 0.2979)
(550, 0.2978)
(560, 0.2976)
(570, 0.2976)
(580, 0.2975)
(590, 0.2975)
(600, 0.2974)
(610, 0.2971)
(620, 0.2969)
(630, 0.2968)
(640, 0.2967)
(650, 0.2966)
(660, 0.2965)
(670, 0.2965)
(680, 0.2964)
(690, 0.2964)
(700, 0.2963)
(710, 0.2961)
(720, 0.2960)
(730, 0.2961)
(740, 0.2959)
(750, 0.2957)
(760, 0.2957)
(770, 0.2956)
(780, 0.2956)
(790, 0.2955)
(800, 0.2955)
(810, 0.2953)
(820, 0.2951)
(830, 0.2952)
(840, 0.2950)
(850, 0.2949)
(860, 0.2948)
(870, 0.2946)
(880, 0.2945)
(890, 0.2944)
(900, 0.2943)
(910, 0.2942)
(920, 0.2941)
(930, 0.2941)
(940, 0.2941)
(950, 0.2940)
(960, 0.2940)
(970, 0.2939)
(980, 0.2937)
(990, 0.2936)
};
\addlegendentry{$\lambda_2 = 100$}
\addplot [draw=blue, line width=0.3mm] coordinates {
(0, 0.4079)
(10, 0.4064)
(20, 0.4050)
(30, 0.4039)
(40, 0.4023)
(50, 0.4011)
(60, 0.3999)
(70, 0.3989)
(80, 0.3978)
(90, 0.3967)
(100, 0.3953)
(110, 0.3940)
(120, 0.3924)
(130, 0.3911)
(140, 0.3896)
(150, 0.3875)
(160, 0.3848)
(170, 0.3809)
(180, 0.3772)
(190, 0.3724)
(200, 0.3669)
(210, 0.3603)
(220, 0.3539)
(230, 0.3465)
(240, 0.3391)
(250, 0.3320)
(260, 0.3257)
(270, 0.3197)
(280, 0.3143)
(290, 0.3101)
(300, 0.3067)
(310, 0.3036)
(320, 0.3012)
(330, 0.2987)
(340, 0.2973)
(350, 0.2960)
(360, 0.2950)
(370, 0.2937)
(380, 0.2923)
(390, 0.2915)
(400, 0.2909)
(410, 0.2898)
(420, 0.2894)
(430, 0.2886)
(440, 0.2881)
(450, 0.2876)
(460, 0.2873)
(470, 0.2870)
(480, 0.2867)
(490, 0.2864)
(500, 0.2863)
(510, 0.2862)
(520, 0.2859)
(530, 0.2857)
(540, 0.2856)
(550, 0.2852)
(560, 0.2849)
(570, 0.2847)
(580, 0.2842)
(590, 0.2841)
(600, 0.2839)
(610, 0.2837)
(620, 0.2836)
(630, 0.2836)
(640, 0.2835)
(650, 0.2833)
(660, 0.2831)
(670, 0.2831)
(680, 0.2829)
(690, 0.2827)
(700, 0.2827)
(710, 0.2827)
(720, 0.2826)
(730, 0.2825)
(740, 0.2824)
(750, 0.2823)
(760, 0.2818)
(770, 0.2820)
(780, 0.2819)
(790, 0.2818)
(800, 0.2816)
(810, 0.2815)
(820, 0.2813)
(830, 0.2812)
(840, 0.2810)
(850, 0.2810)
(860, 0.2809)
(870, 0.2806)
(880, 0.2806)
(890, 0.2804)
(900, 0.2803)
(910, 0.2803)
(920, 0.2803)
(930, 0.2801)
(940, 0.2800)
(950, 0.2799)
(960, 0.2799)
(970, 0.2798)
(980, 0.2795)
(990, 0.2794)
};
\addlegendentry{$\lambda_2 = 250$}
\addplot [draw=red, dotted, line width=0.3mm] coordinates {
(0, 0.4216)
(10, 0.4205)
(20, 0.4192)
(30, 0.4182)
(40, 0.4170)
(50, 0.4159)
(60, 0.4152)
(70, 0.4145)
(80, 0.4135)
(90, 0.4126)
(100, 0.4118)
(110, 0.4109)
(120, 0.4101)
(130, 0.4095)
(140, 0.4086)
(150, 0.4080)
(160, 0.4074)
(170, 0.4072)
(180, 0.4065)
(190, 0.4056)
(200, 0.4049)
(210, 0.4043)
(220, 0.4039)
(230, 0.4029)
(240, 0.4022)
(250, 0.4014)
(260, 0.4007)
(270, 0.3998)
(280, 0.3990)
(290, 0.3979)
(300, 0.3968)
(310, 0.3957)
(320, 0.3946)
(330, 0.3933)
(340, 0.3922)
(350, 0.3907)
(360, 0.3885)
(370, 0.3861)
(380, 0.3838)
(390, 0.3813)
(400, 0.3777)
(410, 0.3739)
(420, 0.3697)
(430, 0.3647)
(440, 0.3597)
(450, 0.3544)
(460, 0.3489)
(470, 0.3428)
(480, 0.3368)
(490, 0.3308)
(500, 0.3249)
(510, 0.3199)
(520, 0.3147)
(530, 0.3098)
(540, 0.3056)
(550, 0.3019)
(560, 0.2984)
(570, 0.2952)
(580, 0.2927)
(590, 0.2898)
(600, 0.2875)
(610, 0.2853)
(620, 0.2836)
(630, 0.2818)
(640, 0.2806)
(650, 0.2793)
(660, 0.2781)
(670, 0.2769)
(680, 0.2761)
(690, 0.2750)
(700, 0.2743)
(710, 0.2736)
(720, 0.2728)
(730, 0.2720)
(740, 0.2714)
(750, 0.2710)
(760, 0.2706)
(770, 0.2700)
(780, 0.2696)
(790, 0.2693)
(800, 0.2690)
(810, 0.2688)
(820, 0.2684)
(830, 0.2680)
(840, 0.2678)
(850, 0.2676)
(860, 0.2671)
(870, 0.2668)
(880, 0.2667)
(890, 0.2663)
(900, 0.2662)
(910, 0.2658)
(920, 0.2657)
(930, 0.2656)
(940, 0.2654)
(950, 0.2653)
(960, 0.2651)
(970, 0.2650)
(980, 0.2647)
(990, 0.2647)
};
\addlegendentry{$\lambda_2 = 500$}
\addplot [draw=red, line width=0.3mm] coordinates {
(0, 0.4375)
(10, 0.4367)
(20, 0.4356)
(30, 0.4347)
(40, 0.4340)
(50, 0.4332)
(60, 0.4323)
(70, 0.4313)
(80, 0.4304)
(90, 0.4297)
(100, 0.4290)
(110, 0.4284)
(120, 0.4279)
(130, 0.4272)
(140, 0.4267)
(150, 0.4260)
(160, 0.4255)
(170, 0.4249)
(180, 0.4242)
(190, 0.4237)
(200, 0.4231)
(210, 0.4226)
(220, 0.4223)
(230, 0.4219)
(240, 0.4217)
(250, 0.4212)
(260, 0.4207)
(270, 0.4202)
(280, 0.4196)
(290, 0.4192)
(300, 0.4187)
(310, 0.4182)
(320, 0.4177)
(330, 0.4175)
(340, 0.4171)
(350, 0.4166)
(360, 0.4161)
(370, 0.4156)
(380, 0.4150)
(390, 0.4144)
(400, 0.4140)
(410, 0.4134)
(420, 0.4130)
(430, 0.4125)
(440, 0.4120)
(450, 0.4115)
(460, 0.4111)
(470, 0.4108)
(480, 0.4103)
(490, 0.4101)
(500, 0.4096)
(510, 0.4090)
(520, 0.4085)
(530, 0.4081)
(540, 0.4076)
(550, 0.4070)
(560, 0.4064)
(570, 0.4060)
(580, 0.4054)
(590, 0.4046)
(600, 0.4037)
(610, 0.4031)
(620, 0.4025)
(630, 0.4018)
(640, 0.4012)
(650, 0.4003)
(660, 0.3996)
(670, 0.3988)
(680, 0.3979)
(690, 0.3966)
(700, 0.3954)
(710, 0.3945)
(720, 0.3931)
(730, 0.3922)
(740, 0.3908)
(750, 0.3898)
(760, 0.3883)
(770, 0.3867)
(780, 0.3853)
(790, 0.3837)
(800, 0.3819)
(810, 0.3802)
(820, 0.3782)
(830, 0.3764)
(840, 0.3738)
(850, 0.3715)
(860, 0.3688)
(870, 0.3664)
(880, 0.3632)
(890, 0.3601)
(900, 0.3569)
(910, 0.3533)
(920, 0.3494)
(930, 0.3457)
(940, 0.3417)
(950, 0.3385)
(960, 0.3351)
(970, 0.3309)
(980, 0.3268)
(990, 0.3225)
};
\addlegendentry{$\lambda_2 = 1000$}
\end{axis}\end{tikzpicture}
  \caption{Two Player Utility}
\label{fig:empirical_two}
\end{center}
\end{subfigure}
\label{fig:ridge-utilities}
\caption{
Results for the California housing prices data.
The player's utilities under ridge regression in the
  single-player and two-player environment.  In both environments
  Player 1's utility is plotted as a function of regularizer parameter
  $\ridge_1$.  In the two player figure, Player 1's utility is shown
  for a range of values of Player 2's regularizer parameter
  $\ridge_2$.  The single player utility is optimized at $\ridge_1 \approx 100 >
  0$; the two-player utility is optimized at $\ridge_1 = 0$ for all
  $\ridge_2$.}
\end{figure*}

For the red wine quality data,
Player~1's utility in the
single-player setting is depicted 
in \Cref{fig:empirical_single wine}
as a function of the player's regularization
parameter, ranging from $\ridge_1=0$ to $\ridge_1=2$.
The optimal
parameter choice is $\ridge_1\approx 0.26$.  In \Cref{fig:empirical_two wine},
Player~1's utility in the two-player game is plotted as a function of
$\ridge_1$ for fixed values of Player~2's regularization parameter
$\ridge_2$.  
For all choices of $\ridge_2$, Player~1's utility is
optimized by selecting $\ridge_1 = 0$.\footnote{
Similar results hold for 
the white wine quality data. In particular, 
Player~1's utility 
is optimized by selecting $\ridge_1 = 0$
in the two-player game.}

\begin{figure*}
\begin{subfigure}{.49\linewidth}
\begin{center}
\begin{tikzpicture}[scale=0.45]\begin{axis}[height=11cm,width=18cm,ymin=.567,ymax=.57,xmin=0,xmax=2,y tick label style={
        /pgf/number format/.cd,
        fixed,
        fixed zerofill,
        precision=4,
        /tikz/.cd
    },tick pos=left]
\addplot [draw = blue, line width=0.3mm] coordinates {
(0.01, 0.5691522960447066)
(0.02, 0.5691806066185896)
(0.03, 0.5692023314914836)
(0.04, 0.5692214244292393)
(0.05, 0.5692387185990803)
(0.060000000000000005, 0.5692545319091911)
(0.06999999999999999, 0.5692690336238644)
(0.08, 0.5692823354667642)
(0.09, 0.5692945219783618)
(0.09999999999999999, 0.5693056627328125)
(0.11, 0.5693158179734296)
(0.12, 0.56932504150605)
(0.13, 0.5693333823232857)
(0.14, 0.5693408855900725)
(0.15000000000000002, 0.5693475932846161)
(0.16, 0.569353544641906)
(0.17, 0.5693587764779715)
(0.18000000000000002, 0.5693633234385064)
(0.19, 0.5693672181973384)
(0.2, 0.5693704916202187)
(0.21000000000000002, 0.5693731729036562)
(0.22, 0.5693752896951785)
(0.23, 0.5693768681992397)
(0.24000000000000002, 0.5693779332717646)
(0.25, 0.5693785085054132)
(0.26, 0.5693786163071157)
(0.27, 0.5693782779690386)
(0.28, 0.5693775137338765)
(0.29000000000000004, 0.569376342855167)
(0.3, 0.5693747836532065)
(0.31, 0.5693728535670269)
(0.32, 0.5693705692028236)
(0.33, 0.5693679463791745)
(0.34, 0.5693650001693153)
(0.35000000000000003, 0.5693617449407318)
(0.36000000000000004, 0.5693581943922912)
(0.37, 0.5693543615890879)
(0.38, 0.5693502589951912)
(0.39, 0.5693458985044411)
(0.4, 0.5693412914694332)
(0.41000000000000003, 0.5693364487288208)
(0.42000000000000004, 0.5693313806330438)
(0.43, 0.5693260970686091)
(0.44, 0.5693206074809859)
(0.45, 0.5693149208962387)
(0.46, 0.5693090459414529)
(0.47000000000000003, 0.5693029908640541)
(0.48000000000000004, 0.5692967635500685)
(0.49, 0.5692903715413974)
(0.5, 0.5692838220521793)
(0.51, 0.5692771219842713)
(0.52, 0.5692702779419219)
(0.53, 0.5692632962456654)
(0.54, 0.5692561829455135)
(0.55, 0.5692489438334483)
(0.56, 0.56924158445528)
(0.5700000000000001, 0.5692341101219024)
(0.5800000000000001, 0.5692265259199789)
(0.59, 0.5692188367220913)
(0.6, 0.5692110471963734)
(0.61, 0.5692031618156783)
(0.62, 0.569195184866284)
(0.63, 0.5691871204561737)
(0.64, 0.5691789725229144)
(0.65, 0.5691707448411447)
(0.66, 0.5691624410297192)
(0.67, 0.5691540645584804)
(0.68, 0.5691456187547427)
(0.6900000000000001, 0.5691371068094379)
(0.7000000000000001, 0.569128531782983)
(0.7100000000000001, 0.5691198966108806)
(0.72, 0.569111204109044)
(0.73, 0.5691024569788746)
(0.74, 0.5690936578121156)
(0.75, 0.5690848090954731)
(0.76, 0.5690759132150276)
(0.77, 0.5690669724604444)
(0.78, 0.5690579890289947)
(0.79, 0.5690489650293904)
(0.8, 0.5690399024854547)
(0.81, 0.5690308033396161)
(0.8200000000000001, 0.5690216694562589)
(0.8300000000000001, 0.5690125026249178)
(0.8400000000000001, 0.569003304563336)
(0.85, 0.5689940769203818)
(0.86, 0.5689848212788507)
(0.87, 0.5689755391581259)
(0.88, 0.5689662320167423)
(0.89, 0.5689569012548332)
(0.9, 0.5689475482164585)
(0.91, 0.5689381741918627)
(0.92, 0.5689287804196046)
(0.93, 0.5689193680886131)
(0.9400000000000001, 0.5689099383401612)
(0.9500000000000001, 0.5689004922697396)
(0.9600000000000001, 0.568891030928869)
(0.97, 0.5688815553268199)
(0.98, 0.568872066432282)
(0.99, 0.5688625651749357)
(1.0, 0.5688530524469909)
(1.01, 0.5688435291046298)
(1.02, 0.5688339959694237)
(1.03, 0.5688244538296611)
(1.04, 0.568814903441641)
(1.05, 0.5688053455309117)
(1.06, 0.568795780793441)
(1.07, 0.5687862098967762)
(1.08, 0.5687766334811142)
(1.09, 0.568767052160366)
(1.1, 0.5687574665231478)
(1.11, 0.5687478771337554)
(1.12, 0.5687382845330942)
(1.1300000000000001, 0.5687286892395608)
(1.1400000000000001, 0.5687190917499055)
(1.1500000000000001, 0.56870949254005)
(1.1600000000000001, 0.5686998920658801)
(1.17, 0.5686902907640051)
(1.18, 0.5686806890524826)
(1.19, 0.568671087331523)
(1.2, 0.5686614859841604)
(1.21, 0.5686518853769045)
(1.22, 0.5686422858603588)
(1.23, 0.5686326877698207)
(1.24, 0.5686230914258558)
(1.25, 0.5686134971348576)
(1.26, 0.5686039051895742)
(1.27, 0.5685943158696185)
(1.28, 0.568584729441967)
(1.29, 0.5685751461614357)
(1.3, 0.5685655662711244)
(1.31, 0.5685559900028678)
(1.32, 0.5685464175776563)
(1.33, 0.5685368492060342)
(1.34, 0.5685272850885057)
(1.35, 0.5685177254158948)
(1.36, 0.5685081703697217)
(1.37, 0.5684986201225473)
(1.3800000000000001, 0.5684890748383081)
(1.3900000000000001, 0.5684795346726413)
(1.4000000000000001, 0.5684699997731985)
(1.4100000000000001, 0.5684604702799454)
(1.42, 0.5684509463254471)
(1.43, 0.5684414280351524)
(1.44, 0.5684319155276607)
(1.45, 0.5684224089149804)
(1.46, 0.5684129083027749)
(1.47, 0.5684034137906122)
(1.48, 0.5683939254721813)
(1.49, 0.5683844434355302)
(1.5, 0.5683749677632748)
(1.51, 0.5683654985327993)
(1.52, 0.5683560358164716)
(1.53, 0.5683465796818167)
(1.54, 0.5683371301917168)
(1.55, 0.5683276874045854)
(1.56, 0.5683182513745374)
(1.57, 0.5683088221515569)
(1.58, 0.5682993997816578)
(1.59, 0.5682899843070383)
(1.6, 0.5682805757662281)
(1.61, 0.5682711741942359)
(1.62, 0.568261779622681)
(1.6300000000000001, 0.5682523920799318)
(1.6400000000000001, 0.568243011591239)
(1.6500000000000001, 0.5682336381788449)
(1.6600000000000001, 0.5682242718621163)
(1.6700000000000002, 0.5682149126576612)
(1.68, 0.5682055605794277)
(1.69, 0.5681962156388188)
(1.7, 0.5681868778447995)
(1.71, 0.568177547203991)
(1.72, 0.5681682237207655)
(1.73, 0.5681589073973465)
(1.74, 0.568149598233893)
(1.75, 0.5681402962285842)
(1.76, 0.568131001377709)
(1.77, 0.5681217136757388)
(1.78, 0.568112433115415)
(1.79, 0.56810315968781)
(1.8, 0.5680938933824129)
(1.81, 0.5680846341871949)
(1.82, 0.5680753820886653)
(1.83, 0.5680661370719531)
(1.84, 0.5680568991208597)
(1.85, 0.5680476682179241)
(1.86, 0.5680384443444695)
(1.87, 0.5680292274806774)
(1.8800000000000001, 0.5680200176056263)
(1.8900000000000001, 0.5680108146973482)
(1.9000000000000001, 0.5680016187328836)
(1.9100000000000001, 0.5679924296883248)
(1.9200000000000002, 0.5679832475388611)
(1.93, 0.5679740722588302)
(1.94, 0.5679649038217559)
(1.95, 0.5679557422003898)
(1.96, 0.5679465873667572)
(1.97, 0.5679374392921883)
(1.98, 0.5679282979473655)
(1.99, 0.5679191633023459)
};
\end{axis}\end{tikzpicture}
  \caption{Single Player Utility}
\label{fig:empirical_single wine}
\end{center}
\end{subfigure}
\begin{subfigure}{.49\linewidth}
\begin{center}
\begin{tikzpicture}[scale=0.45]\begin{axis}[height=11cm,width=18cm,ymin=0.3,ymax=0.31,xmin=0,xmax=2,tick pos=left,y tick label style={
        /pgf/number format/.cd,
        fixed,
        fixed zerofill,
        precision=3,
        /tikz/.cd
    }]
\addplot [draw=blue, dotted, line width=0.3mm] coordinates {
(0.01, 0.30336975201722786)
(0.060000000000000005, 0.30312823265552413)
(0.11, 0.30278480437292914)
(0.16000000000000003, 0.30253700600864852)
(0.21000000000000002, 0.30229590735952522)
(0.26, 0.30208697541915554)
(0.31000000000000005, 0.30195044047547553)
(0.36000000000000004, 0.30183502034783337)
(0.41000000000000003, 0.3017404407326374)
(0.46, 0.3016142159159017)
(0.51, 0.3015402559837167)
(0.56, 0.30156092296419447)
(0.6100000000000001, 0.30154664093371596)
(0.66, 0.3016132680722063)
(0.7100000000000001, 0.30154193649277756)
(0.76, 0.30151611971900964)
(0.81, 0.30151863361149257)
(0.8600000000000001, 0.30145751303596036)
(0.91, 0.30147616127964304)
(0.9600000000000001, 0.3014872206357936)
(1.01, 0.3014537756541133)
(1.06, 0.30142194520607735)
(1.11, 0.30149780979465655)
(1.1600000000000001, 0.30143992775543905)
(1.2100000000000002, 0.3014952413062964)
(1.26, 0.3014377476291395)
(1.31, 0.3015763266439391)
(1.36, 0.30165024775498768)
(1.4100000000000001, 0.3016784506101746)
(1.4600000000000002, 0.30168309809473708)
(1.51, 0.30170207037679503)
(1.56, 0.3017326205398705)
(1.61, 0.3017553105900348)
(1.6600000000000001, 0.30173511613310444)
(1.7100000000000002, 0.3017342483741887)
(1.76, 0.30175762283368206)
(1.81, 0.30175574683127128)
(1.86, 0.3017498711768652)
(1.9100000000000001, 0.30176600827880075)
(1.9600000000000002, 0.30177139061282443)};
\addlegendentry{$\lambda_2 = 0$}
\addplot [draw=blue, densely dashdotted, line width=0.3mm] coordinates {
(0.01, 0.30472940788392272)
(0.060000000000000005, 0.30433636066385404)
(0.11, 0.30394068691187732)
(0.16000000000000003, 0.30345760974795816)
(0.21000000000000002, 0.30310136586456128)
(0.26, 0.3028405087188405)
(0.31000000000000005, 0.30265297373659072)
(0.36000000000000004, 0.3025981946132425)
(0.41000000000000003, 0.30240927741473838)
(0.46, 0.30229039097249537)
(0.51, 0.30221636788041895)
(0.56, 0.30199866553891813)
(0.6100000000000001, 0.3018687213263871)
(0.66, 0.30185687486752943)
(0.7100000000000001, 0.30177259484089475)
(0.76, 0.3017076764223397)
(0.81, 0.30161641813148815)
(0.8600000000000001, 0.30157747137350015)
(0.91, 0.30158439745187605)
(0.9600000000000001, 0.3015716360421778)
(1.01, 0.30155454012047157)
(1.06, 0.3015155084428688)
(1.11, 0.3015722453503168)
(1.1600000000000001, 0.3015016317843256)
(1.2100000000000002, 0.3014823236706733)
(1.26, 0.30145243283343705)
(1.31, 0.30140165323280232)
(1.36, 0.30141261598731633)
(1.4100000000000001, 0.301462150481329)
(1.4600000000000002, 0.3014479849686838)
(1.51, 0.30143700027092896)
(1.56, 0.30142506857124256)
(1.61, 0.301463844422074)
(1.6600000000000001, 0.30145549284888405)
(1.7100000000000002, 0.3014925841648138)
(1.76, 0.30150396853473798)
(1.81, 0.30143759162602805)
(1.86, 0.30143782208538854)
(1.9100000000000001, 0.3014993920059377)
(1.9600000000000002, 0.3015596395240854)
};
\addlegendentry{$\lambda_2 = 0.25$}
\addplot [draw=blue, line width=0.3mm] coordinates {
(0.01, 0.30545461060299612)
(0.060000000000000005, 0.30507635082008395)
(0.11, 0.3046818935252263)
(0.16000000000000003, 0.3043619532732242)
(0.21000000000000002, 0.30403074267769996)
(0.26, 0.30379656195020022)
(0.31000000000000005, 0.3034877959593456)
(0.36000000000000004, 0.303222078600195)
(0.41000000000000003, 0.30298888396800883)
(0.46, 0.3028873535525007)
(0.51, 0.30274157779726446)
(0.56, 0.30258177872938375)
(0.6100000000000001, 0.30248517252583935)
(0.66, 0.30234741419311963)
(0.7100000000000001, 0.30227691671983036)
(0.76, 0.30208438169353435)
(0.81, 0.30199862140200977)
(0.8600000000000001, 0.301930831016333)
(0.91, 0.30180668755687024)
(0.9600000000000001, 0.30181596932249144)
(1.01, 0.3016426017445927)
(1.06, 0.30152229540581017)
(1.11, 0.30151096260105436)
(1.1600000000000001, 0.30154106356417205)
(1.2100000000000002, 0.3015138539906609)
(1.26, 0.30148547023753997)
(1.31, 0.3014318602225414)
(1.36, 0.3014027219052794)
(1.4100000000000001, 0.30138589895381603)
(1.4600000000000002, 0.30134731065443243)
(1.51, 0.3013643272703002)
(1.56, 0.3014292295786421)
(1.61, 0.30139154362308147)
(1.6600000000000001, 0.30144122683712072)
(1.7100000000000002, 0.3014695771208683)
(1.76, 0.3014009858420934)
(1.81, 0.30145077034841716)
(1.86, 0.30150888492468504)
(1.9100000000000001, 0.3015166004147768)
(1.9600000000000002, 0.3015151106107153)
};
\addlegendentry{$\lambda_2 = 0.5$}
\addplot [draw=red, dotted, line width=0.3mm] coordinates {
(0.01, 0.30676436547004103)
(0.060000000000000005, 0.30640593762932062)
(0.11, 0.305936383131995)
(0.16000000000000003, 0.305473616894489)
(0.21000000000000002, 0.30503733425315066)
(0.26, 0.30475604755230767)
(0.31000000000000005, 0.30447889737899903)
(0.36000000000000004, 0.30421837689454705)
(0.41000000000000003, 0.3039769389220251)
(0.46, 0.30374487924025757)
(0.51, 0.30346217039528493)
(0.56, 0.30329847863590602)
(0.6100000000000001, 0.30322515128075767)
(0.66, 0.30304512453385148)
(0.7100000000000001, 0.30290057935345817)
(0.76, 0.30273362333531468)
(0.81, 0.30255628966242915)
(0.8600000000000001, 0.30234647278603078)
(0.91, 0.30225985425466357)
(0.9600000000000001, 0.30218213617100594)
(1.01, 0.3020933841564396)
(1.06, 0.30203836887913477)
(1.11, 0.3019100572378496)
(1.1600000000000001, 0.3018292902357118)
(1.2100000000000002, 0.30179101114388575)
(1.26, 0.30167762831734113)
(1.31, 0.30158020035676948)
(1.36, 0.3014577690687071)
(1.4100000000000001, 0.30143097145806718)
(1.4600000000000002, 0.30140478470407857)
(1.51, 0.30141898863591237)
(1.56, 0.3014245045792993)
(1.61, 0.3013865348415336)
(1.6600000000000001, 0.30137678285540776)
(1.7100000000000002, 0.30127972041158736)
(1.76, 0.3011954090879342)
(1.81, 0.30120945905776714)
(1.86, 0.30113291408360287)
(1.9100000000000001, 0.30117705146126533)
(1.9600000000000002, 0.30113561777704627)
};
\addlegendentry{$\lambda_2 = 1$}
\addplot [draw=red, line width=0.3mm] coordinates 
{(0.01, 0.30790371204849842)
(0.060000000000000005, 0.30744872952589136)
(0.11, 0.30689692089747526)
(0.16000000000000003, 0.3065663139575379)
(0.21000000000000002, 0.30616721899341844)
(0.26, 0.3058789821725556)
(0.31000000000000005, 0.30561371104549255)
(0.36000000000000004, 0.3052601926235356)
(0.41000000000000003, 0.30498337210871286)
(0.46, 0.30460765436257274)
(0.51, 0.30440880658441177)
(0.56, 0.3042441587644954)
(0.6100000000000001, 0.3039823239690819)
(0.66, 0.30379281261201728)
(0.7100000000000001, 0.30362903353405614)
(0.76, 0.30348713889816026)
(0.81, 0.3033248086994922)
(0.8600000000000001, 0.30318719974109)
(0.91, 0.30301875400297253)
(0.9600000000000001, 0.30286799940392933)
(1.01, 0.3027657256516606)
(1.06, 0.30270751821003916)
(1.11, 0.3026441524296536)
(1.1600000000000001, 0.3025616343924763)
(1.2100000000000002, 0.30243688140105234)
(1.26, 0.30233172007116466)
(1.31, 0.30226755040209864)
(1.36, 0.302138808366203)
(1.4100000000000001, 0.3019773676543005)
(1.4600000000000002, 0.3018989133286164)
(1.51, 0.30177649016071683)
(1.56, 0.30177946579702845)
(1.61, 0.3017121426846358)
(1.6600000000000001, 0.3016340598560887)
(1.7100000000000002, 0.30147312089667987)
(1.76, 0.30141381442388933)
(1.81, 0.3013903095421532)
(1.86, 0.30131876246488203)
(1.9100000000000001, 0.30117713729336033)
(1.9600000000000002, 0.30110116930338514)
};
\addlegendentry{$\lambda_2 = 2$}
\end{axis}\end{tikzpicture}
  \caption{Two Player Utility}
\label{fig:empirical_two wine}
\end{center}
\end{subfigure}
\label{fig:ridge-utilities wine}
\caption{
Results for the red wine quality data.
The player's utilities under ridge regression in the
  single-player and two-player environment.  In both environments
  Player 1's utility is plotted as a function of regularizer parameter
  $\ridge_1$.  In the two player figure, Player 1's utility is shown
  for a range of values of Player 2's regularizer parameter
  $\ridge_2$.  The single player utility is optimized at $\ridge_1 \approx 0.26 >
  0$; the two-player utility is optimized at $\ridge_1 = 0$ for all
  $\ridge_2$.}
\end{figure*}

The conclusions of our
theoretical study are corroborated by this empirical ridge regression game.
Specifically, the optimal
single-player regularization parameter in a ridge regression is generally non-zero, while, as
long as the benefit of winning is sufficiently large, the two-player
best response is to lower the regularization parameter to zero.
In \Cref{apx:empirical-validation},
we also demonstrate that two of the assumptions used in Theorem~\ref{thm:general-framework-non-constant-tradeoff} for the ex ante game
hold in our empirical ridge regression 
game.
We view these results as affirming the validity of our insight on the preference of variance over bias
in competitive settings beyond our theoretical and numerical analyses.


\section{Conclusions}
\label{sec:conclusions}

In this paper we studied competing machine learning algorithms by abstracting
the problem to a distribution-selection game in which bias and
variance can be traded off.  While outcomes of real learning
algorithms can be complex, our bias-variance game is amenable to
theoretical analysis, and we formally prove that for normal
distributions reducing bias at the expense of variance is a
best response.  Thus, there is a clear preference for lower bias, even at the expense of higher total error.

In addition to our theoretical analyses, we verified the robustness of our insight on the
preference for lower bias numerically under several variations of the game. We also considered the empirical game on a
benchmark data set using ridge regression, where the same qualitative
conclusions from our theoretical analysis were demonstrated.


Many aspects of machine learning problems change
significantly in competitive environments. For example, in a different context but a similar vein, 
\citet{mansour2018competing} consider the classical bandit model of
online learning and study the effect competition between the
algorithms on the exploration vs.\ exploitation tradeoff.  These authors show
that the presence of competition may lead to the strategic choice of
algorithms that do not explore as much as they would absent
competition, and may thus be worse learners. More generally,
we believe that there are many more
unresolved issues in the intersection of machine learning
and competition, and suggest this general area as a
fruitful and important one for future study.

\bibliographystyle{apalike}
\bibliography{biasvariance-bib}

\appendix
\newpage
\begin{Large}
\begin{center}
    \textbf{Appendix}
\end{center}
\end{Large}
\ifitcssubm
\input{itcs-appx-no-tradeoff}
\else
\section{Proofs from 
\ifMS
\Cref{sec:no-tradeoff}
\else
Section 3
\fi
}
\label{app:no-tradeoff}
\begin{numberedtheorem}{
\ref{thm:lower-bias}}
Let $Z$ be monotone increasing, convex on $[-\infty,0]$, and symmetric
around 0. Let $X_i=\sigma Z+\mu$ be IR (so as to satisfy
$\mu^2+\sigma^2\leq 1$) and $X_i'=\sigma Z$. Then $u_i(X'_i,c)\geq
u_i(X_i,c)$ for any realization $c$ of player $j$.
\end{numberedtheorem}

\ifMS \proof{Proof.}
\else \begin{proof}\fi
The proof consists of several cases.
\begin{itemize}
\item[1)] $\mu \geq c$: Since $\mu^2+\sigma^2\leq 1$, it must be the case that $c\leq 1$. Thus, for any realization in which player $i$ gets non-zero utility, his utility
is nonnegative under both $X_i$ and $X'_i$. 

Consider first the distribution $X_i''=\sigma Z+c$. Observe that, by monotonicity, for each point $x\in[-c,c]$, the pdf at $x$ under $X_i''$ is higher than under $X_i$.
Since all such realizations lead to positive utility, $u(X_i'')\geq u(X_i)$.

Next, consider the comparison between $X_i'$ and $X_i''$. On the interval $[0,c]$ the distribution $X_i'$ is an inversion of $X_i''$ with higher probability closer to the
origin, and so on this sub-interval $X'$ leads to higher utility. On the interval $[-c,0]$ the pdf of $X_i'$ dominates that of $X_i''$, and so also here $X_i'$ leads to higher
utility. Thus, $u(X_i')\geq u(X_i'')$, and so $u(X_i')\geq u(X_i)$.

\item[2a)] $\mu<c\leq 1$: Consider \Cref{fig:lower-bias1}, in which $X'$ is the green pdf and $X$ is the blue pdf.
Area E (from $-1$ to $1$, and below both curves) leads to the same utility for both distributions. Area A (under $X'$) leads to higher utility than area B (under $X$).
And finally, area D (under $X'$) leads to strictly positive utility. Thus, overall, $u(X')\geq u(X)$.

\begin{figure}
\begin{subfigure}{.4\textwidth}
\begin{center}
\includegraphics[height=4cm]{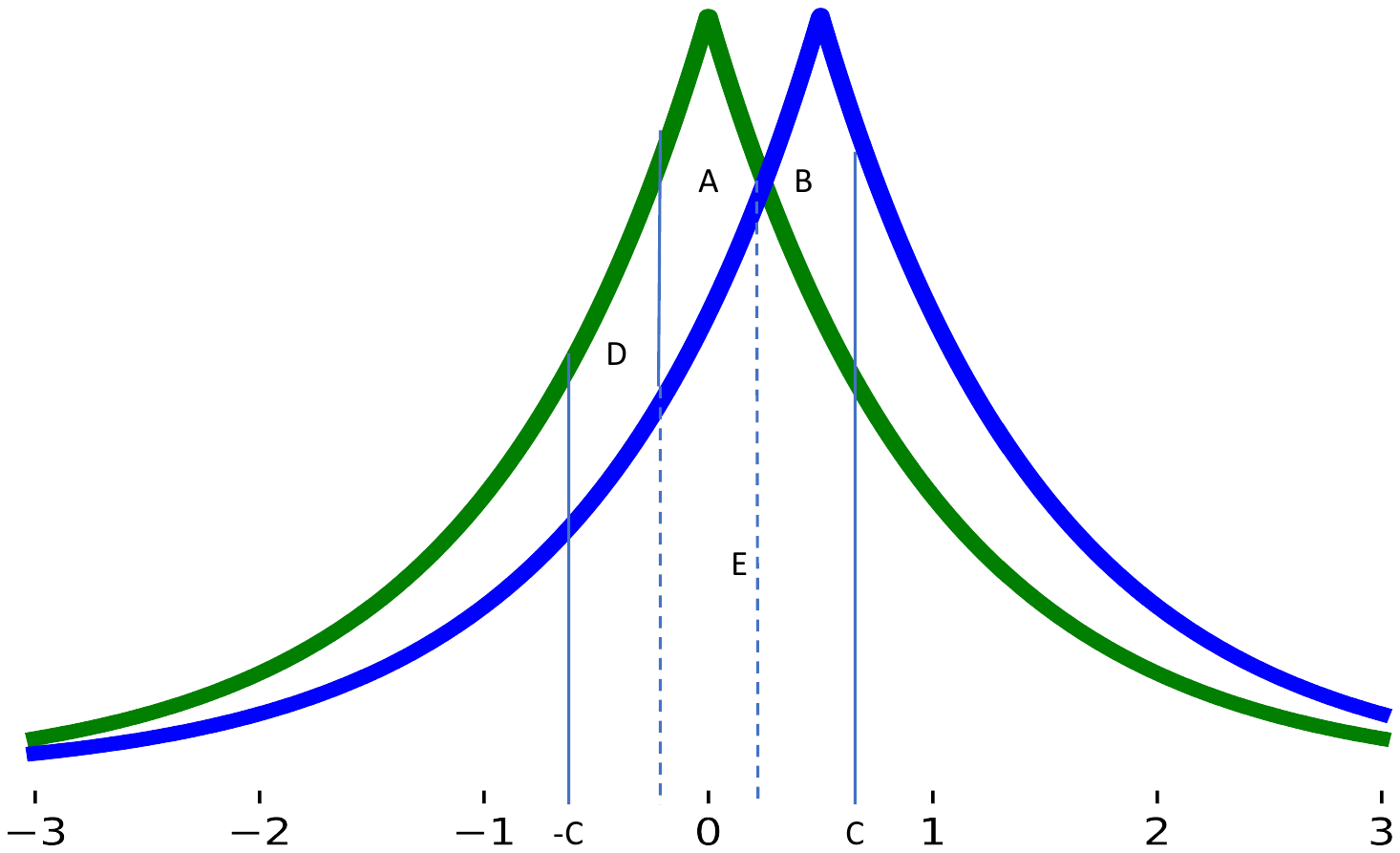}
  \caption{Case 2a}
\label{fig:lower-bias1}
\end{center}
\end{subfigure}
\begin{subfigure}{.4\textwidth}
\begin{center}
\includegraphics[height=4cm]{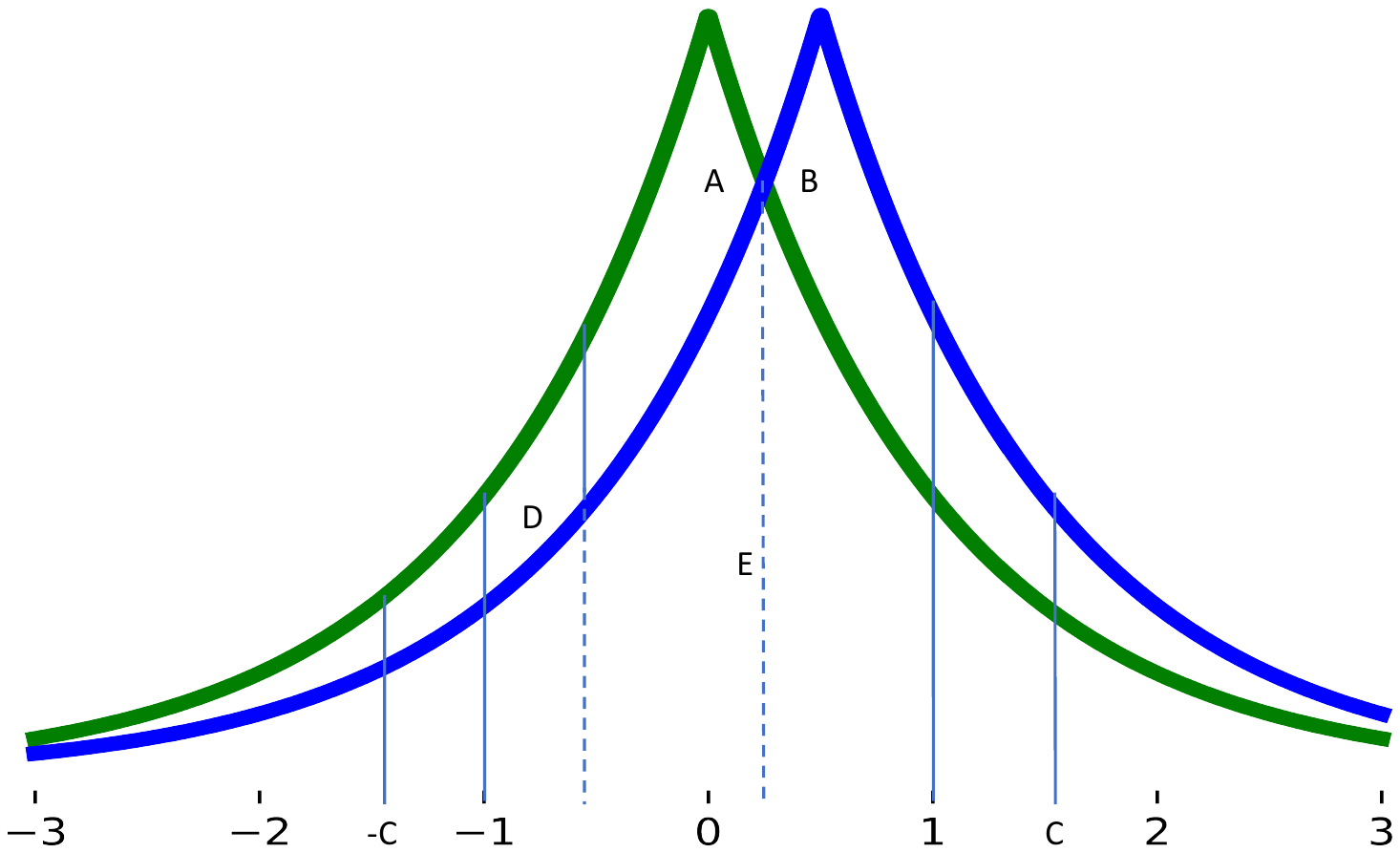}
  \caption{Case 2b}
\label{fig:lower-bias2}
\end{center}
\end{subfigure}
\label{fig:lower-bias}
\caption{Case 2 of 
Theorem 2.1.
}
\end{figure}

\item[2b)] $c>1$: Consider \Cref{fig:lower-bias2}, in which $X'$ is the green pdf and $X$ is the blue pdf.
Area E (from $-c$ to $c$, and below both curves) leads to the same utility for both distributions. Area A (under $X'$) leads to higher utility than area B (under $X$).
Area D (under $X'$) leads to strictly positive utility. 

It remains to show that the losses under $X'$ due to realizations in $[-c,-1)\cup (1,c]$ are smaller than the 
losses on the same intervals due to $X$. To this end, we will consider points $x\in(1,c]$, and show that the sum of the pdfs of $X$ at $x$ and $-x$ is larger
than the sum of the pdfs of $X'$ at those same points. Let $g$ be the pdf of $X'$. Then the sum of the pdfs of $X'$ at points $x$ and $-x$ is $g(x)+g(-x) = 2g(x)$.
The sum of the pdfs of $X$ at the points $x$ and and $-x$ is $g(x+\mu)+g(x-\mu)$.
By convexity, $g(x+\mu)+g(x-\mu) \geq 2g(x)$, completing the claim.\ifMS\Halmos\fi 
\end{itemize}
\ifMS \endproof
\else \end{proof} \fi

\begin{numberedtheorem}{
\ref{thm:lower-bias normal}}
Let $Z$ be normal; let $X_i=\sigma Z+\mu$ and satisfy $\mu+\sigma\leq 1$; and let $X_i'=\sigma Z$. Then $u_i(X'_i,c)\geq u_i(X_i,c)$ for any realization $c$ of player $j$.
\end{numberedtheorem}

\ifMS \proof{Proof.}
\else \begin{proof}\fi
The proof is nearly identical to that of 
\Cref{thm:lower-bias}, 
except for case 2b: the upper tail of the normal distribution is convex only from $\mu+\sigma$ onward. So this case must be handled differently.

However, to complete the proof, we actually only need that the pdf $g$ of $X'$ be convex from the point $1-\mu$ and higher. Since $\mu+\sigma \leq 1$ it holds that $1-\mu\geq \sigma$, and so $g$ is convex on $[c-\mu,c+\mu]$ whenever $c\geq 1$, completing the proof.
\ifMS \Halmos \endproof
\else \end{proof} \fi
\fi

\section{Proof of 
\ifMS
  \Cref{thm:lower-bias normal under tradeoff}
  \else
  Theorem 4.1
  \fi
  }
\label{apx:critical inequality}
\begin{numberedtheorem}{\ref{thm:lower-bias normal under tradeoff}}
Let $Z$ be normal with mean 0 and variance 1, and let $X_i=\sigma Z+\mu$ 
and $X_i'=\tau Z + \nu$, where $\mu^2+\sigma^2=\nu^2+\tau^2=1$.
If $\nu^2<\mu^2$ then $u_i(X'_i,a)> u_i(X_i,a)$ for any realization $a>0$ of player $j$.
\end{numberedtheorem}

\ifMS \proof{Proof.}
\else \begin{proof} \fi
Since players are symmetric,
we drop all subscripts in the discussion below.
We 
compute the expected 
utility $u(X, a)$ of player $i$
when he plays action
$X=\sigma Z+\mu$
against 
the realization $a$
of player $j$.
Without loss of generality,
we assume 
$a \geq 0$.
First we characterize the closed form of 
$\expect{u(X, a)}$ for all $\mu \in [0, 1)$;
then we argue that
$\lim_{\mu\rightarrow 1^-}\expect{u(X, a)} = 0
=
\expect{u(X',a)}$, where $X' = 0\cdot Z + 1$;
and finally we show that
$\frac{\partial \expect{u(X, a)}}{\partial \mu} \leq 0$
for all $\mu \in [0, 1)$, which completes the proof.

We first focus on $\mu\in [0, 1)$. Observe that 
\begin{align}
\expect{u(X, a)} &=  \expect{(1 - X^2)\cdot {\bf 1}\{|X|<a\}}
=
\prob{|X|< a}
-
\nonumber
\expect{X^2\cdot {\bf 1}\{|X|<a\}}.\\
\intertext{We can now evaluate}
\nonumber
\prob{|X|< a}
&=\int^{a}_{-a}\frac{1}{\sqrt{2\pi \sigma^2}}
\exp\left(\frac{-(x-\mu)^2}{2\sigma^2}\right)dx\\
\nonumber
&=\Phi\left(
\frac{a-\mu}{\sigma}
\right)-
\Phi\left(
\frac{-a-\mu}{\sigma}
\right),\\
\intertext{where $\Phi(\cdot)$
is the CDF of $Z$, i.e., the standard normal distribution. Now,}
\nonumber
\expect{X^2\cdot{\bf 1}\{|X|<a\}}
&=
\int^{a}_{-a}x^2\,\frac{1}{\sqrt{2\pi \sigma^2}}
\exp\left(\frac{-(x-\mu)^2}{2\sigma^2}\right)dx\\
\nonumber
&=
-\frac{\mu\sigma }{\sqrt{2\pi}}
\left(
\exp\left(\frac{-(a-\mu)^2}{2\sigma^2}\right)-
\exp\left(\frac{-(a+\mu)^2}{2\sigma^2}\right)
\right)\\
\nonumber
&\quad-\frac{a\sigma}{\sqrt{2\pi}}
\left(
\exp\left(\frac{-(a-\mu)^2}{2\sigma^2}\right)+
\exp\left(\frac{-(a+\mu)^2}{2\sigma^2}\right)
\right)\\
\nonumber
&\quad+(\sigma^2+\mu^2)\left(
\Phi\left(
\frac{a-\mu}{\sigma}
\right)-\Phi\left(
\frac{-a-\mu}{\sigma}
\right)
\right).
\\
\intertext{Because of the constraint that $\mu^2+\sigma^2=1$, we can eliminate some terms:}
\nonumber
\expect{u(X, a)}
&=
\frac{\mu\sqrt{1-\mu^2}}{\sqrt{2\pi}}
\left(
\exp\left(\frac{-(a-\mu)^2}{2-2\mu^2}\right)-
\exp\left(\frac{-(a+\mu)^2}{2-2\mu^2}\right)
\right)\\
\nonumber
&\quad+\frac{a\sqrt{1-\mu^2}}{\sqrt{2\pi}}
\left(
\exp\left(\frac{-(a-\mu)^2}{2-2\mu^2}\right)+
\exp\left(\frac{-(a+\mu)^2}{2-2\mu^2}\right)
\right).
\intertext{Next we consider $\expect{u(X, a)}$ for $\mu = 1$.
Note that in this case, agent $i$'s realization 
is 1 deterministically,
which implies that
her utility conditioning on winning is zero.
Her utility conditioning on losing is also zero by definition.
Thus, $\expect{u(X, a)} = 0$ for $\mu = 1$.
To see that $\lim_{\mu\rightarrow 1^-}\expect{u(X, a)} = 0$,
note that 
$0< \exp\left(\frac{-(a-\mu)^2}{2-2\mu^2}\right)\leq 1$
and 
$0< \exp\left(\frac{-(a+\mu)^2}{2-2\mu^2}\right)\leq 1$.
Thus,
}
\nonumber
    \lim\limits_{\mu\rightarrow 1^-}\expect{u(X, a)}
    &\leq
\lim\limits_{\mu\rightarrow 1^-}
\left(\frac{\mu\sqrt{1-\mu^2}}{\sqrt{2\pi}}
+2\frac{a\sqrt{1-\mu^2}}{\sqrt{2\pi}}
\right) = 0
\intertext{and} 
\nonumber
    \lim\limits_{\mu\rightarrow 1^-}\expect{u(X, a)}
&\geq
\lim\limits_{\mu\rightarrow 1^-}
-\frac{\mu\sqrt{1-\mu^2}}{\sqrt{2\pi}}
= 0
\end{align}
\noindent Invoking the Squeeze Theorem
yields $\lim_{\mu\rightarrow 1^-}\expect{u(X, a)} = 0$.  Finally, taking the derivative of 
$\expect{u(X, a)}$ with respect to $\mu$ 
for $\mu \in [0, 1)$ yields
\begin{align}
\nonumber
&\frac{\partial \expect{u(X, a)}}{\partial \mu}\\
\nonumber
=&\left[
\sqrt{1 - \mu^2}
-
(a + \mu)
\frac{\mu}{\sqrt{1 - \mu^2}} +
\frac{(a^2 - \mu^2)(1 - a\mu)\sqrt{1-\mu^2}}
{(1 - \mu^2)^2}
\right]
\exp\left(-\frac{1}{2}\frac{(a - \mu)^2}{1-\mu^2}\right) \\
\nonumber
&-
\left[
\sqrt{1 - \mu^2}
+
(a - \mu)
\frac{\mu}{\sqrt{1 - \mu^2}}+
\frac{(a^2 - \mu^2)(1 + a\mu)\sqrt{1-\mu^2}}
{(1 - \mu^2)^2}
\right]
\exp\left(-\frac{1}{2}\frac{(a + \mu)^2}{1-\mu^2}\right). 
\end{align}
To prove 
\Cref{thm:lower-bias normal under tradeoff},
it is sufficient to show that this 
derivative is strictly negative
for all $\mu \in (0, 1)$
and $a > 0$. This condition is expressed as
\begin{align}
\begin{split}\label{eq:critical inequality}
 [(2\mu^2 - a^2 - 1)(1 - a\mu) + 
        2\mu^2(1-\mu^2)]
        \exp\left(
        \frac{a\mu}{1-\mu^2}
        \right)&\\
        -~
        [(2\mu^2 - a^2 - 1)(1 + a\mu) + 
        2\mu^2(1-\mu^2)]
        \exp\left(
        \frac{-a\mu}{1-\mu^2}
        \right)
        &> 0.
\end{split}
\end{align}
The remaining part of the proof, showing that inequality \eqref{eq:critical
  inequality} holds, follows from a long and algebraic calculation
that we formalize as 
\ifMS
Lemma 1 
in Online Appendix 1.
\else
\Cref{lem:critical inequality}.
  \fi
\ifMS \Halmos \endproof
\else \end{proof} \fi

\ifEC
We now  show that the derivative 
of the ex post utility 
with respect to $\mu$
against realization $a$ of 
the opponent  is strictly negative 
for all $\mu \in (0, 1)$ and 
$a >0$.
\begin{lemma}\label{lem:critical inequality}
Inequality \eqref{eq:critical inequality}
holds for all $\mu \in (0, 1)$ and 
$a > 0$.
\end{lemma}


Notice that there may be two regimes:
(a) the player always gains positive payoff,
i.e., $a < 1$;
(b) the player sometimes suffers non-positive payoff,
i.e., $a \geq 1$.
We break Lemma
\ref{lem:critical inequality}
into these two regimes and show them 
separately.

\newcommand{\geqs}{>}

\begin{lemma}
\label{lem:critical inequality small a}
Inequality \eqref{eq:critical inequality}
holds for all $\mu \in (0, 1)$ and 
$a \in (0, 1)$.
\end{lemma}
\ifMS \proof{Proof.}
\else \begin{proof}\fi
We start with  
inequality~\eqref{eq:critical inequality}, copied here:
\begin{align*}
 &[(2\mu^2 - a^2 - 1)(1 - a\mu) + 
        2\mu^2(1-\mu^2)]
        \exp\left(
        \frac{a\mu}{1-\mu^2}
        \right)\\
        \geqs~
        &[(2\mu^2 - a^2 - 1)(1 + a\mu) + 
        2\mu^2(1-\mu^2)]
        \exp\left(
        \frac{-a\mu}{1-\mu^2}
        \right).
\end{align*}
The proof is a two-case analysis, working backwards from inequality~\eqref{eq:critical inequality} to show that it holds from lemma domain and case assumptions in both cases.  Some steps are exact algebraic manipulations (``if and only if" $\Leftrightarrow$ or $\Updownarrow$) and some steps are weakly restrictions to stronger requirements (``is implied by" $\Leftarrow$ or $\Uparrow$).  We use the arrows for visual simplicity to indicate the type of each step.  
Each step includes an explanation.  Consider the following two cases based on the sign of $(2\mu^2-a^2-1)$.
\begin{align}
\intertext{\textbf{Case 1:} $(2\mu^2-a^2-1)> 0$.
To start, multiplying both sides of equation~\eqref{eq:critical inequality} by $\frac{1}{2\mu^2-a^2-1}\cdot\exp\left(\frac{1}{1 - \mu^2}\right)$, we get inequality~$\eqref{eq:critical inequality} \Leftrightarrow$}
\nonumber\begin{split}
&\left[
(1-a\mu)+\frac{(2\mu^2(1-\mu^2))}{(2\mu^2-a^2-1)}
\right]
\exp\left(\frac{1+a\mu}{1-\mu^2}\right) 
\\
\geqs
&\left[
(1+a\mu)+\frac{(2\mu^2(1-\mu^2))}{(2\mu^2-a^2-1)}
\right]
\exp\left(\frac{1-a\mu}{1-\mu^2}\right) 
	\end{split}\\
\intertext{$(\Uparrow)$~Now it is clear that we can drop the term $(2\mu^2-a^2-1)$ because we have $0<(2\mu^2-a^2-1)< 1$ from $\mu^2< 1$.  So we get
}
	\nonumber
\nonumber\begin{split}
&\left[
(1-a\mu)+(2\mu^2(1-\mu^2))
\right]
\exp\left(\frac{1+a\mu}{1-\mu^2}\right) 
\\
\geqs
&\left[
(1+a\mu)+(2\mu^2(1-\mu^2))
\right]
\exp\left(\frac{1-a\mu}{1-\mu^2}\right) 
	\end{split}\\
\intertext{$(\Updownarrow)$~Replacing the exponential functions with their respective Taylor series, we get}
\nonumber\begin{split}
&\left[
(1-a\mu)+(2\mu^2(1-\mu^2))
\right]
\left[\sum_{k=0}^{\infty}\frac{((1+a\mu)/(1-\mu^2))^k}{k!}\right]
\\
\geqs
&\left[
(1+a\mu)+(2\mu^2(1-\mu^2))
\right]
\left[\sum_{k=0}^{\infty}\frac{((1-a\mu)/(1-\mu^2))^k}{k!}\right]
	\end{split}\\
\intertext{$(\Updownarrow)$~Pulling out the first two terms of the series, we get}
\nonumber\begin{split}
&\left[
(1-a\mu)+(2\mu^2(1-\mu^2))
\right]
\left[1+\frac{1+a\mu}{1-\mu^2}+\sum_{k=2}^{\infty}\frac{((1+a\mu)/(1-\mu^2))^k}{k!}
\right]
\\
\geqs
&\left[
(1+a\mu)+(2\mu^2(1-\mu^2))
\right]
\left[1+\frac{1-a\mu}{1-\mu^2}+\sum_{k=2}^{\infty}\frac{((1-a\mu)/(1-\mu^2))^k}{k!}
\right]
	\end{split}\\
\intertext{$(\Leftarrow)$~Separate the previous line into two inequalities described by \textbf{Condition 1a} and $\textbf{Condition 1b}$; if both are true then the combined inequality is true.}
\intertext{\textbf{Condition 1a:}}
\nonumber
\begin{split}
& \left[
(1-a\mu)+(2\mu^2(1-\mu^2))
\right]\cdot\left[1+\frac{1+a\mu}{1-\mu^2}\right]\\
\geq& 
\left[
(1+a\mu)+(2\mu^2(1-\mu^2))
\right]\cdot\left[1+\frac{1-a\mu}{1-\mu^2}\right]
\end{split}
\\
\intertext{$(\Updownarrow)$~Multiplying through by $(1-\mu^2)$, we get}
\nonumber\begin{split}
& \left[
(1-a\mu)+(2\mu^2(1-\mu^2))
\right]\cdot\left[1-\mu^2+1+a\mu\right]\\
\geq &
\left[
(1+a\mu)+(2\mu^2(1-\mu^2))
\right]\cdot\left[1-\mu^2+1-a\mu\right]
\end{split}\\
\intertext{$(\Updownarrow)$~Splitting out the terms in the first bracket and canceling $(1-a\mu)(1+a\mu)$, we get}
\nonumber\begin{split}
& \left(1-a\mu \right)\left(1-\mu^2\right)+(2\mu^2(1-\mu^2))\left[1-\mu^2+1+a\mu\right]\\
\geq &
\left(1+a\mu \right)\left(1-\mu^2\right)+(2\mu^2(1-\mu^2))\left[1-\mu^2+1-a\mu\right]
\end{split}\\
\intertext{$(\Updownarrow)$~Further canceling additive constants from both sides, we get}
\nonumber\begin{split}
& \left(-a\mu \right)\left(1-\mu^2\right)+(2\mu^2(1-\mu^2))\left[+a\mu\right]\\
\geq &
\left(+a\mu \right)\left(1-\mu^2\right)+(2\mu^2(1-\mu^2))\left[-a\mu\right]
\end{split}\\
\intertext{$(\Updownarrow)$~Dividing out $a\mu(1-\mu^2)$ and grouping all terms on one side, we get}
\nonumber
& 2(2\mu^2-1)\geq 0~~\checkmark~~\text{which is finally true, directly from the assumption of Case 1.}
\intertext{\textbf{Condition 1b:}}
\nonumber\begin{split}
&\left[
(1-a\mu)+(2\mu^2(1-\mu^2))
\right]
\left[\sum_{k=2}^{\infty}\frac{((1+a\mu)/(1-\mu^2))^k}{k!}
\right]
\\
\geqs
&\left[
(1+a\mu)+(2\mu^2(1-\mu^2))
\right]
\left[\sum_{k=2}^{\infty}\frac{((1-a\mu)/(1-\mu^2))^k}{k!}
\right]
	\end{split}\\
\intertext{$(\Uparrow)$~Dropping the  $(2\mu^2(1-\mu^2))$ terms -- by the left-hand side sum terms dominating for every $k$ -- we get}
\nonumber
\begin{split}
&\left[
(1-a\mu)
\right]
\left[\sum_{k=2}^{\infty}\frac{((1+a\mu)/(1-\mu^2))^k}{k!}
\right] \\
\geqs&
\left[
(1+a\mu)
\right]
\left[\sum_{k=2}^{\infty}\frac{((1-a\mu)/(1-\mu^2))^k}{k!}
\right]
\end{split}\\
\intertext{$(\Updownarrow)$~Noting $k\geq 2$ and canceling a factor of $(1-a\mu)(1+a\mu)$, we get}
\nonumber
&
\left[\sum_{k=2}^{\infty}\frac{(1+a\mu)^{k-1}}{(1-\mu^2)^k\cdot k!}
\right]\geqs
\left[\sum_{k=2}^{\infty}\frac{(1-a\mu)^{k-1}}{(1-\mu^2)^k\cdot k!}
\right]~~\checkmark~~\text{which is true for every $k$ from $a\mu\geqs0$.}
\intertext{\textbf{Case 2:} $(2\mu^2-a^2-1)\leq 0$.}
\intertext{Note that we have $a\mu < 1$ because the lemma's domain has $a<1$ and $\mu< 1$.
Starting anew for \textbf{Case 2} from inequality~\eqref{eq:critical inequality}, replacing the exponential functions with their respective Taylor Series, we get inequality~$\eqref{eq:critical inequality} \Leftrightarrow$}
\nonumber\begin{split}
&\left[
(2\mu^2-a^2-1)\cdot(1-a\mu)+(2\mu^2(1-\mu^2))
\right]
\left[\sum_{k=0}^{\infty}\frac{((a\mu)/(1-\mu^2))^k}{k!}\right]
\\
\geqs
&\left[
(2\mu^2-a^2-1)\cdot(1+a\mu)+(2\mu^2(1-\mu^2))
\right]
\left[\sum_{k=0}^{\infty}\frac{((-a\mu)/(1-\mu^2))^k}{k!}\right]
	\end{split}\\
\intertext{$(\Updownarrow)$~Pulling out the first two terms of the series, we get}
\nonumber\begin{split}
&\left[
(2\mu^2-a^2-1)\cdot(1-a\mu)+(2\mu^2(1-\mu^2))
\right]
\left[1+\frac{a\mu}{1-\mu^2}+\sum_{k=2}^{\infty}\frac{((a\mu)/(1-\mu^2))^k}{k!}\right]
\\
\geqs
&\left[
(2\mu^2-a^2-1)\cdot(1+a\mu)+(2\mu^2(1-\mu^2))
\right]
\left[1-\frac{a\mu}{1-\mu^2}+\sum_{k=2}^{\infty}\frac{((-a\mu)/(1-\mu^2))^k}{k!}\right]
	\end{split}\\
\intertext{$(\Leftarrow)$~Separate the previous line into two inequalities described by \textbf{Condition 2a} and $\textbf{Condition 2b}$; if both are true then the combined inequality is true.
}
\intertext{\textbf{Condition 2a:}}
\nonumber\begin{split}
&\left[
(2\mu^2-a^2-1)\cdot(1-a\mu)+(2\mu^2(1-\mu^2))
\right]
\left[1+\frac{a\mu}{1-\mu^2}\right]
\\
\geqs
&\left[
(2\mu^2-a^2-1)\cdot(1+a\mu)+(2\mu^2(1-\mu^2))
\right]
\left[1-\frac{a\mu}{1-\mu^2}\right]
	\end{split}\\
\intertext{$(\Updownarrow)$~Multiplying through by $(1-\mu^2)$, we get}
\nonumber\begin{split}
&\left[
(2\mu^2-a^2-1)\cdot(1-a\mu)+(2\mu^2(1-\mu^2))
\right]
\left[1-\mu^2+a\mu\right]
\\
\geqs
&\left[
(2\mu^2-a^2-1)\cdot(1+a\mu)+(2\mu^2(1-\mu^2))
\right]
\left[1-\mu^2-a\mu\right]
	\end{split}\\
\intertext{$(\Updownarrow)$~Splitting out the terms in the first bracket and canceling the resulting (additively) matching terms, we get}
\nonumber\begin{split}
&\left[
(2\mu^2-a^2-1)\cdot(1-a\mu)\cdot(-\mu^2)+(2\mu^2(1-\mu^2))\cdot(+a\mu)
\right]
\\
\geqs
&\left[
(2\mu^2-a^2-1)\cdot(1+a\mu)\cdot(-\mu^2)+(2\mu^2(1-\mu^2))\cdot(-a\mu)
\right]
	\end{split}\\
\intertext{$(\Updownarrow)$~Further canceling additively and then moving the $-1$ factor of $-\mu^2$, we get}
\nonumber\begin{split}
&\left[
(1+a^2-2\mu^2)\cdot(-a\mu)\cdot(\mu^2)+(2\mu^2(1-\mu^2))\cdot(+a\mu)
\right]
\\
\geqs
&\left[
(1+a^2-2\mu^2)\cdot(+a\mu)\cdot(\mu^2)+(2\mu^2(1-\mu^2))\cdot(-a\mu)
\right]
	\end{split}\\
\intertext{$(\Updownarrow)$~Combining like-terms on each side and dividing through by $2$, we get}
\nonumber\begin{split}
&(2\mu^2(1-\mu^2))\cdot(+a\mu) \geqs (1+a^2-2\mu^2)\cdot(+a\mu)\cdot(\mu^2)
	\end{split}\\
\intertext{$(\Updownarrow)$~Dividing by $a\mu^3$ and re-organizing, we get}
\nonumber\begin{split}
&(1-\mu^2)+(1-\mu^2) \geqs (1-\mu^2)+(a^2-\mu^2)~~\checkmark~~\text{which is true because the domain has}~a< 1
	\end{split}\\
\intertext{\textbf{Condition 2b:}}
\nonumber\begin{split}
&\left[
(2\mu^2-a^2-1)\cdot(1-a\mu)+(2\mu^2(1-\mu^2))
\right]
\left[\sum_{k=2}^{\infty}\frac{((a\mu)/(1-\mu^2))^k}{k!}\right]
\\
\geq
&\left[
(2\mu^2-a^2-1)\cdot(1+a\mu)+(2\mu^2(1-\mu^2))
\right]
\left[\sum_{k=2}^{\infty}\frac{((-a\mu)/(1-\mu^2))^k}{k!}\right]
	\end{split}\\
\intertext{$(\Leftarrow)$~To prove that the inequality of the previous line holds, it is sufficient to show that the next inequality holds for pairs of consecutive terms within its sums $\forall~\text{even}~k\geq 2$; for each fixed, even $k\geq2$, we require:}
\nonumber\begin{split}
&\left[
(2\mu^2-a^2-1)\cdot(1-a\mu)+(2\mu^2(1-\mu^2))
\right]
\left[\frac{((a\mu)/(1-\mu^2))^k}{k!}+\frac{((a\mu)/(1-\mu^2))^{k+1}}{(k+1)!}\right]
\\
\geq
&\left[
(2\mu^2-a^2-1)\cdot(1+a\mu)+(2\mu^2(1-\mu^2))
\right]
\left[\frac{((a\mu)/(1-\mu^2))^k}{k!}-\frac{((a\mu)/(1-\mu^2))^{k+1}}{(k+1)!}\right]
	\end{split}\\
\intertext{$(\Updownarrow)$~Factoring out common terms within the bracket of the Taylor series terms, we get}
\nonumber\begin{split}
&\left[
(2\mu^2-a^2-1)\cdot(1-a\mu)+(2\mu^2(1-\mu^2))
\right]
\left[\frac{((a\mu)/(1-\mu^2))^k}{k!}\right]\left[1+\frac{((a\mu)/(1-\mu^2))}{(k+1)}\right]
\\
\geq
&\left[
(2\mu^2-a^2-1)\cdot(1+a\mu)+(2\mu^2(1-\mu^2))
\right]
\left[\frac{((a\mu)/(1-\mu^2))^k}{k!}\right]\left[1-\frac{((a\mu)/(1-\mu^2))}{(k+1)}\right]
	\end{split}\\
\intertext{$(\Updownarrow)$~Multiplying through by $\frac{(1-\mu^2)^{k+1}\cdot(k+1)!}{(a\mu)^k}$, we get}
\nonumber\begin{split}
&\left[
(2\mu^2-a^2-1)\cdot(1-a\mu)+(2\mu^2(1-\mu^2))
\right]
\left[(k+1)(1-\mu^2)+a\mu\right]
\\
\geq
&\left[
(2\mu^2-a^2-1)\cdot(1+a\mu)+(2\mu^2(1-\mu^2))
\right]
\left[(k+1)(1-\mu^2)-a\mu\right]
	\end{split}\\
\intertext{$(\Updownarrow)$~Expanding within the second brackets, we get}
\nonumber\begin{split}
&\left[
(2\mu^2-a^2-1)\cdot(1-a\mu)+(2\mu^2(1-\mu^2))
\right]
\left[1-\mu^2+a\mu+k(1-\mu^2)\right]
\\
\geq
&\left[
(2\mu^2-a^2-1)\cdot(1+a\mu)+(2\mu^2(1-\mu^2))
\right]
\left[1-\mu^2-a\mu+k(1-\mu^2)\right]
	\end{split}\\
\intertext{$(\Uparrow)$~Dropping the $[1-\mu^2+a\mu]$ terms because they correspond exactly to an inequality already shown to hold in the sequence of steps to prove \textbf{Condition 2a}, we get}
\nonumber\begin{split}
&\left[
(2\mu^2-a^2-1)\cdot(1-a\mu)+(2\mu^2(1-\mu^2))
\right]
\left[k(1-\mu^2)\right]
\\
\geq
&\left[
(2\mu^2-a^2-1)\cdot(1+a\mu)+(2\mu^2(1-\mu^2))
\right]
\left[k(1-\mu^2)\right]
	\end{split}\\
\intertext{$(\Updownarrow)$~Dropping the multiplicative constant and then the additive constant, both of which are non-negative, we get}
\nonumber\begin{split}
&\left[
(2\mu^2-a^2-1)\cdot(1-a\mu)\right]
\\
\geq
&\left[
(2\mu^2-a^2-1)\cdot(1+a\mu)\right]~~\checkmark~~\text{by Case 2 assumption}~(2\mu^2-a^2-1)\leq0,~\text{and}~0\leq a\mu<1
\end{split}\\
\intertext{This last inequality holds because both sides are non-positive and the left-hand-side has weakly smaller magnitude.\qedhere
\ifMS \Halmos \fi } 
\nonumber
\end{align}
\ifMS \endproof
\else \end{proof} \fi

\begin{lemma}
	\label{lem:critical inequality large a}
	Inequality \eqref{eq:critical inequality}
	holds for all $\mu \in (0, 1)$ and 
	$a \geq 1$.
\end{lemma}

\ifMS
\proof{Proof.}
\else \begin{proof}\fi
	Let $f(a, \mu)$ be the left-hand-side in inequality~\eqref{eq:critical inequality},
	i.e.,
	\begin{align*}
		f(a, \mu)
		&\triangleq [(2\mu^2 - a^2 - 1)(1 - a\mu) + 
	2\mu^2(1-\mu^2)]
	\exp\left(
		\frac{a\mu}{1-\mu^2}
	\right)\\
	&\qquad\qquad\qquad-
	[(2\mu^2 - a^2 - 1)(1 + a\mu) + 
	2\mu^2(1-\mu^2)]
	\exp\left(
		\frac{-a\mu}{1-\mu^2}
	\right)
\end{align*}
Next we show the following inequalities:
for all $\mu \in (0, 1)$ and $a = 1$,

(i) 
$f(a, \mu) > 0$;

(ii)
$f_1(a, \mu) \triangleq (1-\mu^2)\frac{\partial}{\partial a}f(a,\mu) \geq 0$;

(iii)
$f_2(a, \mu) \triangleq (1-\mu^2)\frac{\partial}{\partial a}f_1(a,\mu) \geq 0$;

\noindent and for all $\mu\in (0, 1)$ and $a \geq 1$,

(iv)
$f_3(a, \mu) \triangleq 
\frac{1-\mu^2}{\mu^2}
\exp\left(
	\frac{1}{1-\mu^2}
\right)
\frac{\partial}{\partial a}
f_2(a, \mu)
\geq 0$.


\noindent Combining (i)--(iv) proves the lemma.

\paragraph{Proof of (i).}
By definition, plugging in $a = 1$, $f(a,\mu)$ is 
\begin{align*}
	f(1,\mu) &= 
	-2(1-\mu^2)(1-\mu-\mu^2)
	\exp\left(\frac{\mu}{1-\mu^2}\right)
	+2(1-\mu^2)(1+\mu-\mu^2)
	\exp\left(\frac{-\mu}{1-\mu^2}\right)
\end{align*}
Now consider the Taylor series expansion of
$\exp\left(
	\frac{-\mu}{1-\mu^2}
\right)$ and $
\exp\left(
	\frac{\mu}{1-\mu^2}
\right)$ in $f(1,\mu)$.
We analyze the first term and the remaining terms separately.

\noindent\textsl{The first term of the Taylor series expansion in $f(1, \mu)$} is 
\begin{align*}
	-2(1-\mu^2)(1-\mu-\mu^2)
	\exp\left(\frac{\mu}{1-\mu^2}\right)
	+2(1-\mu^2)(1+\mu-\mu^2)
	\exp\left(\frac{-\mu}{1-\mu^2}\right)
	=4\mu(1-\mu^2)> 0
\end{align*}
for all $\mu\in(0, 1)$.

\noindent\textsl{The k-th even terms and (k+1)-th odd terms
of the Taylor series expansion, for $(k \geq 2)$, in $f(1, \mu)$}
 are 
\begin{align*}
	&
	-2(1-\mu^2)(1-\mu-\mu^2)
\left(
	\frac{1}{k!}\left(\frac{\mu}{1-\mu^2}\right)^k
	+
	\frac{1}{(k+1)!}\left(\frac{\mu}{1-\mu^2}\right)^{k+1}
\right)\\
&
\qquad\qquad\qquad
+
	2(1-\mu^2)(1+\mu-\mu^2)
\left(
	\frac{1}{k!}\left(\frac{-\mu}{1-\mu^2}\right)^k
	+
	\frac{1}{(k+1)!}\left(\frac{-\mu}{1-\mu^2}\right)^{k+1}
\right) \\
=&
\frac{4\mu^{k+1}
	k 
}{(1-\mu^2)^{k-1}(k+1)!}
> 0
\end{align*}
for all $\mu\in(0, 1)$.

\paragraph{Proof of (ii).}
By definition,
\begin{align*}
	f_1(a, \mu) &= 
	\left[
		a^3\mu^2 + \mu^3 + a^2\mu(2 - 3\mu^2)
		- a(2 - 3\mu^2 + 2\mu^4)
	\right]
	\exp\left(
		\frac{a\mu}{1-\mu^2}
	\right)\\
	&
	\qquad\qquad\qquad
	-
	\left[
		a^3\mu^2 - \mu^3 - a^2\mu(2 - 3\mu^2)
		- a(2 - 3\mu^2 + 2\mu^4)
	\right]
	\exp\left(
		\frac{-a\mu}{1-\mu^2}
	\right)
	\intertext{Plugging in $a = 1$ yields}
	f_1(1, \mu) &=
	-\left[
		1 - \mu - 2\mu^2
		+ \mu^3 + \mu^4
	\right]
	\exp\left(
		\frac{\mu}{1-\mu^2}
	\right)\\
	&
	\qquad\qquad\qquad
	+
	\left[
		1 + \mu - 2\mu^2
		- \mu^3 + \mu^4
	\right]
	\exp\left(
		\frac{-\mu}{1-\mu^2}
	\right)
\end{align*}
Now consider the Taylor series expansion of 
$\exp\left(
	\frac{-\mu}{1-\mu^2}
\right)$ and $
\exp\left(
	\frac{\mu}{1-\mu^2}
\right)$ in $f_1(1,\mu)$.
We analyze the first term and the remaining terms separately.

\noindent\textsl{The first term of Taylor series expansion in $f_1(1, \mu)$} is 
\begin{align*}
	&-\left[
	1 - \mu - 2\mu^2
	+ \mu^3 + \mu^4
\right]
+
\left[
	1 + \mu - 2\mu^2
	- \mu^3 + \mu^4
\right]
=2\mu(1-\mu^2) \geq 0
\end{align*}
for all $\mu\in[0, 1)$.

\noindent\textsl{The k-th even terms and (k+1)-th odd terms
of the Taylor series expansion, for $(k \geq 2)$, in $f_1(1, \mu)$}
are
\begin{align*}
	&-\left[
	1 - \mu - 2\mu^2
	+ \mu^3 + \mu^4
\right]
\left(
	\frac{1}{k!}\left(\frac{\mu}{1-\mu^2}\right)^k
	+
	\frac{1}{(k+1)!}\left(\frac{\mu}{1-\mu^2}\right)^{k+1}
\right)\\
&
\qquad\qquad\qquad
+
\left[
	1 + \mu - 2\mu^2
	- \mu^3 + \mu^4
\right]
\left(
	\frac{1}{k!}\left(\frac{-\mu}{1-\mu^2}\right)^k
	+
	\frac{1}{(k+1)!}\left(\frac{-\mu}{1-\mu^2}\right)^{k+1}
\right) \\
=&
\frac{2\mu^{k+1}
	k 
}{(1-\mu^2)^{k-1}(k+1)!}
\geq 0
\end{align*}
for all $\mu\in[0, 1)$.

\paragraph{Proof of (iii).}
By definition,
\begin{align*}
	f_2(a, \mu) &= 
	\left[
		- 2 + 5\mu^2 + a^3\mu^3
		- 4\mu^4 + 2\mu^6 + a^2\mu^2(5 - 6\mu^2)
		+ a\mu(2 - 7\mu^2 + 4\mu^4)
	\right]
	\exp\left(
		\frac{a\mu}{1-\mu^2}
	\right)\\
	&
	\qquad
	-\left[
		- 2 + 5\mu^2 - a^3\mu^3
		- 4\mu^4 + 2\mu^6 + a^2\mu^2(5 - 6\mu^2)
		- a\mu(2 - 7\mu^2 + 4\mu^4)
	\right]
	\exp\left(
		\frac{-a\mu}{1-\mu^2}
	\right)
	\intertext{Plugging in $a = 1$ yields}
	f_2(1, \mu) &=
	-\left[
		1 - \mu - 5\mu^2 
		+3\mu^3 + 5\mu^4
		-2\mu^5 - \mu^6
	\right]
	\exp\left(
		\frac{\mu}{1-\mu^2}
	\right)\\
	&
	\qquad
	+
	\left[
		1 + \mu - 5\mu^2 
		-3\mu^3 + 5\mu^4
		+2\mu^5 - \mu^6
	\right]
	\exp\left(
		\frac{-\mu}{1-\mu^2}
	\right)
\end{align*}
Note that $f_2(1, 0) = 0$ and we show $f_2(1, \mu)$ 
is non-decreasing in $\mu$ below. To see this, 
consider the partial derivative of $f_2(1, \mu)$ with respect to $\mu$,
it is 
\begin{align*}
	\frac{\partial }{\partial \mu}f_2(1, \mu)
	&=
	-\left[
		-11 + 7\mu
		+31 \mu^2 - 22 \mu^3 
		- 24 a^4 
		+ 11 \mu^5 + 6 \mu^6
	\right]
	\exp\left(
		\frac{\mu}{1-\mu^2}
	\right)
	\frac{\mu}{1-\mu^2}\\
	&
	\qquad
	+
	\left[
		-11 - 7 \mu + 31 \mu^2 
		+ 22 \mu^3 - 24 \mu^4 
		- 11 \mu^5 + 6 \mu^6
	\right]
	\exp\left(
		\frac{-\mu}{1-\mu^2}
	\right)\frac{\mu}{1-\mu^2}
\end{align*}
Multiplying $
\frac{\partial }{\partial \mu}f_2(1, \mu)$
by
$
\frac{1-\mu^2}{\mu}\exp\left(\frac{1}{2-2\mu^2}\right)
$, we get
\begin{align}
	\begin{split}
		\label{eq:(iii)}
		&
		-\left[
			-11 + 7\mu
			+31 \mu^2 - 22 \mu^3 
			- 24 a^4 
			+ 11 \mu^5 + 6 \mu^6
		\right]
		\exp\left(
			\frac{1+2\mu}{2-2\mu^2}
		\right)
		\\
		&
		\qquad
		+
		\left[
			-11 - 7 \mu + 31 \mu^2 
			+ 22 \mu^3 - 24 \mu^4 
			- 11 \mu^5 + 6 \mu^6
		\right]
		\exp\left(
			\frac{1-2\mu}{2-2\mu^2}
		\right)
	\end{split}
\end{align}

Now consider the Taylor series expansion of 
$\exp\left(
	\frac{1+2\mu}{2-2\mu^2}
\right)$ and $
\exp\left(
	\frac{1-2\mu}{2-2\mu^2}
\right)$ in \eqref{eq:(iii)}.
We analyze the first two terms and the remaining terms separately.

\noindent\textsl{The first and second terms of the Taylor series expansion in \eqref{eq:(iii)}.} It is 
\begin{align*}
	&-\left[
	-11 + 7\mu
	+31 \mu^2 - 22 \mu^3 
	- 24 \mu^4 
	+ 11 \mu^5 + 6 \mu^6
\right]
\left[
1 + \frac{1 + 2\mu}{2 - 2\mu^2}\right]
\\
&
\qquad
+\left[
	-11 - 7 \mu + 31 \mu^2 
	+ 22 \mu^3 - 24 \mu^4 
	- 11 \mu^5 + 6 \mu^6
\right]
\left[
1 + \frac{1 - 2\mu}{2 - 2\mu^2}\right]
\\
=&\mu+19\mu^3-10\mu^5 \geq 0
\end{align*}
for all $\mu\in[0, 1)$.

\noindent\textsl{The k-th term
of the Taylor series expansion, for $(k \geq 3)$, in \eqref{eq:(iii)}.} 
It is 
\begin{align}
	\begin{split}
		\label{eq:(iii) 2}
		&
		-\left[
			-11 + 7\mu
			+31 \mu^2 - 22 \mu^3 
			- 24 \mu^4 
			+ 11 \mu^5 + 6 \mu^6
		\right]
		\left[
		\frac{1}{k!}\frac{(1 + 2\mu)^k}{(2 - 2\mu^2)^k}\right]
		\\
		&
		\qquad
		+
		\left[
			-11 - 7 \mu + 31 \mu^2 
			+ 22 \mu^3 - 24 \mu^4 
			- 11 \mu^5 + 6 \mu^6
		\right]
		\left[
		\frac{1}{k!}\frac{(1 - 2\mu)^k}{(2 - 2\mu^2)^k}\right]
	\end{split}
\end{align}
Multiplying by $k!(2-2\mu^2)^k$, 
\begin{align*}
	&
	-\left[
		-11 + 7\mu
		+31 \mu^2 - 22 \mu^3 
		- 24 \mu^4 
		+ 11 \mu^5 + 6 \mu^6
	\right]
	(1+2\mu)(1+2\mu)^{k-1}
	\\
	&
	\qquad
	+
	\left[
		-11 - 7 \mu + 31 \mu^2 
		+ 22 \mu^3 - 24 \mu^4 
		- 11 \mu^5 + 6 \mu^6
	\right]
	(1-2\mu)(1-2\mu)^{k-1}
\end{align*}
Note that 
$
-\left[
	-11 + 7\mu
	+31 \mu^2 - 22 \mu^3 
	- 24 \mu^4 
	+ 11 \mu^5 + 6 \mu^6
\right] \geq 0$ 
and 
$1+2\mu \geq \abs{1 - 2\mu}$
for all $\mu\in[0, 1)$.
Thus, to show \eqref{eq:(iii) 2} 
is non-negative for all $\mu\in[0, 1)$,
it is sufficient to argue 
\begin{align*}
	&-\left[
	-11 + 7\mu
	+31 \mu^2 - 22 \mu^3 
	- 24 \mu^4 
	+ 11 \mu^5 + 6 \mu^6
\right]
(1+2\mu)
\\
&
\qquad
\geq 
\abs{
	\left[
		-11 - 7 \mu + 31 \mu^2 
		+ 22 \mu^3 - 24 \mu^4 
		- 11 \mu^5 + 6 \mu^6
	\right]
(1-2\mu)}
\end{align*}
which is true for all $\mu\in [0, 1)$.

\paragraph{Proof of (iv).}
By definition,
\begin{align*}
	f_3(a,\mu) = 
	&\left[
	a^3\mu^2 + a^2\mu(8 - 9\mu^2)
	-\mu(4 - 7\mu^2 + 2\mu^4)
	+a(12 - 29\mu^2 + 16\mu^4)
\right]
\exp\left(
	\frac{1+a\mu}{1-\mu^2}
\right)
\\
&
\qquad-
\left[
	a^3\mu^2 - a^2\mu(8 - 9\mu^2)
	+\mu(4 - 7\mu^2 + 2\mu^4)
	+a(12 - 29\mu^2 + 16\mu^4)
\right]
\exp\left(
	\frac{1-a\mu}{1-\mu^2}
\right)
\end{align*}

Now consider the Taylor series expansion of 
$\exp\left(\frac{1+a\mu}{1-\mu^2}\right)$
and 
$\exp\left(\frac{1-a\mu}{1-\mu^2}\right)$
in 
$f_3(a, \mu)$.
We analyze the first
two terms and the remaining terms separately.

\noindent\textsl{The first and second terms of the Taylor series expansion of $f_3(a, \mu)$}
are 
\begin{align*}
	&\left[
	a^3\mu^2 + a^2\mu(8 - 9\mu^2)
	-\mu(4 - 7\mu^2 + 2\mu^4)
	+a(12 - 29\mu^2 + 16\mu^4)
\right]
\left[
1 + \frac{1 + a\mu}{1 - \mu^2}\right]
\\
&\qquad-
\left[
	a^3\mu^2 - a^2\mu(8 - 9\mu^2)
	+\mu(4 - 7\mu^2 + 2\mu^4)
	+a(12 - 29\mu^2 + 16\mu^4)
\right]
\left[
1 + \frac{1 - a\mu}{1 - \mu^2}\right]\\
=&
\frac{2\mu }{1-\mu^2}
(
-8 + 18\mu^2
+a^4\mu^2 - 11\mu^4
+2\mu^6 + a^2(28 - 55\mu^2 + 25\mu^4)
)
\geq 0
\end{align*}
which is true for all $\mu \in[0,1)$ and $a\geq 1$.

\noindent\textsl{The k-th term
of the Taylor series expansion, for $(k \geq 3)$, of $f_3(a,\mu)$.} 
It is 
\begin{align}
	\begin{split}
		\label{eq:(iv)}
		&\left[
		a^3\mu^2 + a^2\mu(8 - 9\mu^2)
		-\mu(4 - 7\mu^2 + 2\mu^4)
		+a(12 - 29\mu^2 + 16\mu^4)
	\right]
	\left[
	\frac{1}{k!}\frac{(1 + a\mu)^k}{(1 - \mu^2)^k}\right]
	\\
	&\qquad
	-\left[
		a^3\mu^2 - a^2\mu(8 - 9\mu^2)
		+\mu(4 - 7\mu^2 + 2\mu^4)
		+a(12 - 29\mu^2 + 16\mu^4)
	\right]
	\left[
	\frac{1}{k!}\frac{(1 - a\mu)^k}{(1 - \mu^2)^k}\right]
\end{split}
															\end{align}
															Multiplying by $k!(1 - \mu^2)^k$,
															\begin{align*}
																\begin{split}
																	&\left[
																	a^3\mu^2 + a^2\mu(8 - 9\mu^2)
																	-\mu(4 - 7\mu^2 + 2\mu^4)
																	+a(12 - 29\mu^2 + 16\mu^4)
																\right]
																\left(
																1  + a\mu\right)(1+a\mu)^{k-1}  
																\\
																&\qquad
																-\left[
																	a^3\mu^2 - a^2\mu(8 - 9\mu^2)
																	+\mu(4 - 7\mu^2 + 2\mu^4)
																	+a(12 - 29\mu^2 + 16\mu^4)
																\right]
																\left(
																1  - a\mu\right)(1-a\mu)^{k-1}
															\end{split}
														\end{align*}
														Notice that 
														$
														\left[
															a^3\mu^2 + a^2\mu(8 - 9\mu^2)
															-\mu(4 - 7\mu^2 + 2\mu^4)
															+a(12 - 29\mu^2 + 16\mu^4)
														\right]
														\left(
														1  + a\mu\right)\geq 0$
														and 
														$1+a\mu \geq \abs{1 - a\mu}$
														for all $\mu\in[0, 1)$ and $a \geq 1$.
														Thus, to show \eqref{eq:(iv)} 
														is non-negative for all $\mu\in[0, 1)$ and $a\geq 1$,
														it is sufficient to argue 
														\begin{align*}
															\begin{split}
																&\left[
																a^3\mu^2 + a^2\mu(8 - 9\mu^2)
																-\mu(4 - 7\mu^2 + 2\mu^4)
																+a(12 - 29\mu^2 + 16\mu^4)
															\right]
															\left(
															1  + a\mu\right)
															\\
															&\qquad
															\geq 
															\abs{\left[
																a^3\mu^2 - a^2\mu(8 - 9\mu^2)
																+\mu(4 - 7\mu^2 + 2\mu^4)
																+a(12 - 29\mu^2 + 16\mu^4)
															\right]
															\left(
														1  - a\mu\right)}
														\end{split}
													\end{align*}
													which is true for all $\mu\in[0, 1)$ and $a\geq 1$.
\ifMS \Halmos \endproof
\else \end{proof}
\fi

\fi

\section{Bias and Variance in Ridge Regression}

\subsection{One-Dimensional Case}\label{apx:one-dimensional-ridge}
Suppose that features $x_i$ of training points $i$ are 
single-dimensional with labels
$$
f(x_i)=wx_i+e_i,
$$
where $e_i$ are normally distributed with mean zero
and variance $\sigma^2_e.$

In this model, when running ridge regression, the regularization parameter $\lambda^*$ that minimizes the regularized risk
is independent of $x$:
\begin{proposition}\label{prop:one-dimensional-ridge}
The $\lambda^*$ that minimizes the regularized risk minimizes the sum of variance and squared-bias {\em for every feature $x$}. 
\end{proposition}

\begin{proof}
When a 
ridge regression with regularization parameter
$\lambda$ is used
to estimate $w$ from the 
training set of $m$ points, the corresponding estimator is
$$
\widehat{w}=\frac{\sum^m_{i=1}x_i\,f(x_i)}{\sum^m_{i=1}x_i^2+\lambda}
=w -
w\frac{\lambda}{
\sum^m_{i=1}x_i^2+\lambda
}+\sum^m_{i=1}\frac{x_i}{\sum^m_{j=1}x_j^2+\lambda}e_i.
$$

The prediction for label of the new
example with feature $x$
and label $f(x)=w\,x+e_x$ produced from the ridge regression is
$$
\widehat{f}(x)=\widehat{w}x=
w\,x -
w\,x\frac{\lambda}{
\sum^m_{i=1}x_i^2+\lambda
}+\sum^m_{i=1}\frac{x_i\,x}{\sum^m_{j=1}x_j^2+\lambda}e_i.
$$
The error of this prediction is
$$
-
w\,x\frac{\lambda}{
\sum^m_{i=1}x_i^2+\lambda
}+\sum^m_{i=1}\frac{x_i\,x}{\sum^m_{j=1}x_j^2+\lambda}e_i-e_x.
$$
The expectation of this error
is 
$$
-
w\,x\frac{\lambda}{
\sum^m_{i=1}x_i^2+\lambda
}
$$
which represents the bias. The
magnitude of the bias depends
on the value of the feature vector $x$ and the (unknown)
parameter $w$. 

The term
$$
\sum^m_{i=1}\frac{x_i\,x}{\sum^m_{j=1}x_j^2+\lambda}e_i-e_x
$$
has expectation equal to zero.
Its variance is 
$$
x^2\sigma^2_e\frac{\sum^m_{i=1}x_i^2
}{\left(
\sum^m_{i=1}x_i^2
+\lambda
\right)^2}+\sigma_e^2.
$$

For any given $x$, the squared error of $f(x)$ is thus the squared-bias plus the variance, namely
$$w^2x^2\frac{\lambda^2}{
\left(\sum^m_{i=1}x_i^2+\lambda\right)^2
}+x^2\sigma^2_e\frac{\sum^m_{i=1}x_i^2
}{\left(
\sum^m_{i=1}x_i^2
+\lambda
\right)^2}+\sigma_e^2.
$$
Observe that the $\lambda$ for which this is minimized is independent of $x$.
\end{proof}

%

An implication of Proposition~\ref{prop:one-dimensional-ridge} is that, in this setting, assumption A5 from Section~\ref{sec:non-constant-tradeoff-general-framework} holds.
As in Section~\ref{sec:non-constant-tradeoff}, let us suppose that the value of the bias-variance tradeoff is continuous and differentiable, and that this holds for every $x$.
Recall that $g_x(\lambda)$ is the sum of squared-bias and variance at $x$ and $\lambda$, namely
    $$g_x(\lambda) = \mathrm{Bias}_D[\hat f_\lambda(x;D)]^2 + \mathrm{Var}[\hat f_\lambda(x;D)].$$
Assumption A4 from Section~\ref{sec:non-constant-tradeoff-general-framework} states that the expectation of $g_x(\lambda)$, taken over the randomness of $x$, is convex and differentiable. 
For the following proposition, let us assume that these properties hold for every $x$, and not just in expectation:
\begin{itemize}
  \item[A4'] 
     $g_x(\lambda)$ is strictly convex and differentiable for every $x$.
    \end{itemize}

\begin{proposition}\label{prop:one-dimensional-ridge-A5}
Under assumption A4', assumption A5 holds.
\end{proposition}

\begin{proof}
Assumption A5 states that
 $$\frac{dg_x(\lambda^*)}{d\lambda}~~\mbox{and}~~\frac{\partial u_i\left(x;X(\mu_x^2,T) ,X(\mu_x^2,T_x)\right)}{\partial T}$$
 are weakly negatively correlated. We will show that in the one-dimensional ridge regression setup of this section, the two random variables are, in fact,
 uncorrelated. 
 
 By Proposition~\ref{prop:one-dimensional-ridge},
the $\lambda^*$ that minimizes the regularized risk minimizes the sum of variance and squared-bias {\em for every feature $x$}. By assumption A4', this implies that, for every fixed $x$,
$$\frac{dg_x(\lambda^*)}{d\lambda}=0.$$
Thus, the random variable $\frac{dg_x(\lambda^*)}{d\lambda}$, with distribution over $x$, is constant and equal to 0. It is thus uncorrelated with the second random variable 
from assumption A5.
\end{proof}

\subsection{Multidimensional Case}\label{apx:multidimensional-ridge}
Suppose that feature vector
$x_i$ is $p$-dimensional and
$$
f(x_i)=w^Tx_i+e_i,
$$
where $e_i$ are independent
across examples and $E[e_i]=0$ and $E[e_i^2]=
\sigma^2_e.$

We show that, when running ridge regression in this model, an increase in the regularizer $\lambda$ increases the bias
and decreases the variance on the prediction of any $x$:
\begin{proposition}\label{prop:multidimensional-ridge}
In ridge regression, for every vector $x$, the bias $\mathrm{Bias}_D[\hat f_\lambda(x; D)]$ is increasing in $\lambda$ and the variance $\mathrm{Var}[\hat f_\lambda(x;D)]$
is decreasing in $\lambda$.
\end{proposition}

\begin{proof}
Let $X$ be the design matrix
containing the feature vectors 
of $m$ examples in the training
data, $f=(f(x_1),\ldots,
f(x_m))^T$ and $e=(e_1,\ldots,e_m)^T$. Then the estimator
of the ridge regression can
be written as
$$
\widehat{w}=(X^TX+\lambda\,I)^{-1}
X^T\,f=
w+(X^TX+\lambda\,I)^{-1}
X^TX\,w-w+(X^TX+\lambda\,I)^{-1}X^Te.
$$
This representation shows the 
decomposition of the bias and
variance term of the estimator.
Denote 
$$
\beta(\lambda)=(X^TX+\lambda\,I)^{-1}
X^TX\,w-w
$$
and
$$
\Delta(\lambda)=(X^TX+\lambda\,I)^{-1}X^Te.
$$
Then the prediction for the new
example is
$$
\widehat{f}(x)=
f(x)+x^T\beta(\lambda)+x^T\Delta
(\lambda)-e_x.
$$
The mean squared error is
$$
x^T\left(
\beta(\lambda)\beta(\lambda)^T+
E[\Delta
(\lambda)\Delta
(\lambda)^T]
\right)x+\sigma^2_x.
$$
We note that the first element of the mean squared error is a quadratic form. 
It determines the dependence of the 
mean squared error on $\lambda.$
To find the minimum, note that
$$
\beta(\lambda)=
-\lambda(X^TX+\lambda I)^{-1}w.
$$
Finally
$$
\beta(\lambda)\beta(\lambda)^T=
(X^TX+\lambda I)^{-1}\lambda^2
ww^T(X^TX+\lambda I)^{-1}.
$$
Taking the derivative with respect to $\lambda,$ we obtain a positive
semidefinite matrix, meaning that bias is always increasing in 
$\lambda.$
The variance term is then 
$$
E[\Delta
(\lambda)\Delta
(\lambda)^T]=
\sigma^2_e(X^TX+\lambda I)^{-1}
X^TX(X^TX+\lambda I)^{-1}.
$$
The derivative of this matrix
with respect to $\lambda$ is
$$
-2\lambda \sigma^2_e(X^TX+\lambda I)^{-1}
X^TX(X^TX+\lambda I)^{-1}
(X^TX+\lambda I)^{-1},
$$
which is a negative semidefinite
matrix and so the variance
always increases in $\lambda.$
\end{proof}

%
\ifEC
\section{Proof of Theorem~\ref{thm:general-framework-non-constant-tradeoff} } \label{apx:general-framework-non-constant-tradeoff} 
\begin{proof}
By assumption A1, $$\hat f^*(x;D_j)=X(\mu_x^2,T_x),$$
a normal distribution with squared-bias $\mu_x^2$ and total error $T_x=g_x(\lambda^*)$.
Denote by $g(x;\mu^2)$ the total error of $\hat f^*$ when the feature vector is $x$ and the squared-bias is $\mu^2$.
Then, as in Equation (\ref{eqn:derivative}), for each fixed $x$ we can write
\begin{align*}&\frac{du_i\left(x;X(\mu^2,g(x;\mu^2)) ,X(\mu_x^2,T_x)\right)}{d\mu^2}\\
&~~~=\frac{\partial u_i\left(x;X(\mu^2,T) ,X(\mu_x^2,T_x)\right)}{\partial\mu^2} + \frac{dg(x;\mu^2)}{d\mu^2} \cdot \frac{\partial u_i\left(X(\mu^2,T) ,X(\mu_x^2,T_x)\right)}{\partial T},
\end{align*}
where $T=g(x;\mu^2)$.
Thus,
\begin{align*}&\frac{du_i\left(x; \hat f_{\lambda^*}(x; D_i),\hat f^*(x;D_j)\right)}{d\lambda} \\
&~~~=\frac{d\mu^2}{d\lambda}\cdot\frac{\partial u_i\left(x;X(\mu^2,T_x) ,X(\mu_x^2,T_x)\right)}{\partial\mu^2} + \frac{d\mu^2}{d\lambda}\cdot\frac{dg(x;\mu_x^2)}{d\mu^2} \cdot \frac{\partial u_i\left(X(\mu_x^2,T) ,X(\mu_x^2,T_x)\right)}{\partial T}\\
&~~~=\frac{d\mu^2}{d\lambda}\cdot\frac{\partial u_i\left(x;X(\mu^2,T_x) ,X(\mu_x^2,T_x)\right)}{\partial\mu^2} + \frac{dg_x(\lambda^*)}{d\lambda} \cdot \frac{\partial u_i\left(X(\mu_x^2,T) ,X(\mu_x^2,T_x)\right)}{\partial T}.
\end{align*}

By assumption A3,
$$\frac{d\mu^2}{d\lambda}>0,$$
and by assumption A2, 
$$\frac{\partial u_i\left(x;X(\mu^2,T_x) ,X(\mu_x^2,T_x)\right)}{\partial\mu^2} < 0.$$

By assumption A5, 
\begin{align*}
&\E\left[\frac{dg_x(\lambda^*)}{d\lambda} \cdot \frac{\partial u_i\left(X(\mu^2,T) ,X(\mu_x^2,T_x)\right)}{\partial T}\right]\\
&~~~\leq \E\left[\frac{dg_x(\lambda^*)}{d\lambda}\right] \cdot \E\left[\frac{\partial u_i\left(X(\mu^2,T) ,X(\mu_x^2,T_x)\right)}{\partial T}\right].
\end{align*}

Finally, since $\lambda^*$ minimizes the Bayes risk, assumption A4 implies that
$$\E\left[\frac{dg_x(\lambda^*)}{d\lambda}\right] = \frac{d}{d\lambda}\E[g_x(\lambda^*)]=0.$$

Putting these together implies the claimed inequality.
\end{proof}

\fi

\section{Empirical Demonstration of Assumptions}
\label{apx:empirical-validation}
In this section we empirically demonstrate that two of the assumptions used in Theorem~\ref{thm:general-framework-non-constant-tradeoff} for the ex ante game
hold when the game played on the California housing prices data
and the wine quality data 
from Section~\ref{sec:empirical}.

First, we demonstrate that assumption A1 on the normality of prediction error holds.
For the California housing prices data, Figure~\ref{fig:ridge-normal} plots the distribution, over random
choices of the training data, of the error in prediction on a
particular point in the test data for three values of regularization
parameter $\ridge$.  Notice that the distributions appear roughly
normal, and that the predictions of the lower $\ridge$ value have
higher variance and lower bias.\footnote{Similar observations holds for the wine quality data.}
\begin{figure}
\begin{center}
\includegraphics[height=6cm]{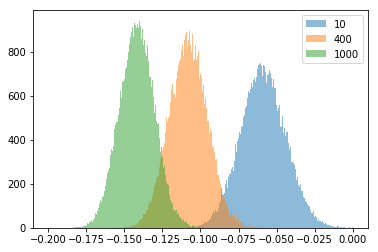}
\end{center}
\caption{Distribution of predictions for three values of
$\ridge$.}
\label{fig:ridge-normal}
\end{figure}

%

Next, we verify that
the negative correlation 
assumption A5 holds.
In particular, 
we estimate $\frac{dg_x(\lambda^*)}{d\lambda}$
by $\frac{\tilde g_x(\lambda^* + \epsilon) - \tilde g_x(\lambda^*)}{\epsilon}$
where $\tilde g_x(\cdot)$
is the empirical total error
and $\epsilon = 0.05$.
Similarly,
we estimate
$\frac{\partial u_i\left(x;X(\mu_x^2,T) ,X(\mu_x^2,T_x)\right)}{\partial T}$
by setting $\mu_x$ and 
$T_x$ as the empirical bias
and empirical total error.
We observe that the assumption A5
is satisfied:
Specifically,
the estimated 
$\E\left[\frac{dg_x(\lambda^*)}{d\lambda} \cdot \frac{\partial u_i\left(X(\mu^2,T) ,X(\mu_x^2,T_x)\right)}{\partial T}\right]$ is -1.485, -31.569, -720.301
and 
$\E\left[\frac{dg_x(\lambda^*)}{d\lambda}\right] \cdot \E\left[\frac{\partial u_i\left(X(\mu^2,T) ,X(\mu_x^2,T_x)\right)}{\partial T}\right]$
is -0.142, -0.187, -70.008
for the California housing prices 
data, the red wine quality data,
and the white wine quality data, 
respectively.

\end{document}